\documentclass[journal]{IEEEtran}
\IEEEoverridecommandlockouts           % This command is only needed if you want to use the \thanks command

\usepackage{setspace,cite}
\usepackage[overload]{empheq} %align in "cases"
\usepackage{graphicx}
\graphicspath{{figures/}}
\ifCLASSOPTIONcompsoc
\usepackage[caption=false, font=normalsize, labelfont=sf, textfont=sf]{subfig}
\else
\usepackage[caption=false, font=footnotesize]{subfig}
\fi

\usepackage{epstopdf}
\usepackage{graphics} % for pdf, bitmapped graphics files
\usepackage{caption} %set align of caption
\usepackage{epsfig} % for postscript graphics files
\usepackage{color}
\usepackage{amsmath}%mathematic equations: aligned environment
\usepackage{amsthm}
\usepackage{amssymb} %\mathbb
\usepackage{verbatim}  %multiline comments
\usepackage{mathtools}%\DeclarePairedDelimiter
\usepackage{subfig}%subfloat
\usepackage{tabularx}
\usepackage{amsthm}
\usepackage[subnum]{cases}
\usepackage{verbatim}  %multiline comments
\usepackage{relsize}% resize sum sign

\newtheorem{example}{Example}%[section]
\newtheorem{theorem}{Theorem}%[section]
\newtheorem{lemma}{Lemma}%\newtheorem{lemma}[theorem]{Lemma}
\newtheorem{corollary}{Corollary}
\newtheorem{definition}{Definition}%[section]
\newtheorem{proposition}{Proposition}
\newtheorem{remark}{Remark}%[section]
\usepackage[ colorlinks = true,
linkcolor = blue,
urlcolor = blue,
citecolor = red,
anchorcolor = green,
]{hyperref}
\DeclarePairedDelimiter{\ceil}{\lceil}{\rceil}

\def\mcl{\mathcal}

\allowdisplaybreaks 

\begin{document}
	\title{Tunable Measures for Information Leakage and Applications to Privacy-Utility Tradeoffs }

	\author{Jiachun Liao, {\em Student Member, IEEE}, Oliver Kosut, {\em Member, IEEE},\\ Lalitha Sankar, {\em Senior Member, IEEE},  and Flavio du Pin Calmon, {\em Member, IEEE}% <-this % stops a space
		%\thanks{This work is supported in part by the National Science Foundation under grants CCF\--1350914 and CIF\--1422358. Prof. Vincent Y. F. Tan is supported in part by an Singapore Ministry of Education AcRF Tier 1 grant R-263-000-C54-114.}
		\thanks{This material is based upon work supported by the National Science Foundation under Grant Nos. CCF-1422358, CCF-1350914, and CIF-1422358. This work was presented in part at IEEE International Symposium on Information Theory and Information Theory Workshop in 2018.}
	}

%	\author{\IEEEauthorblockN{Jiachun Liao, Oliver Kosut, Lalitha Sankar}
%		\IEEEauthorblockA{School of Electrical, Computer and Energy Engineering,\\
%			Arizona State University\\
%			Email: \{jiachun.liao,lalithasankar,okosut\}@asu.edu}
%		\and
%		\IEEEauthorblockN{Flavio P. Calmon}%Flavio du Pin Calmon
%		\IEEEauthorblockA{School of Engineering and Applied Sciences\\
%			Harvard University\\
%			Email: fcalmon@g.harvard.edu\\
%		}
%		\thanks{This material is based upon work supported by the National Science Foundation under Grant No. CCF\--1350914 and CIF-1422358.}
%		%\thanks{This material is based upon work supported by the National Science Foundation under Grant No. CCF\--1350914 and CIF\--1422358.}
%	}

	\maketitle
	
	\begin{abstract}
		We introduce a tunable measure for information leakage called \textit{maximal $\alpha$-leakage}. This measure quantifies the maximal gain of an adversary in inferring any (potentially random) function of a dataset from a release of the data. The inferential capability of the adversary is, in turn, quantified by a class of adversarial loss functions that we introduce as $\alpha$-loss, $\alpha \in [1,\infty) \cup \{\infty\}$. The choice of $\alpha$ determines the specific adversarial action and ranges from refining a belief (about any function of the data) for $\alpha =1$ to guessing the most likely value for $\alpha = \infty$ while refining the $\alpha^{\text{th}}$ moment of the belief for $\alpha$ in between. Maximal $\alpha$-leakage then quantifies the adversarial gain under $\alpha$-loss over all possible functions of the data. In particular, for the extremal values of $\alpha=1$ and $\alpha=\infty$, maximal $\alpha$-leakage simplifies to mutual information and maximal leakage, respectively. For $\alpha\in(1,\infty)$ this measure is shown to be the Arimoto channel capacity of order $\alpha$. We show that maximal $\alpha$-leakage satisfies data processing inequalities and a sub-additivity property thereby allowing for a weak composition result. Building upon these properties, we use maximal $\alpha$-leakage as the privacy measure and study the problem of data publishing with privacy guarantees, wherein the utility of the released data is ensured via a \emph{hard distortion} constraint. Unlike average distortion, hard distortion provides a deterministic guarantee of fidelity. We show that under a hard distortion constraint, for $\alpha>1$ the optimal mechanism is independent of $\alpha$, and therefore, the resulting optimal tradeoff is the same for all values of $\alpha>1$. Finally, the tunability of maximal $\alpha$-leakage as a privacy measure is also illustrated for binary data with average Hamming distortion as the utility measure.		
	\end{abstract}
	
	
	\IEEEpeerreviewmaketitle
		
	\section{Introduction and Overview}
	The measure and control of private information leakage is a recognized objective in communications, information theory, and computer science. Modern cryptography \cite{DirectionCryptograpy_Diffie76,PublicKeyCrypto_Elgamal85,VisualCryptograpy_Prisco14}, for example, aims at designing and analyzing security systems that are believed to be impervious to computationally bounded adversaries. Alternatively, information-theoretic security studies settings where an asymmetry of information between an adversary and the legitimate parties (e.g., the wiretap channel \cite{Gaussian-Wiretap-78,SecrecyCapacity_Bnargav16,BraodcastChannelConfidential16}) can be exploited to guarantee that no private information is leaked regardless of computational assumptions. An adversary that \emph{only} observes  the output of a (computationally) secure cipher or cannot overcome the information asymmetry in a wiretap-like setting does not, for all practical purposes, pose a privacy risk.
	
	However, modern applications such as online data sharing, social networks, cloud-based services, and mobile computing have significantly increased the number of ways in which private information can leak. Services that require a user to disclose data in order to receive utility inevitably incur a privacy risk through unwanted inferences. For example, \textit{sensitive information} such as political preference, medical conditions, and identity can be reliably estimated from movie ratings \cite{DeanonymizationLargeSpareDataset_Narayanan08}, online shopping patterns, \cite{CloudService_Privacy-Ristenpart09}, and via deanonymization and tracking of interactions in social network data \cite{RumorNetwor_Shah11,RumorIdentification_liang15}, respectively. 
	Moreover, practical implementations of cryptographic schemes are susceptible to so-called ``side-channel attacks,'' where sensitive information leaks through unexpected channels. For example, a malicious application may get timing characteristics \cite{Chip-leakage_Ghassami15,NetworkonChip-Biswas18}. In these examples, an adversary that observes information leaked through a side-channel can more reliably infer private data, such as a key or a plaintext.
	
	Several (often overlapping) definitions of privacy/information leakage have been proposed over the past decade. The most widely adopted measure is differential privacy (DP) \cite{Dwork_DP,Dwork_DP_Survey}, which was introduced within the context of querying databases. DP seeks to ensure that changes in the database entries do not significantly influence the value of a query. A variety of information-theoretic measures have also been proposed as leakage measures. Foremost among them is mutual information (MI): its use as a privacy measure in  \cite{Aggarwal2006,Rebollo-Monedero2010,Motwani1,PrivacyAgainstStatistic_Calmon12,sankar_utility-privacy_2013,Sankar2011,maximalcorrelation_Asoodeh2015,Relation_MI_DP_Wang16,HypothesisTest&MutualInf_Liao2017,ITPrivacySmartMetering_Khisti18} is inspired by the common appearance of MI as an operationally-meaningful quantity throughout the literature on communication systems.
	In a similar vein, divergence-based quantities such as total variation distance between the prior and posterior distributions  \cite{TVDprivacy_Rassouli&Gunduz18} have also been proposed as leakage measures. Information-theoretic measures have been studied in the DP community via R{\'e}nyi differential privacy which is based on R{\'e}nyi divergence \cite{R-DP_Mironov17} that allow relaxing the original definition of DP in order to enable better utility guarantees. However, the gamut of information-theoretic leakage measures proposed to address the privacy problem do not \textit{yet} have clear operational meanings or adversarial models in their definitions. 

    More recently, information-theoretic formulations have been introduced to capture privacy against a ``guessing'' adversary. Here, privacy is measured in terms of an adversary's gain in guessing the private information after observing disclosed data. For example, Asoodeh et al. use the probability of correctly guessing to measure privacy \cite{privacyGuessing_asoodeh2017}; and Issa et al. introduce maximal leakage (MaxL), which quantifies the maximal logarithmic gain in the probability of correctly guessing any arbitrary function of the original data from released data \cite{OperationalLeak_issa2018}. A related line of work includes \cite{Guessing_Massey1994,Guesswork_Malone04,Guesswork_Christiansen13}, where security is quantified in terms of the expected number of guesses (or moments thereof) required by an adversary to correctly identify a quantity of interest (e.g., a password or a transmitted codeword).

	This work builds upon the abovementioned efforts to operationally motivate measures and presents a larger class of meaningful information-theoretic measures that can be operationally motivated in the privacy setting. To this end, we introduce a tunable loss function, namely $\alpha$-loss ($1\leq \alpha\leq \infty$), to capture adversarial actions. In particular, for $\alpha=1$ and $\alpha=\infty$ the loss function simplifies to the logarithmic loss (log-loss) \cite{Merhav1998,SurrogateLoss&fDivergence_Nguyen09,Courtade2011} and the probability of error\footnote{Note that the probability of error for a maximum likelihood estimator is exactly the 0-1 loss \cite{Classification_Bartlett06, SurrogateLoss&fDivergence_Nguyen09}.}, respectively.
	The choice of the loss function captures the \textit{inferential action} of an adversary. Specifically, the adversarial action, henceforth referred to as inference, involves refining a posterior belief of one or more sensitive features. 
	Adversarial gain of a computationally unbounded adversary is then simply the decrease in (inferential) loss on average as a result of a data release.

	 We use the $\alpha$-loss function to derive two new privacy measures called \textit{$\alpha$-leakage} and \textit{maximal $\alpha$-leakage}. Specifically, $\alpha$-leakage quantifies an adversary's gain in inferring a \textit{specific} private attribute in the dataset; in contrast, maximal $\alpha$-leakage quantifies an adversary's gain in inferring \textit{any arbitrary} attribute of the dataset. In particular, maximal $\alpha$-leakage includes MI and MaxL as special cases for $\alpha=1$ and $\alpha=\infty$, respectively. This approach allows us to show that MaxL can be interpreted in terms of an adversary seeking to minimize the 0-1 loss function \cite{Classification_Bartlett06, SurrogateLoss&fDivergence_Nguyen09} ($\alpha=\infty$), i.e., the adversary makes a hard decision via a maximum likelihood estimator. On the other hand, we show that when MI is used as a leakage measure ($\alpha=1$), the underlying loss function is the log-loss, that models a (soft decision) belief-refining adversary.
	 In addition to what the adversary observes (e.g., released census dataset or information via a side-channel), the adversary may also have access to other correlated side-information (e.g., voter record database or individual personal information in side-channel attacks); generalizing \textit{$\alpha$-leakage} and \textit{maximal $\alpha$-leakage} to model such side-information is indeed possible as recently shown by the authors in \cite{RobustnessMaxAlpLK_Liao19}; however, this generalization is beyond the scope of this paper.

	Our proposed measures can be applied to the aforementioned privacy and side-channel settings. In most non-trivial settings of data publishing, there is a fundamental privacy-utility tradeoff (PUT): on the one hand, releasing data  ``as is'' can lead to unwanted inferences of private information. On the other hand, perturbing or limiting the released data reduces its quality. We quantify PUTs for two types of data models: one in which the entire dataset is sensitive (as illustrated in Fig.~\ref{fig:Datamodel_1}) and the other in which only a subset of the dataset is sensitive (as illustrated in Fig.~\ref{fig:Datamodel_2}). 	Throughout this paper, we use $X$ to denote the original data that will be \textit{released} as $Y$ via a randomized mapping; $X$ may be entirely sensitive as in Fig.~\ref{fig:Datamodel_1}, or it may be separate from the sensitive features $S$ as in Fig.~\ref{fig:Datamodel_2}. The variable $U$ represents a specific sensitive feature of the dataset that the adversary is interested in learning. Examples of datasets wherein the entire data is sensitive include data collected by smart devices such as smartphone sensors, movie recommendation systems, where it is hard to know \textit{a priori} which aspect of the data ought to be identified as sensitive. In contrast, examples of datasets with clearly defined sensitive features include census and other datasets that explicitly include personally identifiable information.

	The exact nature of the PUT depends on exactly how both privacy and utility are measured. Towards an understanding of our new privacy measures, we consider PUTs in which (maximal) $\alpha$-leakage is the privacy measure, and we study several options for utility measure.
	In general, a meaningful utility measure (between the original and released data) should require the released data to provide either (i) average-case guarantees on fidelity \cite{LocalPrivacy_Duchi13,StaircaseMechanismDP_Geng_Kairouz15,PrivacyAgainstStatistic_Calmon12,privacyGuessing_asoodeh2017,TVDprivacy_Rassouli&Gunduz18}; or (ii) worst-case guarantees on fidelity. Indeed, requirement (i) lends itself to modeling with a large class of expected value constraints including average distortion constraints and is now well studied in information-theoretical privacy via a variety of measures such as Hamming distortion, square error and Kullback--Leibler divergence \cite{DP-Hamming_Kousha18, Privacy&MMSE_Asoodeh2016, StaircaseMechanismDP_Geng_Kairouz15, HypothesisMaxL_Liao2017, HypothesisTest&MutualInf_Liao2017}. We note that average distortion constraints are also well studied in rate-distortion theory.
	To capture utility requirement (ii), we introduce a \emph{hard distortion} measure which constrains the privacy mechanism so that the distortion between original and released datasets is bounded with probability $1$. Such an approach has also been studied in rate-distortion theory as a potential distortion measure (see, for example, \cite{MaxDistortion_Tuncel02} for the use of per symbol distortion constraints). In addition, compared to average-case distortion constraints \cite{DP-Hamming_Kousha18, Privacy&MMSE_Asoodeh2016, StaircaseMechanismDP_Geng_Kairouz15, HypothesisMaxL_Liao2017, HypothesisTest&MutualInf_Liao2017}, a hard distortion measure is quite stringent but allows the data curator to make specific, deterministic guarantees on the fidelity of the released dataset relative to the original. Such a deterministic guarantee can lead to more accurate statistical estimators, e.g., the empirical distribution estimation for publicly released datasets such as the census.

\begin{figure*}
	\centering
	\subfloat[The privacy protection for entirely sensitive datasets. ]{%
		\includegraphics[width=3.5  in]{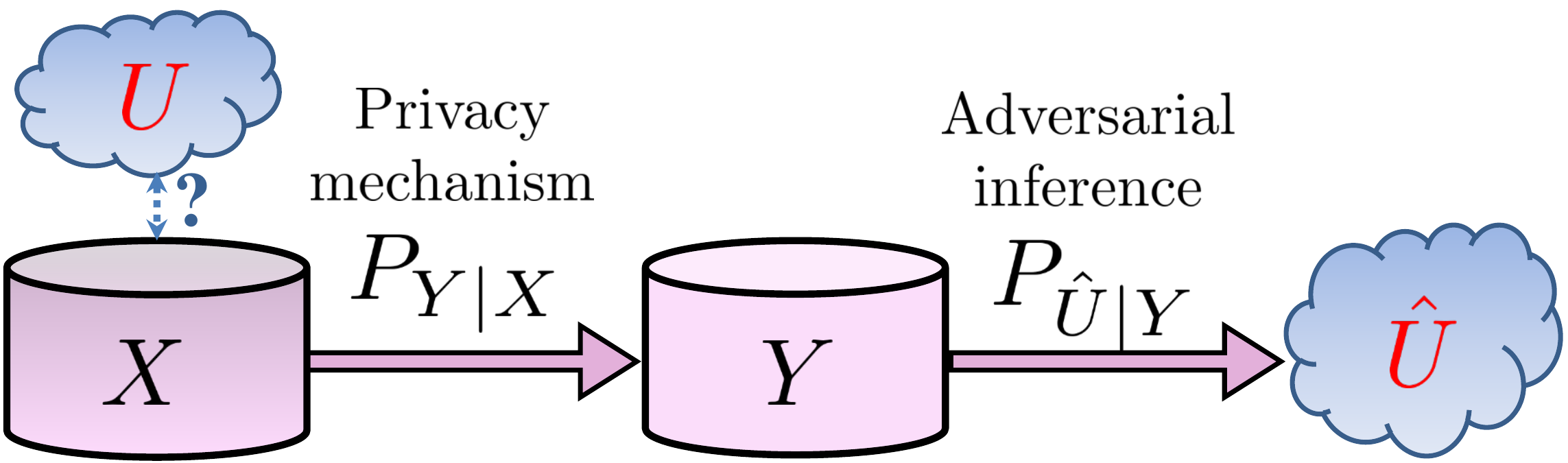}\label{fig:Datamodel_1}  }			                   
	\subfloat[The privacy protection for datasets with non-sensitive and sensitive data. ]{
		\includegraphics[width=3.5  in]{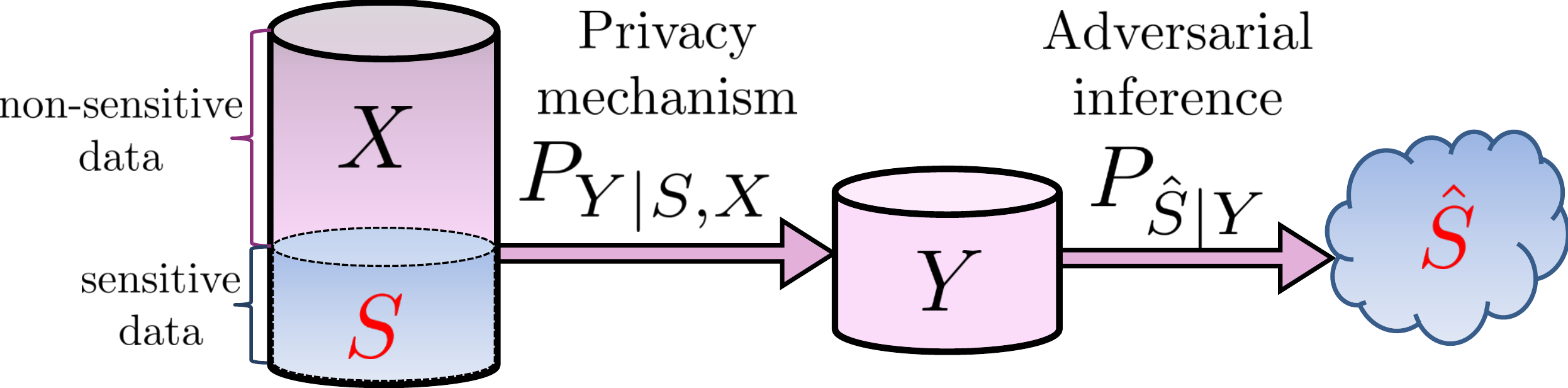}\label{fig:Datamodel_2}  }	                	
	\caption{Two privacy-guaranteed data publishing scenarios: (i) the left figure shows the privacy protection for entirely sensitive datasets, where $X$ and $Y$ represent the original and released data. An adversary intends to infer a function $U$ of $X$ from $Y$, and $\hat{U}$ is the adversary's estimation of $U$. Generally, the function $U$ is unknown to the data curator/provider; (ii) the right figure shows the privacy protection for datasets consisting of non-sensitive and sensitive data, where $X$ and $S$ represent the non-sensitive and sensitive data in original dataset, respectively, and $Y$ is the released version of $X$. The adversary intends to infer $S$ from $Y$, and $\hat{S}$ is the adversary's estimation of $S$. } 
\end{figure*}

\subsection{Contributions and Organization}
The main contributions of this paper include:
	\begin{itemize}
		\item We introduce a tunable loss function, namely $\alpha$-loss ($1\leq \alpha\leq \infty$), which captures log-loss and $0$-$1$ loss, respectively, for extremal values of $\alpha =1$ and $\infty$, respectively (Sec. \ref{Subsec:alpha-loss}).
		\item Based on $\alpha$-loss, we define two operational measures of information leakage: $\alpha$-leakage and maximal $\alpha$-leakage, and show that: (i) $\alpha$-leakage equals to Arimoto mutual information of order $\alpha$ \cite{alphaMI_verdu,AlphaMI_Arimoto1975}; and (ii) maximal $\alpha$-leakage equals to MI for $\alpha=1$ and Arimoto channel capacity \cite{AlphaMI_Arimoto1975} of order $\alpha$ for $\alpha>1$. Note that maximal $\alpha$-leakage captures MI and MaxL at the extremal values of $\alpha$ (Sec. \ref{Subsec:alpha leakage measures}). The proofs of these results rely on the fact that maximizing either the Arimoto MI or the Sibson MI \cite{alphaMI_Sibson1969} over the input distribution yields the same quantity, the Arimoto channel capacity.
		\item Inspired by the fact that maximal $\alpha$-leakage equals to the Arimoto channel capacity, we introduce a broader class of information-leakage measures based on $f$-divergences, which capture maximal $\alpha$-leakage as a special case (Sec. \ref{Subsec:f-divergence-based leakages});
		\item We prove that maximal $\alpha$-leakage satisfies several useful properties, including: (i) quasi-convexity, (ii) data-processing inequalities: post-processing inequality and linkage inequality, (iii) sub-additivity (iv) additivity for memoryless mappings (Sec. \ref{Sec:Properties}).
		\item In the context of privacy-guaranteed data publishing subject to a hard distortion utility constraint on data, we solve the resulting PUT problems exactly for maximal $\alpha$-leakage as well as its $f$-divergence-based variants (Sec. \ref{Subsec:PUT_for-HDvsMaxAlphaLK}). For $\alpha$-leakage, which restricts leakage about specific sensitive data as shown in Fig. \ref{fig:Datamodel_2}, we provide an inner bound of the optimal PUT (Sec. \ref{Subsec:PUT_for-HDvsAlphaLK}). In Sec. \ref{Sec:Examples}, we illustrate these results via two examples.
\end{itemize}

	\section{Preliminaries}\label{Sec:Preliminaries}
	We use capital letters to represent \textit{discrete} random variables, and the corresponding capital calligraphic and lower-case letters represent their \textit{finite} supports and the elements of the supports, respectively. For example, for a random variable $X$, its support is $\mcl X$ with any possible realization $x\in \mcl X$. In addition, we use $\log$ to represent the natural logarithm, and $[a,b]$ to indicate the set of integers from $a$ to $b$. We use $|\cdot|$ to indicate the cardinality of a set, e.g., $|\mcl X|$, and $\|\cdot\|_{p}$ to represent the $p$-norm of a vector, e.g., for $\alpha\geq 1$, $\|P_X\|_{\alpha}\triangleq (\sum_{x\in \mcl X}P_X(x)^{\alpha})^{\frac{1}{\alpha}}$.

	We begin by reviewing R{\'e}nyi entropy and divergence \cite{measures_renyi1961,RenyiDivergence_Erven}.
	\begin{definition}
		Given a distribution $P_X$, the R{\'e}nyi entropy of order $\alpha\in (0,1)\cup(1,\infty)$ is defined as
		\begin{align}
			\label{eq:renyi_entropy}
			H_{\alpha}(P_X)&=\frac{1}{1-\alpha}\log\sum_{x\in \mcl X}P_{X}(x)^{\alpha},\\
			&=\frac{\alpha}{1-\alpha}\log\|P_X\|_{\alpha}, \quad  (\alpha\geq 1). 
		\end{align}
		Let $Q_X$ be a distribution over the support of $P_X$. The R{\'e}nyi divergence (between $P_X$ and $Q_X$) of order $\alpha\in (0,1)\cup(1,\infty)$ is defined as 
		\begin{align}
			\label{eq:renyi_divergence}
			D_{\alpha}(P_X\|Q_X)=\frac{1}{\alpha-1}
			\log\left(\sum\limits_{x\in \mcl X}\frac{P_X(x)^{\alpha}}{Q_X(x)^{\alpha-1}}\right).
		\end{align}
		Both of the two quantities are defined by their continuous extensions for $\alpha=1$ and $\infty$. 
		Specifically, for $\alpha=\infty$, the two quantities are given by
		\begin{align}
			H_{\infty}(P_X) =\min_x\log \frac{1}{P_X(x)} ,
		\end{align}
		which is called min-entropy, and
		\begin{align}
			D_{\infty}(P_X\|Q_X) =\log \max_x \frac{P_X(x)}{Q_X(x)}.
		\end{align}
		For $\alpha=1$, the R{\'e}nyi entropy and divergence reduce to Shannon entropy and Kullback-Leibler divergence, respectively \cite{alphaMI_verdu}. 
	\end{definition}   
	The $\alpha$-leakage and maximal $\alpha$-leakage measures can be expressed in terms of Sibson MI \cite{alphaMI_Sibson1969} and Arimoto MI \cite{AlphaMI_Arimoto1975}. These quantities generalize the usual notion of MI. We review these definitions next.
	\begin{definition}
		Let discrete random variables $(X,Y)\sim P_{X,Y}$ with $P_X$ and $P_{Y|X}$ as the marginal and conditional distributions, respectively, and $Q_Y$ be an arbitrary distribution over the finite support $\mcl Y$. The Sibson mutual information of order $\alpha\in(0,1)\cup(1,\infty)$ is defined as
		\begin{align}
			%\label{eq:Def_SibsionMI}
			\hspace{-5pt} I_\alpha^{\text{S}}(X;Y)\triangleq &\inf_{Q_Y}\,D_\alpha(P_{X,Y}\|P_X\times Q_Y)
			\label{eq:Sibson_MI}\\
			= &\frac{\alpha}{\alpha-1}\log \sum\limits_{y\in \mcl Y}\left(\sum\limits_{x\in \mcl X}P_X(x)P_{Y|X}(y|x)^{\alpha}\right)^{\frac{1}{\alpha}}.
		\end{align}
		The Arimoto mutual information of order $\alpha\in(0,1)\cup(1,\infty)$ is defined as
		\begin{IEEEeqnarray}{l l}
			\label{eq:Def_ArimotoMI}
			I_\alpha^{\text{A}}(X;Y)&\triangleq H_{\alpha}(X)-H_{\alpha}^{\text{A}}(X|Y)\\
			&=\frac{\alpha}{\alpha-1}\log\frac{\sum\limits_{y\in \mcl Y}\left(\sum\limits_{x\in \mcl X}P_{X,Y}(x,y)^{\alpha}\right)^{\frac{1}{\alpha}}}{\left(\sum\limits_{x\in \mcl X}P_X(x)^{\alpha}\right)^{\frac{1}{\alpha}}},\\
			\label{eq:Arimoto_MI}
			&=\frac{\alpha}{\alpha-1}\log\frac{\sum\limits_{y\in \mcl Y }\|P_{X,Y}(\cdot,y)\|_{\alpha}}{\|P_X\|_{\alpha}},\quad (\alpha\geq 1)
		\end{IEEEeqnarray}
		where $H_{\alpha}^{\text{A}}(X|Y)$ is Arimoto conditional entropy of $X$ given $Y$ defined as
		\begin{align}
			\label{eq:Def_ArimotoConditionalEntropy}
			H_{\alpha}^{\text{A}}(X|Y)=\frac{\alpha}{1-\alpha}\log\sum\limits_{y\in \mcl Y}\left(\sum\limits_{x\in \mcl X}P_{X,Y}(x,y)^{\alpha}\right)^{\frac{1}{\alpha}}.
%			&=\frac{\alpha}{1-\alpha}\log\sum\limits_{y}\|P_{X,Y}(Xy)\|_{\alpha}\text{ for } \alpha\geq 1.
		\end{align}
		All of these quantities are defined by their continuous extension for $\alpha=1$ or $\infty$.
	\end{definition}
	Note that for $\alpha=1$, both Sibson and Arimoto MIs reduce to Shannon's MI; however, for $\alpha=\infty$, the Sibson MI is 
	\begin{align}
		I_{\infty}^{\text{S}}(X;Y)=\log \sum\limits_{y}\max_x P_{Y|X}(y|x),
	\end{align}
	and the Arimoto MI is given by 
	\begin{align}
			I_{\infty}^{\text{A}}(X;Y)=\log\frac{\sum\limits_{y}\max\limits_{x}P_{X,Y}(x,y)}{\max\limits_{x}P_X(x)}.
	\end{align}
	The two measures of information generalize Shannon's MI and have a number of interesting and useful properties in various problems \cite{alphaMI_Sibson1969,AlphaMI_Arimoto1975,alphaMI_verdu,ConvexityAlphaMI_Ho}.

	\section{Tunable Loss Function and Information Leakage Measures}\label{Sec:Information Leakage Measures}	
    Information leakage of a data release can be viewed as an increase in adversarial inference as a result of the data release. This inference performance can be precisely characterized by a loss function that an adversary minimizes. In this section, we introduce a tunable loss function, namely \textit{$\alpha$-loss} for $\alpha\in [1,\infty]$, to captures a computationally unbounded adversary's inference in refining a posterior belief of one or more sensitive features from a data release, and introduce two tunable measures, called \textit{$\alpha$-leakage} and \textit{maximal $\alpha$-leakage}, respectively, to measure the corresponding information leakages due to the data release. 
    
	\subsection{$\alpha$-Loss Function}\label{Subsec:alpha-loss}
	 For a Markov chain $X-Y-\hat{X}$, let $\hat{X}$ be an estimator of $X$ and $P_{\hat{X}|Y}$ be a strategy for estimating $X$ from $Y$. We denote the probability of correctly estimating $X=x$ given $Y=y$ as $P_{\hat{X}|Y}(x|y)$. The estimation strategy $P_{\hat{X}|Y}$ is selected in order to minimize an \textit{expected loss} measure. Denoting the loss function by $\ell(x,y,P_{\hat{X}|Y})$, the expected loss is given by $\mathbb{E}\big[\ell\big(X,Y,P_{\hat{X}|Y}\big)\big]$. 
	 %synonymously defined as \textit{risk} or \textit{generalization error} \cite{Estimation&Privacy_Asoodeh}    
    \begin{figure}
    	\centering
    	\includegraphics[width=3.5 in]{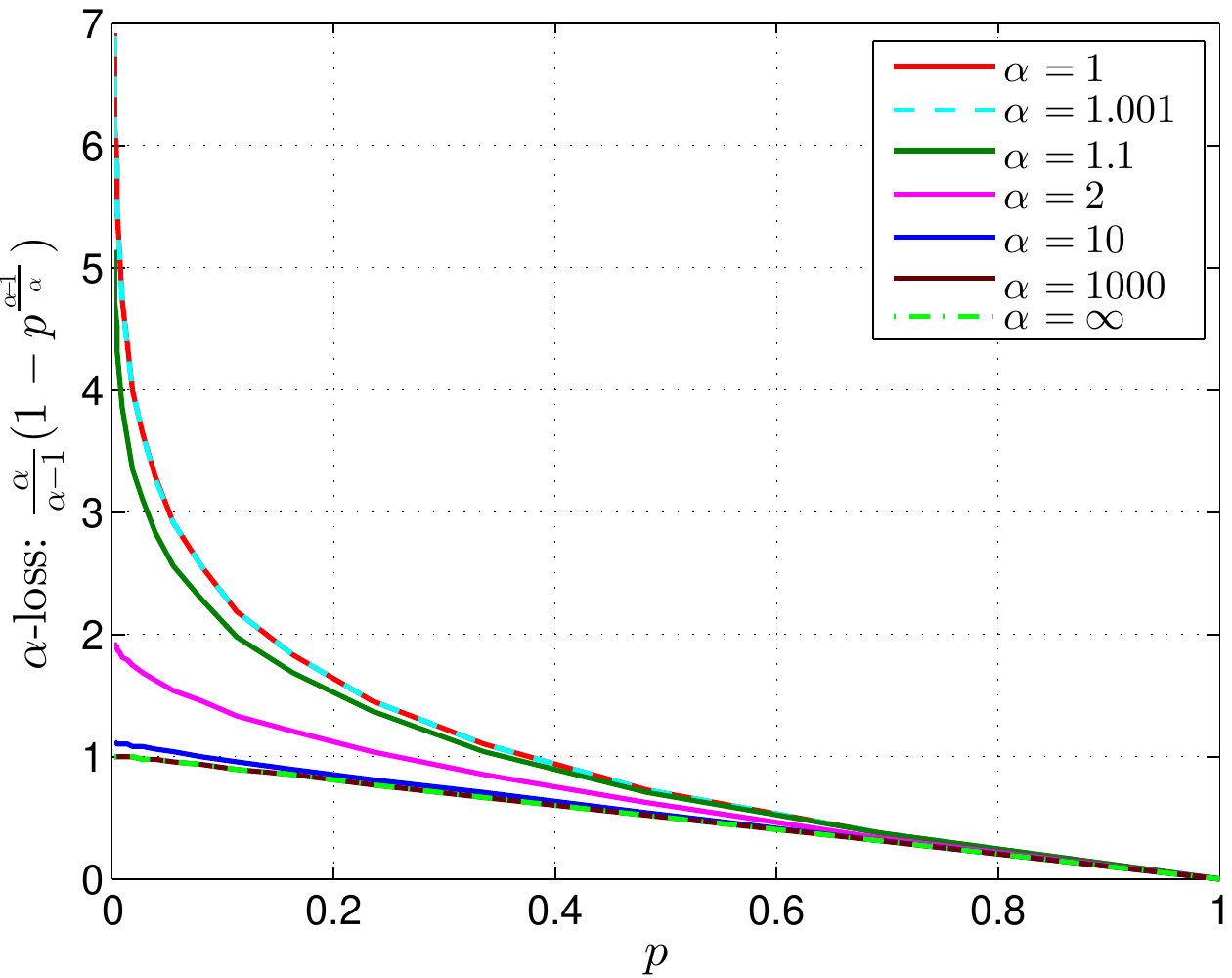}
    	\caption{The plot of $\alpha$-loss as a function of $p$. Note that the $p\in[0.001,1]$ represents the probability of correctly guessing, i.e., $p=P_{\hat{X}|Y}(x|y)$ with an observation $Y=y$ and $p=P_{\hat{X}}(x)$ without any observation. }
    	\label{fig:alpha-loss}
    \end{figure}
    
    One formulation of the loss function is the probability of \textit{incorrectly} guessing given by %
    \begin{equation}\label{pe_loss}
    \ell_{0-1}(x,y,P_{\hat{X}|Y})=1-P_{\hat{X}|Y}(x|y),
    \end{equation}
    such that the expected loss $\mathbb{E}\big[\ell_{0-1}\big(X,Y,P_{\hat{X}|Y}\big)\big]$ is the expected probability of error.
    Here, the optimal strategy $P_{\hat{X}|Y}^{\star}$ is the standard maximal posterior (MAP) estimator given by 
    \begin{align}
    P_{\hat{X}|Y}^{\star}(x|y) = \begin{cases} 
    1,& x=\arg\max\limits_{x\in\mcl X} P_{X|Y}(x|y)\\
    0,&\mbox{otherwise}
    \end{cases},
    \end{align}
    which makes the loss $\ell_{0-1}(x,y,P_{\hat{X}|Y}^{\star})$ be either $0$ or $1$, and therefore, called \textit{$0$-$1$ loss} in the literature \cite{Classification_Bartlett06, SurrogateLoss&fDivergence_Nguyen09}. The corresponding expected loss $\mathbb{E}\big[\ell_{0-1}\big(X,Y,P_{\hat{X}|Y}^{\star}\big)\big]$ is the minimal expected probability of error.
    
    To measure the uncertainty for the strategy $P_{\hat{X}|Y}$, the log-loss  (used, for example, in \cite{SurrogateLoss&fDivergence_Nguyen09,Merhav1998,Courtade2011,Courtade2014}) is given by
    \begin{equation}\label{log_loss}
    \ell_{\text{log}}(x,y,P_{\hat{X}|Y})=\log\frac{1}{P_{\hat{X}|Y}(x|y)}.
    \end{equation}
    The expected loss in this case is the conditional cross-entropy, given by
    \begin{IEEEeqnarray}{l l}
    	&\mathbb{E}\left[\ell_{\text{log}}(X,Y,P_{\hat{X}|Y})\right]\nonumber\\
   	=& \sum_{x,y} P_{X,Y}(x,y)\log\frac{1}{P_{\hat{X}|Y}(x|y)},\\
    =& H(X|Y)+\sum_y P_Y(y)D(P_{X|Y=y}\|P_{\hat{X}|Y=y}). \label{eq:expected_log_loss}
    \end{IEEEeqnarray} 
    Therefore, the optimal strategy is the true posterior distribution of $X$ given $Y$, i.e., $P_{\hat{X}|Y}^{\star}=P_{X|Y}$, which makes the expected loss in \eqref{eq:expected_log_loss} become the conditional entropy $H(X|Y)$. That is, the minimal expected log-loss is the true conditional entropy.

    \begin{figure}
    	\centering
    	\includegraphics[width=3.5 in]{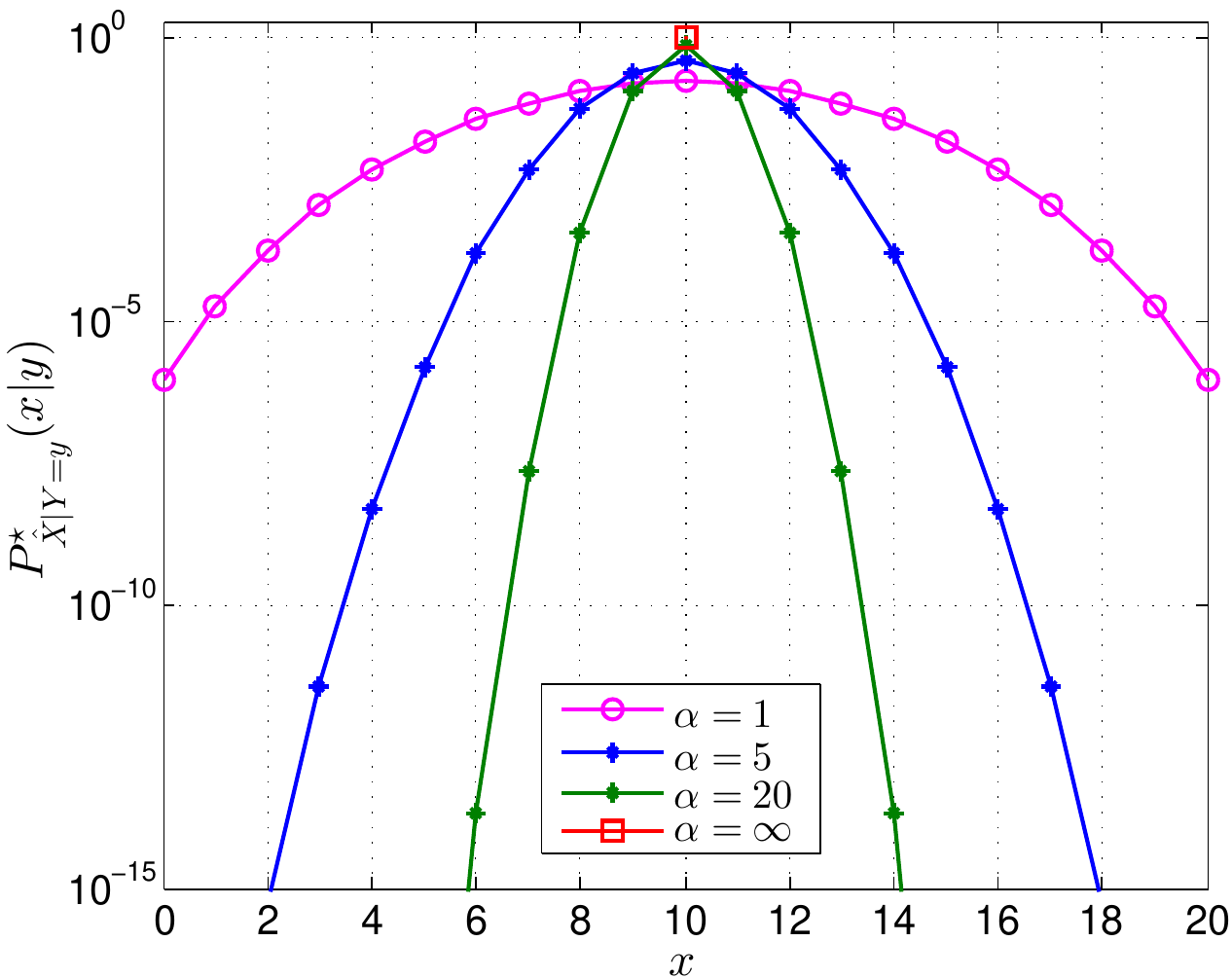}
    	\caption{The optimal strategy in \eqref{eq:alphaLoss-OptStrategy} for different $\alpha$. Note that the magenta circles represent the true conditional probability $P_{X|Y=y}$, which is a binomial distribution with parameters $(n,p)=(20,0.5)$. 
    	 }
    	\label{fig:alpha-loss-mech}
    \end{figure}

    Note that both the $0$-$1$ loss and log-loss functions are decreasing in the probability of correctly estimation $P_{\hat{X}|Y}(x|y)$. Specifically, for $P_{\hat{X}|Y}(x|y)=1$, both the values of $0$-$1$ loss and $\alpha$-loss are $0$, and for $P_{\hat{X}|Y}(x|y)=0$, the values of $0$-$1$ loss and log-loss become $1$ and $\infty$, respectively. To allow a continuous quantification of the loss for $P_{\hat{X}|Y}(x|y)=0$ from $1$ to $\infty$, we formally define a tunable loss function, namely \textit{$\alpha$-loss}, as follows.        
    \begin{definition}[{$\alpha$-loss}]\label{Def:alpha-loss}
    	Let random variables $X$, $Y$ and $\hat{X}$ form a Markov chain $X-Y-\hat{X}$, where $\hat{X}$ is an estimator of $X$. The $\alpha$-loss of the strategy $P_{\hat{X}|Y}$ for estimating $X$ from $Y$ is 
    	 \begin{equation}\label{poly_loss}
    	\ell_{\alpha}(x,y,P_{\hat{X}|Y})=\frac{\alpha}{\alpha-1} \big(1-P_{\hat{X}|Y}(x|y)^{\frac{\alpha-1}{\alpha}}\big),
    	\end{equation}
    	where $\alpha\in(1,\infty)$. It is defined by its continuous extension for $\alpha=1$ and $\alpha=\infty$, respectively, and is given by
    	\begin{align}
    	\label{eq:alpha-loss-to-log-loss}
    	 \hspace{-10pt}   \ell_{1}(x,y,P_{\hat{X}|Y})&=\hspace{-2pt}\lim_{\alpha\to 1}\ell_{\alpha}(x,y,P_{\hat{X}|Y})\hspace{-2pt}= \log \frac{1}{P_{\hat{X}|Y}(x|y)},\\
    	\label{eq:alpha-loss-to-01-loss}
    	\hspace{-10pt}	\ell_{\infty}\hspace{-1pt}(x,y,P_{\hat{X}|Y})\hspace{-2pt}&=\hspace{-2pt}\lim_{\alpha\to \infty}\ell_{\alpha}(x,y,P_{\hat{X}|Y})\hspace{-2pt}= \hspace{-2pt}1\hspace{-2pt}-\hspace{-2pt}P_{\hat{X}|Y}(x|y).
    	\end{align}
    \end{definition}
     Note that for $\alpha=1$, the expression in \eqref{eq:alpha-loss-to-log-loss} follows directly from the L'H{\^o}pital's rule and $\alpha$-loss becomes the log-loss in \eqref{log_loss}; and for $\alpha=\infty$, the loss in \eqref{eq:alpha-loss-to-01-loss} is exactly the probability of error in \eqref{pe_loss}, which becomes $0$-$1$ loss for MAP estimators. Fig.~\ref{fig:alpha-loss} plots the $\alpha$-loss function in \eqref{poly_loss} for different values of $\alpha$. From Fig.~\ref{fig:alpha-loss}, we observe that $\alpha$-loss function is decreasing and convex in  the probability of correctly guessing.

     \begin{lemma}\label{lem:Minimal-expectedAlphaLoss}
     	For $1\leq \alpha\leq \infty$, the minimal expected $\alpha$-loss is given by
     	\begin{align}
     	&\min_{P_{\hat{X}|Y}} \mathbb{E}\left[\ell_{\alpha}(X,Y,P_{\hat{X}|Y})\right]\nonumber\\
     	=&\begin{cases}
     	\frac{\alpha}{\alpha-1}\left(1-\exp\left(\frac{1-\alpha}{\alpha}H_{\alpha}^{\text{A}}(X|Y)\right)\right),& \alpha>1\\
     	H(X|Y), & \alpha=1
     	\end{cases},\label{eq:expected-alpha-loss}
     	\end{align}
     	with the optimal estimation strategy given by \footnote{Note that if there are more than one realization sharing the same maximal posterior belief, for $\alpha=\infty$ the optimal strategy in \eqref{eq:alphaLoss-OptStrategy} will output these most likely values with the same probability.}
     	\begin{align}\label{eq:alphaLoss-OptStrategy}
     	P^{\star}_{\hat{X}|Y}(x|y)=\frac{P_{\hat{X}|Y}(x|y)^{\alpha}}{\sum\limits_{x\in\mcl X}P_{\hat{X}|Y}(x|y)^{\alpha}}.
     	\end{align}
     \end{lemma}
 A detailed proof is in Appendix \ref{Proof:lem:Minimal-expectedAlphaLoss}.
 Note that in \eqref{eq:expected-alpha-loss}, $H_{\alpha}^{\text{A}}(X|Y)$ is Arimoto conditional entropy of $X$ given $Y$ in \eqref{eq:Def_ArimotoConditionalEntropy}. For $\alpha=\infty$, the expression of $ H^{\text{A}}_{\infty}(X|Y)$ is
 \begin{align}
 H^{\text{A}}_{\infty}(X|Y)=\log\sum_{y}P_Y(y)\max_x P_{X|Y}(x|y),
 \end{align} such that $\exp\left(H^{\text{A}}_{\infty}(X|Y)\right)$ is the maximal expected probability of correctly guessing $X$ from $Y$. Therefore, for $\alpha=\infty$, the minimal expected $\alpha$-loss is the minimal expected probability of error. In addition, the optimal estimation strategy in \eqref{eq:alphaLoss-OptStrategy} becomes the true posterior distribution of $X$ for $\alpha=1$ and the MAP estimator for $\alpha=\infty$, respectively.
\begin{example}
 	Let the conditional probability distribution of $X$ given $Y=y$ be a binomial distribution with parameters $(n,p)=(20,0.5)$, i.e., $P_{X|Y}(x|y)={20\choose x}0.5^x0.5^{20-x}$ for $x\in[0,20]$. Fig.~\ref{fig:alpha-loss-mech} shows the optimal strategies in \eqref{eq:alphaLoss-OptStrategy} for different values of $\alpha$. We observe from Fig.~\ref{fig:alpha-loss-mech} that as $\alpha$ grows from $1$ to $\infty$, the optimal strategy gradually eliminates the less likely values of $X$ (given $y$) and transforms from the true posterior distribution to the MAP estimator. 	
 \end{example}

   	\subsection{$\alpha$-Leakage and Maximal $\alpha$-Leakage}\label{Subsec:alpha leakage measures} 
	Let $X$ and $Y$ represent the original data and released data, respectively, and let $U$ represent an arbitrary (potentially random) function of $X$ that the observer (a curious or malicious user of the released data $Y$) is interested in learning. 
	In \cite{OperationalLeak_issa2018}, Issa \textit{et al.} % Kamath and Wagner 
	introduced MaxL to quantify the maximal gain in an adversary's ability of guessing $U$ after observing $Y$. We review the definition below.
	\begin{definition}[{\cite[Def. 1]{OperationalLeak_issa2018}}]\label{Def:MaximalLeakage}
		Given a joint distribution $P_{X,Y}$ on finite alphabets, the \textit{maximal leakage} from $X$ to $Y$ is
		\begin{equation}\label{ml_op_def}
			\mcl L_{\text{MaxL}}(X\to Y)\triangleq\sup_{U- X- Y} \log \frac{\max\limits_{P_{\hat{U}|Y}} \mathbb{E}\left[P_{\hat{U}|Y}(U|Y)\right]}{\max\limits_{u}  P_{U}(u)},
		\end{equation}
		where $\hat{U}$ represents an estimator taking values from the same arbitrary finite support as $U$.
	\end{definition} 
		Note that the numerator of the logarithmic term in \eqref{ml_op_def} is the maximal expected probability of correctly guessing $U$ with $Y$ given by
		\begin{align}\label{eq:MaxL-0-1Loss}
			\max\limits_{P_{\hat{U}|Y}} \mathbb{E}\left[P_{\hat{U}|Y}(U|Y)\right]=\max\limits_{u}\sum\limits_{y}P_{Y}(y)P_{U|Y}(u|y),
		\end{align}
		which is exactly the complement of the minimal expected $0$-$1$ loss in guessing $U$ with $Y$. Similarly, the denominator is the complement of the minimal expected $0$-$1$ loss in guessing $U$ without $Y$. Therefore, MaxL is a leakage measure related to $0$-$1$ loss in \eqref{pe_loss}. 
		
		In addition, in Def. \ref{Def:MaximalLeakage}, $U$ represents any (possibly random) function of $X$. The numerator represents the maximal probability of correctly guessing $U$ based on $Y$, while the denominator represents the maximal probability of correctly guessing $U$ \emph{without} knowing $Y$. Thus, MaxL quantifies the maximal logarithmic gain in guessing any possible function of $X$ when an adversary has access to $Y$.

	Analogously to the derivation of MaxL from $0$-$1$ loss, we introduce $\alpha$-leakage and maximal $\alpha$-leakage based on $\alpha$-loss (under the assumptions of discrete random variables and finite supports). The formal definitions are as follows. 
 
	\begin{definition}[{$\alpha$-Leakage}]\label{Def:alphaLeakge}
		Given a joint distribution $P_{X,Y}$ and an estimator $\hat{X}$ with the same support as $X$, the $\alpha$-leakage from $X$ to $Y$ is defined as
		\begin{align}
			\label{eq:alphaLeak_definition}
			\mcl L_{\alpha}(X\to Y)
				\triangleq\frac{\alpha}{\alpha-1}\log\frac{\max\limits_{P_{\hat{X}|Y}}\mathbb{E}\left[P_{\hat{X}|Y}(X|Y)^{\frac{\alpha-1}{\alpha}}\right]}{\max\limits_{P_{\hat{X}}}\mathbb{E}\left[P_{\hat{X}}(X)^{\frac{\alpha-1}{\alpha}}\right]},
		\end{align}
		for $\alpha\in(1,\infty)$ and by the continuous extension of \eqref{eq:alphaLeak_definition} for $\alpha = 1$ and $\infty$.
	\end{definition}
	Note that for any specific function $U$ of $X$, the joint probability distribution of $X$ and the $U$ is known, and therefore, $\alpha$-leakage can also be used to measure the the inference gain in inferring the specific function $U$ from the released data $Y$. In addition, the two maximizations in the numerator and denominator of the logarithmic ratio in \eqref{eq:alphaLeak_definition} imply the optimal adversarial actions in the sense of minimizing the expected $\alpha$-loss in Lemma \ref{lem:Minimal-expectedAlphaLoss}. Therefore, it limits the inference gain that an adversary can obtain by minimizing the expected $\alpha$-loss, no matter the adversary has prior knowledge (i.e., the probability distribution of the original data) of the original data or not.
	
	Whereas $\alpha$-leakage captures how much an adversary can learn about $X$ (or a specific function of $X$) from $Y$, we also wish to quantify the information leaked about \textit{any function} of $X$ through $Y$. To this end, we define maximal $\alpha$-leakage below.
	\begin{definition}[Maximal $\alpha$-Leakage]\label{Def:GeneralLeakge}
		Given a joint distribution $P_{X,Y}$ on finite alphabets $\mcl X\times\mcl Y$, the maximal $\alpha$-leakage from $X$ to $Y$ is defined as
	\begin{align}
			\label{eq:GealLeak_definition}
			\mcl L_{\alpha}^{\text{max}}(X\to Y)		
			\triangleq&\sup_{U- X- Y }\mcl L_{\alpha}(U; Y),		
		\end{align}
		where $1\leq \alpha\leq \infty$, and $U$ represents any function of $X$ and takes values from an arbitrary finite alphabet. 
	\end{definition}	
	Note that for $\alpha\geq 1$,
	\begin{align}\label{eq:MaxAlphaLK-AlphaLoss}
	&\max\limits_{P_{\hat{U}|Y}} \mathbb{E}\left[P_{\hat{U}|Y}(U|Y)^{\frac{\alpha-1}{\alpha}}\right]\nonumber \\
	=&1-\frac{\alpha-1}{\alpha}\min\limits_{P_{\hat{U}|Y}}\mathbb{E}\left[\ell_{\alpha}(U,Y,P_{\hat{U}|Y})\right].
	\end{align}
	Thus, there is a similar connection between maximal $\alpha$-leakage and $\alpha$-loss (in Def. \ref{Def:alpha-loss}) as that observed in \eqref{eq:MaxL-0-1Loss} between MaxL and $0$-$1$ loss, and maximal $\alpha$-leakage quantifies an adversary's capability to infer \textit{any function} of data $X$ from the released $Y$.

	Making use of the result in Lemma \ref{lem:Minimal-expectedAlphaLoss}, the following theorem simplifies the expression of $\alpha$-leakage in \eqref{eq:alphaLeak_definition}.
	\begin{theorem}\label{Thm:DefEquialentExpression_alphaleakage}
		For $1\leq \alpha\leq \infty$, $\alpha$-leakage defined in \eqref{eq:alphaLeak_definition} simplifies to
		\begin{align}\label{eq:alphaLeak_EquivDef}
			\mcl L_{\alpha}(X\to Y)=I_{\alpha}^{\text{A}}(X;Y) .
		\end{align}	
	\end{theorem}
	From \eqref{eq:MaxAlphaLK-AlphaLoss} and Lemma \ref{lem:Minimal-expectedAlphaLoss}, we simplify the scaled logarithm of the ratio in \eqref{eq:alphaLeak_definition} to Arimoto MI. A detailed proof is in Appendix \ref{Proof:DefEquialentExpression_alphaleakage}, where we show that Arimoto conditional entropy and R{\'e}nyi entropy capture the inference uncertainties of an adversary for knowing $Y$ or not, respectively, and $\alpha$-leakage measures the decrease in the inference uncertainty by knowing $Y$. 
	
	Making use of the conclusion in Thm. \ref{Thm:DefEquialentExpression_alphaleakage}, the following theorem gives equivalent expressions for maximal $\alpha$-leakage. Note that in the following theorem we use the well-known equivalence of the supremums of Sibson and Arimoto MIs \cite[Thm. 5]{alphaMI_verdu}.
	\begin{theorem}\label{Thm:DefEquialentExpression}
		For $1\leq \alpha\leq \infty$, the maximal $\alpha$-leakage defined in \eqref{eq:GealLeak_definition} simplifies to			\vspace{5pt}\\
\begin{subequations}
	\label{eq:GealLeak_EquivDef}
	$\mcl L_{\alpha}^{\text{max}}(X\to Y)$	
	\begin{align}[left={=\empheqlbrace\,}]
	&\sup_{P_{\tilde{X}}}I^{\text{S}}_\alpha(\tilde{X};Y)=\sup_{P_{\tilde{X}}}I^{\text{A}}_\alpha(\tilde{X};Y),& 1<\alpha\leq \infty  \label{eq:GealLeak_EquivDef_1infty}\\
	& I(X;Y),  &\alpha=1  \label{eq:GealLeak_EquivDef_1}
	\end{align}
\end{subequations}	
		where $P_{\tilde{X}}$ is a probability distribution over the support of $P_X$.
	\end{theorem}
		Note that maximal $\alpha$-leakage is essentially the Arimoto channel capacity (with a support-set constrained input distribution) for $\alpha> 1$ \cite{AlphaMI_Arimoto1975}, which is used to characterize probabilities of decoding error for scenarios in which transmission rates are higher than channel capacity. The limit of maximal $\alpha$-leakage for $\alpha=1$ gives the Shannon channel capacity. Recall that the limit of $\alpha$-loss in \eqref{poly_loss} leads to the log-loss (for $\alpha=1$) and 0-1 loss (for $\alpha=\infty$) functions, respectively. Consequently, for $\alpha=1$ and $\infty$, maximal $\alpha$-leakage simplifies to MI and MaxL, respectively. 

   A detailed proof for Thm. \ref{Thm:DefEquialentExpression} is in Appendix \ref{Proof:DefEquialentExpression}. We summarize key steps in the proof as follows: by applying Thm. \ref{Thm:DefEquialentExpression_alphaleakage}, we write maximal $\alpha$-leakage as
	\begin{align}
		\label{eq:Thm_MaxAlphaLeak_ProofSketch}
		\mcl L_{\alpha}^{\text{max}}(X\to Y)&=\sup_{U-X-Y}I_{\alpha}^{\text{A}}(U;Y)\quad \alpha\in[1,\infty].
	\end{align}
	For $\alpha=1$, Arimoto MI is simply the Shannon's MI, and combining with the data processing inequalities, \eqref{eq:Thm_MaxAlphaLeak_ProofSketch} simplifies to $I(X;Y)$. Note that for $\alpha> 1$, Arimoto MI does not satisfy data processing inequalities. By using the facts that Arimoto MI and Sibson MI have the same supremum \cite[Thm. 5]{alphaMI_verdu} and that Sibson MI satisfies data processing inequalities \cite[Thm. 3]{alphaMI_verdu}, we limit the supremum in \eqref{eq:Thm_MaxAlphaLeak_ProofSketch}
	by $\sup_{P_{\tilde{X}}}I^{\text{S}}_\alpha(\tilde{X};Y)$, and then, show that the upper bound $\sup_{P_{\tilde{X}}}I^{\text{S}}_\alpha(\tilde{X};Y)$ can be achieved by a specific $U$ with $H(X|U)=0$.

	\begin{example}
		Given a binary channel 
		\begin{align}
			P_{Y|X}=\begin{bmatrix}
			1-\rho_1 & \rho_1\\
			\rho_2   & 1-\rho_2
			\end{bmatrix},
		\end{align}
		where $\rho_1,\rho_2\in[0,1]$ are the crossover probabilities, maximal $\alpha$-leakage in \eqref{eq:GealLeak_EquivDef} is given by
		\begin{IEEEeqnarray}{l l}
			 &\mcl L_{\alpha}^{\text{max}}(X\to Y)\nonumber\\
			 =&\frac{\alpha}{\alpha-1}\log \Bigg(\Big|(1-\rho_1)^\alpha(1-\rho_2)^\alpha-\rho_1^\alpha\rho_2^\alpha\Big|^\frac{1}{\alpha} \nonumber\\
			 &\,\cdot\bigg(\Big|(1-\rho_2)^{\alpha}\hspace{-2pt}-\hspace{-1pt}\rho_1^{\alpha}\Big|^\frac{1}{1-\alpha}\hspace{-2pt}+\hspace{-1pt}\Big|(1-\rho_1)^{\alpha}\hspace{-2pt}-\hspace{-1pt}\rho_2^{\alpha}\Big|^\frac{1}{1-\alpha}\bigg)^{\hspace{-2pt}\frac{\alpha-1}{\alpha}}\hspace{-1pt}\Bigg).\label{eq:MaxL-Binary}
		\end{IEEEeqnarray}
		If $\rho_1=\rho_2$, \eqref{eq:MaxL-Binary} simplifies to 
		\begin{align}
			\mcl L_{\alpha}^{\text{max}}(X\to Y)=\frac{1}{\alpha-1}\log \left((1-\rho_1)^\alpha+\rho_1^\alpha\right)+\log2,
		\end{align} which is exactly the $\alpha$-leakage for the binary symmetric channel with the uniform input distribution.
		Fig.~\ref{fig:MaxAlphaLK} plots the values of maximal $\alpha$-leakage for example channels where $\rho_1=\rho_2$ and $\rho_1\neq\rho_2$, and shows that the ordering of leakages for the two channels varies with $\alpha$.%various values of $\alpha$.
	\end{example}
 \begin{figure}
	\centering
	\captionsetup{justification=centering} 	
	\includegraphics[width=3.5  in]{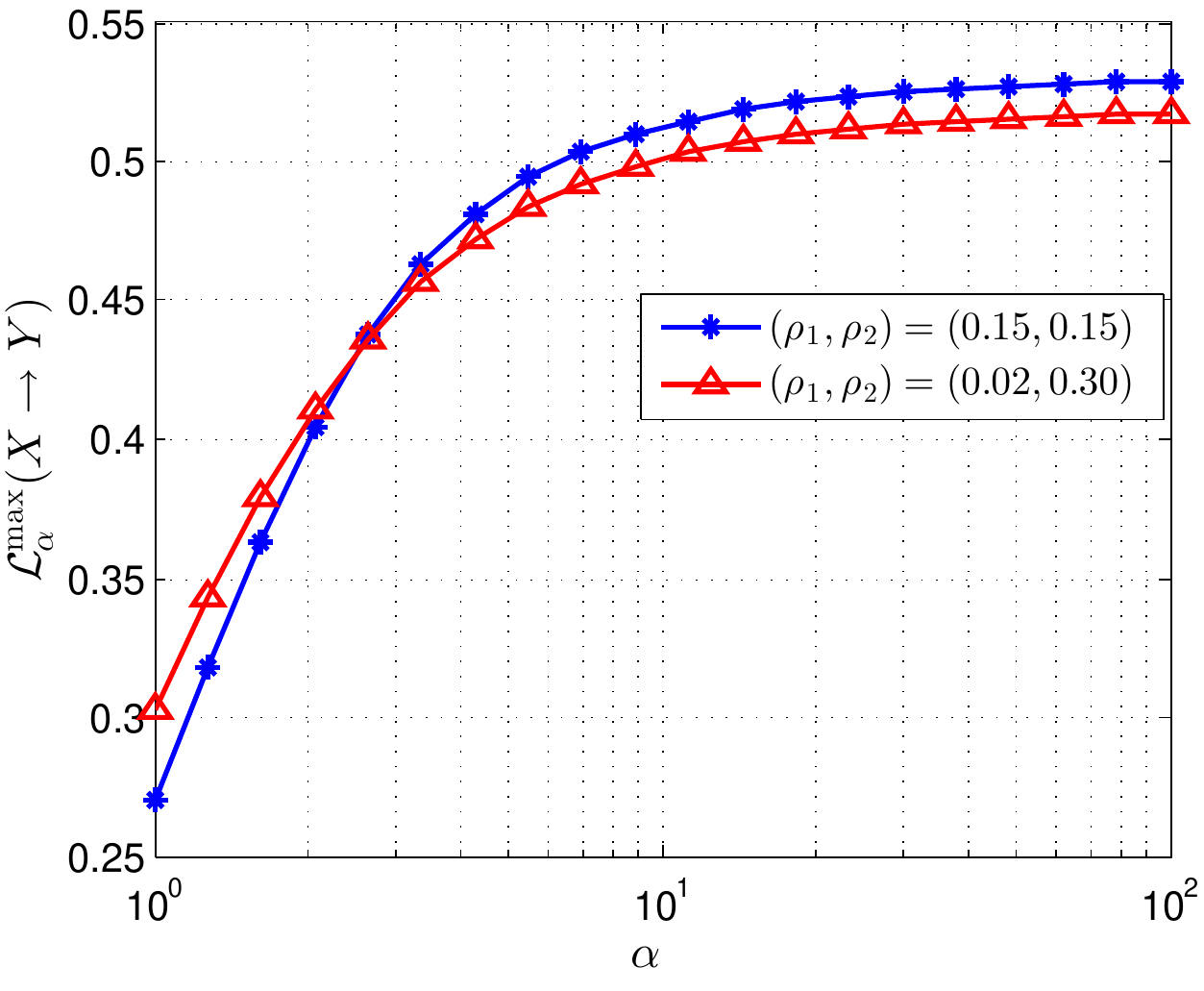}\label{fig:MaxAlphaLK-value} 			                   
	\caption{The values of maximal $\alpha$-leakage for binary channels determined by a pair of crossover probabilities $(\rho_1,\rho_2)$.} 
	\label{fig:MaxAlphaLK}  	
\end{figure}

	\subsection{Leakage Measures Based on $f$-Divergence}\label{Subsec:f-divergence-based leakages}
	We introduce two classes of information leakages derived from $f$-divergence, called \textit{$f$-leakage} and \textit{maximal $f$-leakage}. The $f$-leakage depends on the distribution of original data, and in contrast, maximal $f$-divergence only depends on the support of original data. We also show the relation between the $f$-divergence-based measures and maximal $\alpha$-leakage for $\alpha=1$ and $\alpha>1$, respectively.
	
	Recall that for a convex function $f:\mathbb{R}_+\to \mathbb{R}$ such that $f(1)=0$, an $f$-divergence $D_f$ is a measure of the distance between two distributions given by
	\begin{align}\label{eq:f-divergence}
	D_f(P_Y\|Q_Y) = \sum_y Q(y)\, f\left(\frac{P(y)}{Q(y)}\right).
	\end{align}		
	\begin{definition}\label{Def:Def-IntendedMeasure}
		Given a joint distribution $P_{X,Y}=P_{Y|X}P_X$ and a $f$-divergence $D_f$, the $f$-leakage is defined as %on finite alphabets $\mcl X \times \mcl Y$ 
		\begin{align}\label{eq:Def-fLeakKforDist}
		\mcl L_f(X\to Y)=\inf_{Q_Y} D_f(P_{X,Y}\|P_X\times Q_Y),
		\end{align} 
        and the maximal $f$-leakage is defined as
		\begin{align}\label{eq:Definition-f_divergence_measure}
		\mcl L_f^{\text{max}}(X\to Y)=\sup_{P_{\tilde{X}}}\,\inf_{Q_Y}\, D_f(P_{Y|X}P_{\tilde{X}}\|P_{\tilde{X}}\times Q_Y),
		\end{align}
		where $P_{\tilde{X}}$ is a distribution over the support of $P_X$. %constrained to be
	\end{definition}
    Note that in Definition \ref{Def:Def-IntendedMeasure}, maximal $f$-leakage ($\mcl L_f^{\text{max}}$) depends on the distribution of $X$ only through its support. In contrast, $f$-leakage ($\mathcal{L}_f$) depends fully on the distribution of $X$. Both measures depend on the chosen mechanism $P_{Y|X}$. 
    
	Recall that for $\alpha=1$, maximal $\alpha$-leakage is MI. Therefore, it is a special case of $\mcl L_f(X\to Y)$ in \eqref{eq:Def-fLeakKforDist} with $f(t)=t\log t$. Furthermore, for $\alpha>1$, maximal $\alpha$-leakage has a one-to-one relationship with a special case of $\mcl L_f^{\text{max}}$ in \eqref{eq:Definition-f_divergence_measure} for $f$ given by 
	\begin{equation}\label{eq:Hellinger_f-function}
	f_\alpha(t)=\frac{1}{\alpha-1} (t^\alpha-1),
	\end{equation}
	such that $D_f$ is the Hellinger divergence of order $\alpha$ \cite{Liese2006}. The following lemma makes this observation precise.
	\begin{lemma}\label{Lem:alphaLeakage_fDivergenceLeakage}
		For discrete random variables $X$ and $Y$, the maximal $\alpha$-leakage ($\alpha>1$) from $X$ to $Y$ can be written as	
		\begin{align}	\label{eq:alphaLeakage_HellingerDivergenceLeakage}	
		\mcl L_{\alpha}^{\text{max}}(X\to Y)=\mathsmaller{\frac{1}{\alpha-1} \log \big(1+(\alpha-1) \mathcal{L}_{f_\alpha}^{\text{max}}(X\to Y)\big)},
		\end{align}
	where $\mathcal{L}_{f_\alpha}^{\text{max}}(X\to Y)$ indicates a set of maximal $f$-leakage in \eqref{eq:Definition-f_divergence_measure} defined  from the function given by \eqref{eq:Hellinger_f-function}.
	\end{lemma}
    A detailed proof is in Appendix \ref{Proof:alphaLeakage_fDivergenceLeakage}. Note that substituting $f_{\alpha}$ defined in \eqref{eq:Hellinger_f-function} into \eqref{eq:f-divergence}, we can obtain the Hellinger divergences of order $\alpha>1$. Thus, from Lemma \ref{Lem:alphaLeakage_fDivergenceLeakage}, for $\alpha>1$, maximal $\alpha$-leakage can be transformed to the maximal $f$-leakage based on Hellinger divergences via a one-to-one mapping. In this sense, maximal $\alpha$-leakage is a special case of maximal $f$-leakage.

	\section{Properties of Maximal $\alpha$-Leakage}\label{Sec:Properties} 
	Thm. \ref{Thm:DefEquialentExpression_alphaleakage} shows that $\alpha$-leakage is exactly Arimoto MI, and therefore, several basic properties of $\alpha$-leakage have been shown including (i) non-negativity \cite[Sec. II-A]{alphaMI_verdu}, (ii) quasi-convexity\footnote{For $\alpha\geq 1$ and $P_X$, the Arimoto MI $I^{\text{A}}_{\alpha}(X;Y)$ is the logarithm of a linear combination of the $p$-norm ($p=\alpha$) $\|P_{Y|X}(\cdot|x)\|_{\alpha}$. From \cite[Chapter 3.5]{boydconvex}, we know a log-convex function is quasi-convex such that $I^{\text{A}}_{\alpha}(X;Y)$ is quasi-convex in $P_{Y|X}$ given $P_X$.} in $P_{Y|X}$ given $P_X$ \cite[Chapter 3.5]{boydconvex}, and (iii) post-processing inequality\footnote{From the monotonicity of conditional Arimoto entropy \cite[Cor. 1]{Arimotoconditional_fehr2014}, one can derive that for a Markov chain $X-Y-Z$, $I^{\text{A}}_{\alpha}(X;Z)\leq I^{\text{A}}_{\alpha}(X;Y)$.} \cite[Cor. 1]{Arimotoconditional_fehr2014}. We now explore proprieties of maximal $\alpha$-leakage and show that its properties include: (i) quasi-convexity in the conditional distribution $P_{Y|X}$; (ii) data processing inequalities; (iii) sub-additivity (composition property \cite{OperationalLeak_issa2018}) and additivity for memoryless mechanisms.

	The following theorem results from the expression of maximal $\alpha$-leakage in Thm. \ref{Thm:DefEquialentExpression} as well as some known properties of Sibson MI \cite{alphaMI_Sibson1969,ConvexityAlphaMI_Ho,alphaMI_verdu}.
	\begin{theorem}\label{Thm:Geneleak_qusiconvex_nondecreasing_dataprocessing}
		For $1\leq \alpha\leq \infty$, maximal $\alpha$-leakage %$\mcl L_{\alpha}^{\text{max}}(X\hspace{-0.05in}\to\hspace{-0.04in}Y)$ 
		\begin{itemize}
			\item[1.] is quasi-convex in $P_{Y|X}$;
			\item[2.] is monotonically non-decreasing in $\alpha$;
			\item[3.] satisfies data processing inequalities: let random variables $X,Y,Z$ form a Markov chain, i.e., $X-Y-Z$, then
			\begin{subequations}\label{eq:GeneLeak_DataProcessIneq}
				\begin{align}
					\mcl L_{\alpha}^{\text{max}}(X\to Z)\leq \mcl L_{\alpha}^{\text{max}}(X\to Y) \label{eq:GeneLeak_DataProcessIneq_XY}\\
					\mcl L_{\alpha}^{\text{max}}(X\to Z)\leq \mcl L_{\alpha}^{\text{max}}(Y\to Z) \label{eq:GeneLeak_DataProcessIneq_YZ}.
				\end{align}
			\end{subequations}
			\item[4.]
			satisfies
			\begin{align}
				\mcl L_{\alpha}^{\text{max}}(X\to Y)\geq 0
			\end{align}with equality if and only if $X$ is independent of $Y$, and 
			\begin{align}\label{eq:MaxAlphaLeak-RangeUpper}
				\mcl L_{\alpha}^{\text{max}}(X\to Y)\leq \begin{cases}
					\log|\mathcal{X}|  &\alpha>1\\
					H(P_X)&\alpha=1 
				\end{cases}
			\end{align} with equality if and only if $X$ is a deterministic function of $Y$. 
		\end{itemize} 
	\end{theorem}
	A detailed proof is in Appendix \ref{Proof:Geneleak_qusiconvex_nondecreasing_dataprocessing}. 
	\begin{remark}
		Note that:
		\begin{itemize}
			\item Since both MI and MaxL are convex in $P_{Y|X}$, $\mcl L^{\text{max}}_{1}(X\to Y)$ and $\mcl L^{\text{max}}_{\infty}(X\to Y)$ are convex in $P_{Y|X}$. 
			\item From the monotonicity in Part 2, we can bound maximal $\alpha$-leakage from above by\footnote{For $\alpha=\infty$, the $I^{\text{S}}_{\infty}(P_X,P_{Y|X})$ depends on the marginal distribution $P_X$ only through the support of $X$.}
			\begin{align}
			\mcl L_{\alpha}^{\text{max}}(X \to Y)\leq \mcl L_{\text{MaxL}}(X\to Y)=I^{\text{S}}_{\infty}(X;Y).
			\end{align}
			\item The data processing inequalities in \eqref{eq:GeneLeak_DataProcessIneq_XY} and \eqref{eq:GeneLeak_DataProcessIneq_YZ} are called \textit{post-processing inequality} and \textit{linkage inequality}, respectively \cite{AxiomaticStatisticalPrivacyUtility_Kifer12,PUTConstrainedDataReleaseMech_Wang17}. It is worth noting that not all information leakage measures satisfy the linkage inequality \cite{PUTConstrainedDataReleaseMech_Wang17,TVDprivacy_Rassouli&Gunduz18}. Examples include $\alpha$-leakage, maximal information leakage \cite{PrivacyAgainstStatistic_Calmon12}, probability of correctly guessing, and DP.%,privacyGuessing_asoodeh2017,dwork2011differential
			\item From the monotonicity of maximal $\alpha$-leakage and the upper bound of MaxL in \cite[Lemma 1]{OperationalLeak_issa2018}, we know that if $|\mathcal{Y}|< |\mathcal{X}|$, the upper bound in \eqref{eq:MaxAlphaLeak-RangeUpper} can be tighter as
					\begin{align}
					\mcl L_{\alpha}^{\text{max}}(X\to Y)\leq \begin{cases}
					\log|\mathcal{Y}|  &\alpha>1\\
					\min\{H(P_X), \log|\mathcal{Y}|\}&\alpha=1,
					\end{cases}\nonumber
					\end{align}
				with equality for $\alpha>1$ if and only if $Y$ is a deterministic function of $X$.
		\end{itemize}
	\end{remark}
From Thm. \ref{Thm:DefEquialentExpression}, we know that for $\alpha>1$, maximal $\alpha$-leakage is the supremum of Arimoto/Sibson MI over all possible distributions on the support of original data, and therefore, is a function of a conditional probability distribution. The following theorem bounds the supremum from below by a closed-form expression of the conditional probability distribution.
\begin{theorem}[Lower Bound]\label{Thm:MaxAlphaLeak-Bounds}
	For $1<\alpha\leq \infty$, maximal $\alpha$-leakage is bounded from below by
	 \begin{align}
	\mcl L_{\alpha}^{\text{max}}(X\to Y)\geq\frac{\alpha}{\alpha-1}\log\frac{\sum\limits_{y\in\mcl Y} \|P_{Y|X}(y|\cdot)\|_{\alpha} }{|\mathcal{X}|^{\frac{1}{\alpha}}},
	\end{align} with equality if and only if for all $x_1,x_2\in \mcl X$, there is
	\begin{align}\label{eq:MaxAlphaLeak-LowBD-ComplementCnst}
	\sum_y\frac{P_{Y|X}(y|x_1)^\alpha}{\|P_{Y|X}(y|\cdot)\|_\alpha^{\alpha-1}}=\sum_y\frac{P_{Y|X}(y|x_2)^\alpha}{\|P_{Y|X}(y|\cdot)\|_\alpha^{\alpha-1}} .
	\end{align}
\end{theorem}
A detailed proof is in Appendix \ref{Proof:Thm:MaxAlphaLeak-Bounds}.

When data may be revealed multiple times (e.g., entering a password multiple times), it is essential to quantify how mechanisms are designed with maximal $\alpha$ leakage compose in terms of total leakage. 
Consider two released versions $Y_1$ and $Y_2$ of $X$. The following theorem limits maximal $\alpha$-leakage to an adversary who has access to both $Y_1$ and $Y_2$ simultaneously.
	\begin{theorem}[Sub-additivity/Composition]\label{Thm:GeneLeak_CompositionTheory}
		Given a Markov chain $Y_1-X-Y_2$, we have $(\alpha\in [1,\infty])$ 
		\begin{align}
			\mcl L_{\alpha}^{\text{max}}(X\to Y_1,Y_2)\leq \sum_{i\in\{1,2\}}\mcl L_{\alpha}^{\text{max}}(X\to Y_i).
		\end{align}
	\end{theorem}
A detailed proof is in Appendix \ref{proof:Thm:GeneLeak_CompositionTheory}.

The following theorem shows the additivity of maximal $\alpha$-leakage for memoryless mechanisms.
\begin{theorem}[Additivity for Memoryless Mechanisms]\label{thm:MaxAlphaLK-Addtivity}
	For $\alpha\in [1,\infty]$ and a finite integer $n>0$, let $X^n$ and $Y^n$ be $n$-length input and output, respectively, of a memoryless mechanism with no feedback, i.e., 
	\begin{align}
	P_{Y^n|X^n}=\prod_{i=1}^{n}P_{Y_{i}|X_{i}},
	\end{align}
	where $X_{i}$ and $Y_{i}$ represent the $i^{\text{th}}$ element of $X^n$ and $Y^n$, respectively, such that 
	\begin{itemize}
		\item[(1)] for $\alpha>1$
		\begin{align}
		\mcl L_{\alpha}^{\text{max}}(X^n\to Y^n)=\sum_{i=1}^{n}  \mcl L_{\alpha}^{\text{max}}(X_{i}\to Y_{i})
		\end{align}
		\item[(2)] for $\alpha=1$
		\begin{align}
			\mcl L_{1}^{\text{max}}(X^n\to Y^n)\leq \sum_{i=1}^{n}  \mcl L_{1}^{\text{max}}(X_{i}\to Y_{i})
		\end{align}
		with equality if and only if entries of $X^n$ are mutually independent. 
	\end{itemize}
\end{theorem}
A detailed proof is in Appendix \ref{proof:thm:MaxAlphaLK-Addtivity}.

\section{Privacy-Utility Tradeoff with a Hard Distortion Constraint}\label{Sec:PUT_for-HD}
In a privacy-guaranteed data publishing setting, a data curator/provider uses a mapping called \textit{privacy mechanism} to generate distorted versions of original data for releases. The privacy mechanism determines the fidelity of the released data. With a higher fidelity, more utility is maintained, while less privacy preserved. Therefore, a privacy-utility tradeoff (PUT) problem arises in the design of the privacy mechanism. 

We consider the two different data publishing scenarios shown in Figs. \ref{fig:Datamodel_1} and \ref{fig:Datamodel_2}: the first where the entirety of the dataset $X$ is considered private, and the second where the dataset consists of two parts $S$ and $X$, where only $S$ is considered private. For the first case (Fig. \ref{fig:Datamodel_1}), we use maximal $\alpha$-leakage as the privacy measure, thereby limiting the inference of any private information about the dataset represented by the function $U$. For the second case (Fig. \ref{fig:Datamodel_2}), we use $\alpha$-leakage as the privacy measure, thereby limited the inference only of the specific private information represented by $S$.

 We measure utility in terms of a \emph{hard distortion} measure, which constrains the privacy mechanism so that the distortion between each pair of original and released data is bounded with probability $1$. Unlike average distortion measures, the hard distortion measure gives all data samples the same fidelity guarantee (distortion bound), which is independent of the probabilities of the samples. Therefore, the hard distortion measure excludes the case that large distortions are applied to samples with very small probabilities, which is possible under average distortion constraints. The fidelity guarantee can lead to better performance for applications for which low probability events cannot be ignored or excluded easily (e.g., anomaly detection from released datasets or high-fidelity empirical distribution estimation for census applications) but is incompatible with some privacy measures like DP and L-DP. 
Specifically, for the original and released data $X,Y$ and a distortion function $d(\cdot,\cdot)$, the utility guarantee is modeled as the hard distortion constraint $d(X,Y)\le D$ \textit{with probability $1$}, where $D$ is the maximal permitted distortion. In other words, if a privacy mechanism $P_{Y|X}$ satisfies the hard distortion constraint, given input $x$, the output $y$ of the privacy mechanism must lie in a \textit{non-empty} set $B_D(x)$ given by
\begin{equation}
\label{eq:PUT_HardDist_CollectofFeasibleY}
B_D(x)\triangleq\{y:d(x,y)\le D\},
\end{equation}
i.e., for any $x$ with $P_X(x)>0$, $P_{Y|X}(y|x)=0$ if $y\notin B_D(x)$. Thus, a mathematical model of the PUT problem is given by
\begin{subequations}
	\label{eq:PUT-HD}
	\begin{align}
		\inf_{P_{Y|X}\in\mcl P_{Y|X}}\quad &\mcl L^{(\cdot)}_{(\cdot)}(X\to Y) \label{eq:PUT-HD-obj}\\
	   \text{s.t.,} \quad  & d(X,Y)\leq D,\label{eq:PUT-HD-constriant}
	\end{align}
\end{subequations}
where 
the set $\mcl P_{Y|X}$ is the collection of stochastic matrices, and the superscript and subscript of $\mcl L$ depend on the privacy measure under consideration (see Sec. \ref{Sec:Information Leakage Measures} for notation). 

\begin{remark}
	Note that given any input $x$, the hard distortion constraint in \eqref{eq:PUT-HD-constriant} will force the conditional probabilities of the outputs that are not in $B_D(x)$ to be zero. Thus, this utility guarantee is \emph{incompatible} with some privacy notions, which require each input to be mapped to all outputs with some positive probabilities; e.g., DP and any maximal $f$-leakage with $f(0)=\infty$.
\end{remark}

\subsection{PUTs for Entirely Sensitive Datasets}\label{Subsec:PUT_for-HDvsMaxAlphaLK}

For the privacy-guaranteed publishing of an entirely sensitive dataset shown in Fig. \ref{fig:Datamodel_1}, we use maximal $\alpha$-leakage as the privacy measure. From Section \ref{Subsec:f-divergence-based leakages}, we know that maximal $\alpha$-leakage is a specific case of $f$-leakage and maximal $f$-leakage (in Def. \ref{Def:Def-IntendedMeasure}) for $\alpha=1$ and $\alpha>1$, respectively. Hereby, we solve the PUT problems which minimize either $f$-leakage or maximal $f$-leakage, subject to a hard distortion constraint. By applying the relations between maximal $\alpha$-leakage and the $f$-divergence-based variants, we derive the optimal PUTs and optimal privacy mechanisms for the PUT problem with maximal $\alpha$-leakage as the privacy measure. We denote an optimal PUT as $\text{PUT}_{\text{HD},\mcl L^{(\cdot)}_{(\cdot)}}$, where HD and $\mcl L^{(\cdot)}_{(\cdot)}$ in the subscript indicate the hard distortion and the involved privacy measure, respectively. 

The following theorem characterizes the optimal tradeoff, i.e., the minimal leakage for any given distortion bound $D$, denoted as $\text{PUT}_{\text{HD},\mcl L_f}(D)$, in \eqref{eq:PUT-HD} for the case that $f$-leakage is used as the privacy measure. 
\begin{theorem}\label{thm:PUT_fLeakKforDistvsHardDist}
	For any $f$-leakage $\mcl L_f$ in \eqref{eq:Def-fLeakKforDist} and a distortion function $d(\cdot,\cdot)$ with $B_D(x)$ in \eqref{eq:PUT_HardDist_CollectofFeasibleY}, the optimal PUT in \eqref{eq:PUT-HD} is given by
		\begin{align}
		\hspace{-50 pt}&\text{PUT}_{\text{HD},\mcl L_f} (D)\nonumber\\
		=&\inf_{P_{Y|X}: d(X,Y)\le D}\,\mcl L_f(X; Y) \label{eq:PUT_fLeakKforDistvsHardDist},\\		
		=&f(0)+\inf_{Q_Y} \mathbb{E} \mathsmaller{\left[Q_{Y}(B_D(X))\Big(\hspace{-2 pt}f\big(\frac{1}{Q_Y(B_D(X))}\big)\hspace{-1pt}-\hspace{-1 pt}f(0)\hspace{-2 pt}\Big)\hspace{-2 pt}\right]}.
		\label{eq:PUT-fLeakKforDist-HDdist-OptValue}
		\end{align}
	Moreover, letting $Q_Y^\star$ be the distribution achieving the infimum in \eqref{eq:PUT-fLeakKforDist-HDdist-OptValue}, an optimal mechanism $P^{\star}_{Y|X}$ is given by 
	\begin{equation}\label{eq:opt_mech}
	P_{Y|X}^{\star}(y|x)=\frac{\mathbf{1}\big(d(x,y)\le D\big)Q_Y^\star(y)}{Q_Y^\star(B_D(x))}.
	\end{equation}
\end{theorem}
A detailed proof in Appendix \ref{Proof:thm:PUT_fLeakKforDistvsHardDist}. Note that as a result of the distribution dependence of the leakage measure $\mcl L_f$ in \eqref{eq:Def-fLeakKforDist}, the optimal tradeoff in \eqref{eq:PUT-fLeakKforDist-HDdist-OptValue} is an \textit{expected} function of $X$. The optimal mechanism $P^*_{Y|X}(y|x)$ is, in fact, the normalized probability of $y$ when the conditional support of $Y$ given $X=x$ is restricted to $B_D(x)$ (i.e., $Y$ is restricted to taking values in $B_D(x)$ for a given $x$).

In \eqref{eq:PUT-HD}, making use of maximal $f$-divergence as the privacy measure, the optimal PUT, denoted as $\text{PUT}_{\text{HD},\mcl L_f^{\text{max}}}(D)$, with respect to the hard distortion constraint is given as the minimal leakage for any given distortion bound $D$ in the following theorem.
\begin{theorem}\label{Thm:PUT_fLeak_HardDist}
	For any maximal $f$-leakage $\mcl L_f^{\text{max}}$ in \eqref{eq:Definition-f_divergence_measure}, a distortion function $d(\cdot,\cdot)$ and $B_D(x)$ in \eqref{eq:PUT_HardDist_CollectofFeasibleY}, the optimal PUT in \eqref{eq:PUT-HD} is given by
	\begin{align}
	\text{PUT}_{\text{HD},\mcl L_f^{\text{max}}} (D)=&\inf_{P_{Y|X}: d(X,Y)\le D}\,\mcl L_f^{\text{max}}(X\to Y), 	\label{eq:PUT_fDivergenceLeak_HardDistortion}\\
	=&q^\star f((q^\star)^{-1})+(1-q^\star) f(0) \label{eq:PUT-fLeakK-HDdist-OptValue},
	\end{align}
	with $q^{\star}$ defined as
	\begin{equation}\label{eq:q_star_def}
	q^\star\triangleq \sup_{Q_Y}\, \inf_x\, Q_Y(B_D(x)).
	\end{equation}
	Moreover, letting $Q_Y^\star$ be the distribution achieving the supremum in \eqref{eq:q_star_def}, an optimal mechanism $P_{Y|X}^{\star}$ is given by \eqref{eq:opt_mech}.
\end{theorem}
A detailed proof is in Appendix \ref{Proof:Thm:PUT_fLeak_HardDist}. Observe that, in contrast to the optimal tradeoff $\text{PUT}_{\text{HD},\mcl L_f}$ for $f$-leakage in Thm. \ref{thm:PUT_fLeakKforDistvsHardDist}, which depends on the probability distribution $P_X$, the optimal tradeoff $\text{PUT}_{\text{HD},\mcl L_f^{\text{max}}}$ for maximal $f$-leakage depends only on the support of $P_X$. This results from the fact that the maximal $f$-leakage $\mcl L_f^{\text{max}}$ in \eqref{eq:Definition-f_divergence_measure} is the supremum over all possible probability distribution on the support of $P_X$, and therefore, depends on $P_X$ only through the support.

The next corollary characterizes the optimal tradeoff $\text{PUT}_{\text{HD},\mcl L_\alpha^{\text{max}}}$ for maximal $\alpha$-leakage. Recall that for $\alpha=1$, $\mcl L_{1}^{\text{max}}$ equals $\mcl L_f$ with $f(t)=t\log t$. For $\alpha>1$, from the one-to-one relationship between $\mcl L_\alpha^{\text{max}}$ and $\mcl L_{f_\alpha}^{\text{max}}$ in \eqref{eq:alphaLeakage_HellingerDivergenceLeakage}, we know that finding $\text{PUT}_{\text{HD},\mcl L_\alpha^{\text{max}}}$ is equivalent to finding the optimal tradeoff $\text{PUT}_{\text{HD},\mcl L_f^{\text{max}}}$ in \eqref{eq:PUT_fDivergenceLeak_HardDistortion} for $\mcl L_f^{\text{max}}=\mcl L_{f_\alpha}^{\text{max}}$. 
\begin{corollary}\label{Col:PUT_MaximalAlphaLeak_HardDist}
	For maximal $\alpha$-leakage, the optimal PUT in \eqref{eq:PUT-HD} is given by 
	\hspace{-50pt}\begin{align}\label{eq:PUT-MaxAlphaLeak-U_HardDistortion}
	\text{PUT}_{\text{HD},\mcl L_\alpha^{\text{max}}} (D)=\hspace{-5pt}\inf_{P_{Y|X}: d(X,Y)\le D}\mcl L^{\text{max}}_\alpha(X\to Y),\,\qquad
	\end{align}
	\begin{subnumcases}{\hspace{+45pt}=}%\text{PUT}_{\text{HD},\mcl L_\alpha^{\text{max}}} (D)\hspace{-3pt}
	\inf_{Q_Y} \mathbb{E}\left[\log \mathsmaller{\frac{1}{Q_Y(B_D(X))}}\right],  \hspace{-15pt} & $\alpha=1$ \label{eq:PUT-MI-HDdist-OptValue}
	\\
	-\log q^\star,  \hspace{-15pt} & $\alpha>1$ \label{eq:PUT-MaxL-HDdist-OptValue}
	\end{subnumcases}
	where $q^\star$ is defined in \eqref{eq:q_star_def}. Moreover, an optimal mechanism is given by \eqref{eq:opt_mech}, where for $\alpha=1$, $Q_Y^\star$ achieves the infimum in \eqref{eq:PUT-MI-HDdist-OptValue}; and for $\alpha>1$, $Q_Y^\star$ achieves the supremum in \eqref{eq:q_star_def}.
\end{corollary}

\begin{remark}
	The optimal PUTs in \eqref{eq:PUT-fLeakKforDist-HDdist-OptValue} and \eqref{eq:PUT-fLeakK-HDdist-OptValue} simplify to finding an output distribution $Q_Y^*$ that can be viewed as a ``target'' distribution, i.e., the optimal mechanism aims to produce this distribution as closely as possible, subject to the utility constraint. In particular, given an input, the optimal mechanism in \eqref{eq:opt_mech} distributes the outputs according to $Q_Y^*$ while conditioning the output to be within a ball of radius $D$ around the input. The optimization in \eqref{eq:q_star_def} ensures that all inputs are uniformly masked while \eqref{eq:PUT-fLeakKforDist-HDdist-OptValue} provides average guarantees. 
	
	Moreover, for any arbitrarily chosen maximal $f$-leakage, the optimal PUT in \eqref{eq:PUT-fLeakK-HDdist-OptValue} leads to the same target distribution $Q^*_Y$ given by \eqref{eq:q_star_def}. Therefore, the corresponding optimal mechanism in \eqref{eq:opt_mech} is independent of the choice of maximal $f$-leakage. As a special case of \eqref{eq:PUT-fLeakK-HDdist-OptValue},
		the optimal tradeoff in \eqref{eq:PUT-MaxL-HDdist-OptValue}, for maximal $\alpha$-leakage with $\alpha>1$, is achieved by the optimal mechanism (in \eqref{eq:opt_mech}) that is no longer depending on $\alpha$. And the optimal tradeoff (in \eqref{eq:PUT-MaxL-HDdist-OptValue}) itself is also independent of the value of $\alpha>1$. 
\end{remark}

\subsection{PUTs for Datasets Containing Non-Sensitive Data}\label{Subsec:PUT_for-HDvsAlphaLK}
For datasets containing both sensitive and non-sensitive data, indicated by $S$ and $X$, respectively, as shown in Fig \ref{fig:Datamodel_2}, the purpose of privacy protection is to limit information leakage of sensitive data while releasing non-sensitive data. We use $\alpha$-leakage from $S$ to $Y$ as the privacy measure, where $Y$ is the released version of $X$. Therefore, with $P_{Y|S,X}$ in the place of $P_{Y|X}$ in \eqref{eq:PUT-HD}, we obtain the optimal PUT as
\begin{equation}\label{eq:PUT_AlphaLeak_HD}
\text{PUT}_{\text{HD},\mcl L_\alpha}(D)=\inf_{P_{Y|S,X}: d(X,Y)\le D}\,\mcl L_\alpha(S; Y).
\end{equation}

The following theorem bounds $\text{PUT}_{\text{HD},\mcl L_\alpha}$ from below. Note that the $B_D$ in the following is the distortion ball defined in \eqref{eq:PUT_HardDist_CollectofFeasibleY}.
\begin{theorem}\label{Thm:PUT_AlphaLeak_HardDist}
	The minimal leakage $\text{PUT}_{\text{HD},\mcl L_\alpha}$ ($1\leq \alpha\leq \infty$) in \eqref{eq:PUT_AlphaLeak_HD} is bounded from below by
	\begin{align}\label{eq:PUT_AlphaLeak_HD_Alpha>1-LB}
	&\text{PUT}_{\text{HD},\mcl L_\alpha} (D)\geq \nonumber\\
	&\begin{cases}
	\sum\limits_{s,x}P(s,x)\log\Big(\max\limits_{\substack{y\in B_D(x)}} \sum\limits_{s'\in\mcl S_D(y)}P(s')\Big)^{-1}, & \hspace{-10pt}\alpha\hspace{-2pt}=\hspace{-2pt}1\\
	\log\sum\limits_{s,x}\frac{P(s)P(s,x)}{\max\limits_s P_S(s)} \Big(\max\limits_{y\in B_D(x)}\sum\limits_{s'\in\mcl S_D(y) }P(s')\Big)^{-1}, & \hspace{-12pt}\alpha\hspace{-2pt}=\hspace{-2pt}\infty\\
	\frac{\alpha}{\alpha-1}\log\sum\limits_{\substack{s,x}}\frac{P(s)^{\alpha}P(x|s)}{\|P_S\|_{\alpha}} \Big(\hspace{-2pt}\max\limits_{y\in B_D(x)}\hspace{-2pt}\sum\limits_{s'\in\mcl S_D(y) }\hspace{-3pt}P(s')^{\alpha}\Big)^{\frac{1-\alpha}{\alpha}}, &\hspace{-3pt}\text{else}\\	
	\end{cases}\nonumber
	\end{align}
	where the set $\mcl S_D(y)$ of $s$ for each $y$ is defined as
	\begin{align}
	\mcl S_D(y)\triangleq\{s:\exists\, x,\,P_{S,X}(s,x)>0,d(x,y)\leq D \}.
	\end{align}
    The lower bound is tight if there exists an privacy mechanism $P_{Y|S,X}\in \mcl P_{Y|S,X}(D)$ such that 
	\begin{itemize}
		\item[(i)] given $(s,x)$, %all $y$ with $P(y|s,x)>0$ have the same $\sum_{\mcl S_D(y)}P(s)$, i.e., 
		for any $y$ with $P(y|s,x)>0$,
		\begin{align}\label{AlphaLK-HD-optMech_alpha-Cond1}
		\sum_{s'\in\mcl S_D(y)}P(s')=\max_{y'\in B_D(s,x)} \sum_{s'\in\mcl S_D(y')}P(s');
		\end{align} 	
		\item[(ii)] given any $y$ with $P_Y(y)>0$, for any $s\in S_D(y)$,
		\begin{align}
			\sum\limits_{x:d(x,y)\leq D}\hspace{-10pt}P(y|s,x)P(x|s)\hspace{-2pt}=\hspace{-3pt}\frac{P_Y(y)}{\sum\limits_{s'\in\mcl S_D(y)}\hspace{-5pt} P(s')}, \label{AlphaLK-HD-optMech_alpha-Cond2}
		\end{align}
		where $P_Y$ is the marginal distribution of $Y$ from the privacy mechanism $P_{Y|S,X}$ and $P_{S,X}$.
	\end{itemize}

\end{theorem}
The proof details are in Appendix \ref{proof:Thm:PUT_AlphaLeak_HardDist}. 

Note that by using maximal $\alpha$-leakage as the privacy measure, the setting for publishing datasets consisting of sensitive and non-sensitive data can be generalized to restrict leakages about \textit{all} functions of the sensitive data. This will be addressed in future work.

\section{Applications: PUTs for Hard and Average Distortion Constraints}\label{Sec:Examples}
In this section, we first illustrate the results of Sec. \ref{Sec:PUT_for-HD} and present the optimal PUTs for two distinct hard distortion functions. Our first choice for hard distortion, restricted to binary datasets, is the absolute distance between the types  (i.e., empirical distributions) of the original and revealed (binary) datasets.
This choice is motivated by the observation that, for any dataset, the type is a sufficient statistic for any function of the dataset that is unaffected by permutation---for example, mean, variance, correlation between two features. Thus, constraining the distortion between the released type and the original type, one can guarantee the utility of the released dataset for a variety of statistical applications. In Example 1 below, we derive the optimal mechanism under this distortion measure for binary datasets.
Our second choice for hard distortion is the Hamming distance between the original and released datasets for discrete alphabets. 
This choice is motivated by the fact that a hard Hamming distortion is more relevant when the order of the entries in the dataset cannot be changed. In Example 2 below, we derive the optimal mechanism under this distortion for a dataset sampled from an arbitrary discrete alphabet.
Note that, as a consequence of Corollary~\ref{Col:PUT_MaximalAlphaLeak_HardDist}, for these examples the optimal PUT and privacy mechanism are the same for all values of $\alpha>1$.

In contrast to hard distortion measures, in Example 3, we study the PUTs that result from using \textit{average} Hamming distortion as the utility measure and  maximal $\alpha$-leakage as the privacy measure. Due to lack of closed-form solutions, we use numerical results to highlight the dependence of both the optimal PUTs and the privacy mechanisms on $\alpha$.

\subsection{Example 1: Binary Datasets with Hard Distortion on Types}\label{Subsection:PUT_HardDist_BinaryTypes}

Let $X^n$ be a random dataset with $n$ entries and $Y^n$ be the corresponding released dataset generated by a privacy mechanism $P_{ Y^n|X^n}$. Entries of both $X^n$ and $Y^n$ are from the same alphabet $\mcl X$. Adopting the notation of \cite[Chapter 11]{IT_Cover}, let $P_{x^n}$ and $P_{y^n}$ indicate the types of input dataset $x^n$ and output dataset $y^n$, respectively. We define the distortion function as the distance between types, given by
\begin{equation}\label{eq:HDonTypes}
	d_{\text{T}}(x^n,y^n)=\max_{x\in\mathcal{X}}|P_{x^n}(x)-P_{y^n}(x)|,
\end{equation}
and therefore, obtain $\text{PUT}_{\text{HD},\mcl L_\alpha^{\text{max}}} $ as in \eqref{eq:PUT-MaxAlphaLeak-U_HardDistortion} but with datasets $X^n,Y^n$ in place of single letters $X,Y$. %, and $D=\frac{m}{n}$. 
Since types of $n$-length sequences take on only values that are multiples of $\frac{1}{n}$, this distortion function $d_{\text{T}}$ takes on values of the form $\frac{m}{n}$, where $m\in[0,n]$.

We concentrate on binary datasets, i.e., $\mcl X=\{0,1\}$.
Note that for binary datasets, we can simply write $d_{\text{T}}(x^n,y^n)=|P_{x^n}(1)-P_{y^n}(1)|$.
For a $n$-length binary dataset, the number of types is $n+1$. Therefore, all input and output datasets can be categorized into $n+1$ type classes defined as  
\begin{align}\label{eq:CollectionSeq_Type}
T(i)\triangleq\{x^n:n P_{x^n}(1)=i\}.
\end{align}

\begin{theorem}\label{Thm:PUTMaxAlphaLk&HardDist-SeqType-Alpha>1}
	For binary datasets and the distortion function in \eqref{eq:HDonTypes}, given integers $n,m$ where $0\le m\le n$, the optimal tradeoff for $\alpha>1$ is
	\begin{align}
	\label{eq:PUT-MaxAlphaLeak-UvsHardType_OptProbExample}
	\text{PUT}_{\text{HD},\mcl L_\alpha^{\text{max}}} \left(\frac{m}{n}\right)&=\min_{\substack{P_{Y^n|X^n}:\\d_{\text{T}}(X^n,Y^n)\leq \frac{m}{n}}} \,\mcl L_{\alpha}^{\text{max}}(X^n\to Y^n)\\
	\label{eq:PUT-MaxAlphaLeak-UvsHardType_OptProbExample-MinLeak}
	&=\log \mathsmaller{\ceil*{\frac{n+1}{2m+1}}}.
	\end{align}
	An optimal privacy mechanism maps all input datasets in \textit{a type class} to a unique output \textit{dataset} which is feasible and belongs to a type class in the set $\mcl T^{\star}$ given by
	\begin{align}\label{eq:Opt-OutTypeSet}
	\mathsmaller{\mcl T^{\star}\triangleq\Big\{T(j): j\hspace{-2pt}=\hspace{-2pt}l\hspace{-2pt}+\hspace{-2pt}(2m\hspace{-2pt}+\hspace{-2pt}1)k, k\in\mathsmaller{\left[0,\ceil*{\frac{n+1}{2m+1}}\hspace{-1pt}-\hspace{-1pt}1\right]}\Big\}},
	\end{align}
where $l=m$ if $\ceil{\frac{n+1}{2m+1}}-\frac{n+1}{2m+1}\leq \frac{m}{2m+1}$, and otherwise, $l=n-\left(\ceil{\frac{n+1}{2m+1}}-1\right)(2m+1)$.
\end{theorem}
\begin{figure}[t]
	\centering
	\includegraphics[width=3.5 in]{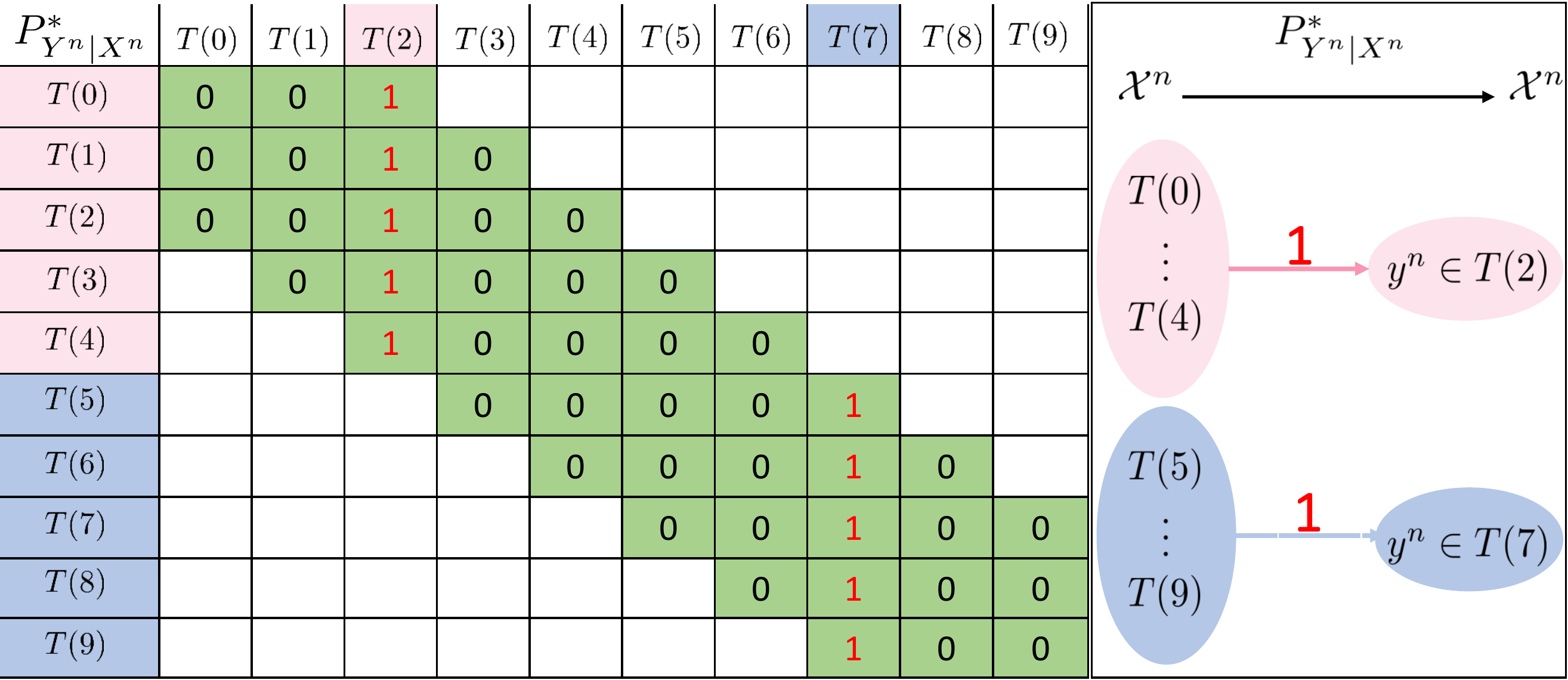}
	\caption{An optimal mechanism $\text{PUT}_{\text{HD},\mcl L_\alpha^{\text{max}}} \left(\frac{m}{n}\right)$ for $\alpha>1$ with $(n,m)=(9,2)$, where rows and columns are types of $X^n$ and $Y^n$, respectively. Note that the hard distortion forces conditional probabilities of outputs outside the feasible ball of given input to be zero. We highlight the conditional probabilities of feasible outputs in green, and give their values in the optimal mechanism.}
	\label{fig:Example1_Channel_1}
\end{figure}
A detailed proof is in Appendix \ref{Proof:Thm:PUTMaxAlphaLk&HardDist-SeqType-Alpha>1}. 
Note that for any $x^n$ in the type class $T(i)$, the corresponding output $y^n$ generated by the optimal mechanism $P^*_{Y^n|X^n}$ is unique and belongs to the unique type class in $\mcl T^{\star}\cap\{T(j): |i-j|\leq m\}$.
For example, if $(n,m)=(9,2)$, then from Thm. \ref{Thm:PUTMaxAlphaLk&HardDist-SeqType-Alpha>1}, we have $\text{PUT}_{\text{HD},\mcl L_\alpha^{\text{max}}}(\frac{2}{9})=1$ bit and $\mcl T^{\star}=\{T(2),T(7)\}$. Fig. \ref{fig:Example1_Channel_1} shows the optimal mechanism, which maps all input datasets in $\{T(i):i\in[0,4]\}$ (resp. $\{T(i):i\in[5,9]\}$) to a \textit{unique output dataset} in $T(2)$ (resp. $T(7)$) with probability $1$.

\subsection{Example 2: Hard Hamming Distortion on Datasets}\label{Subsection:Exp-PUT_HardHammingDist_Sequences}
In the example, we consider hard Hamming distortion on datasets with entries from general finite alphabets. 
Formally, for datasets $x^n,y^n\in \mcl X^n$, we define the Hamming distortion function as
\begin{align}\label{eq:HamingHD_data sets}
	d_{\text{H}}(x^n,y^n)=\frac{1}{n}\sum_{i=1}^{n}\mathbf{1}(x_{i}\neq y_{i}).
\end{align}
Therefore, we obtain $\text{PUT}_{\text{HD},\mcl L_\alpha^{\text{max}}} $ as in \eqref{eq:PUT-MaxAlphaLeak-U_HardDistortion} but with datasets $X^n,Y^n$ in place of single letters $X,Y$.
\begin{theorem}\label{thm:OptPUT-MaxAlphaLeak-UvsHardHamingDist}
	For datasets from a finite alphabet $\mcl X$ and Hamming distortion function, for any integers $n,m$ where $0\le m\le n$, the optimal tradeoff for $\alpha>1$ is
	\begin{align}
	\label{eq:PUT-MaxAlphaLeak-UvsHardHamingDist_OptProbExample}
		\text{PUT}_{\text{HD},\mcl L_\alpha^{\text{max}}}\left(\frac{m}{n}\right)&=\min_{\substack{P_{Y^n|X^n}:\\d_{\text{H}}(x^n,y^n)\leq \frac{m}{n}}}\,\mcl L_{\alpha}^{\text{max}}(X^n\to Y^n)\\
		\label{eq:PUT-MaxAlphaLeak-UvsHardHamingDist_OptProbExample-MinLeak}
		&=\log \frac{\left|\mcl X\right|^n}{\sum_{i=0}^{m}{n \choose i}\left(|\mcl X|-1\right)^i}.
	\end{align}
	An optimal privacy mechanism maps each input $x^n\in \mcl X^n$ \textbf{uniformly} to every feasible output, i.e., for all $x^n,y^n$ where $d_H(x^n,y^n)\le \frac{m}{n}$, $P_{Y^n|X^n}(y^n|x^n)=\frac{1}{\sum_{i=0}^{m}{n \choose i}\left(|\mcl X|-1\right)^i}$.
\end{theorem}
Note that for any pair of $x^n$ and $y^n$, the optimal mechanism $P^*_{Y^n|X^n}(y^n|x^n)$ is the average probability of $y^n$ when the support of $Y^n$ is restricted to $B_D(x_n)$, i.e., $Y^n$ takes values from $\{y^n: d_{\text{H}}(x^n,y^n)\leq \frac{m}{n}\}$.
The key observation to reach the conclusion in Thm. \ref{thm:OptPUT-MaxAlphaLeak-UvsHardHamingDist} is that every output dataset is in the same number of feasible balls, such that a uniform distribution over the output space leads to equal probability for the feasible ball of each input dataset. The proof details are in Appendix \ref{proof:thm:OptPUT-MaxAlphaLeak-UvsHardHamingDist}. Fig. \ref{fig:PUT-MaxAlphaLeak-UvsHardHamingDist_Mechanism} illustrates the optimal mechanism in Thm. \ref{thm:OptPUT-MaxAlphaLeak-UvsHardHamingDist} for $\mcl X=\{0,1,2\}$ and $(n,m)=(2,1)$.\\
\begin{figure}
	\centering
	\includegraphics[width=3.5in]{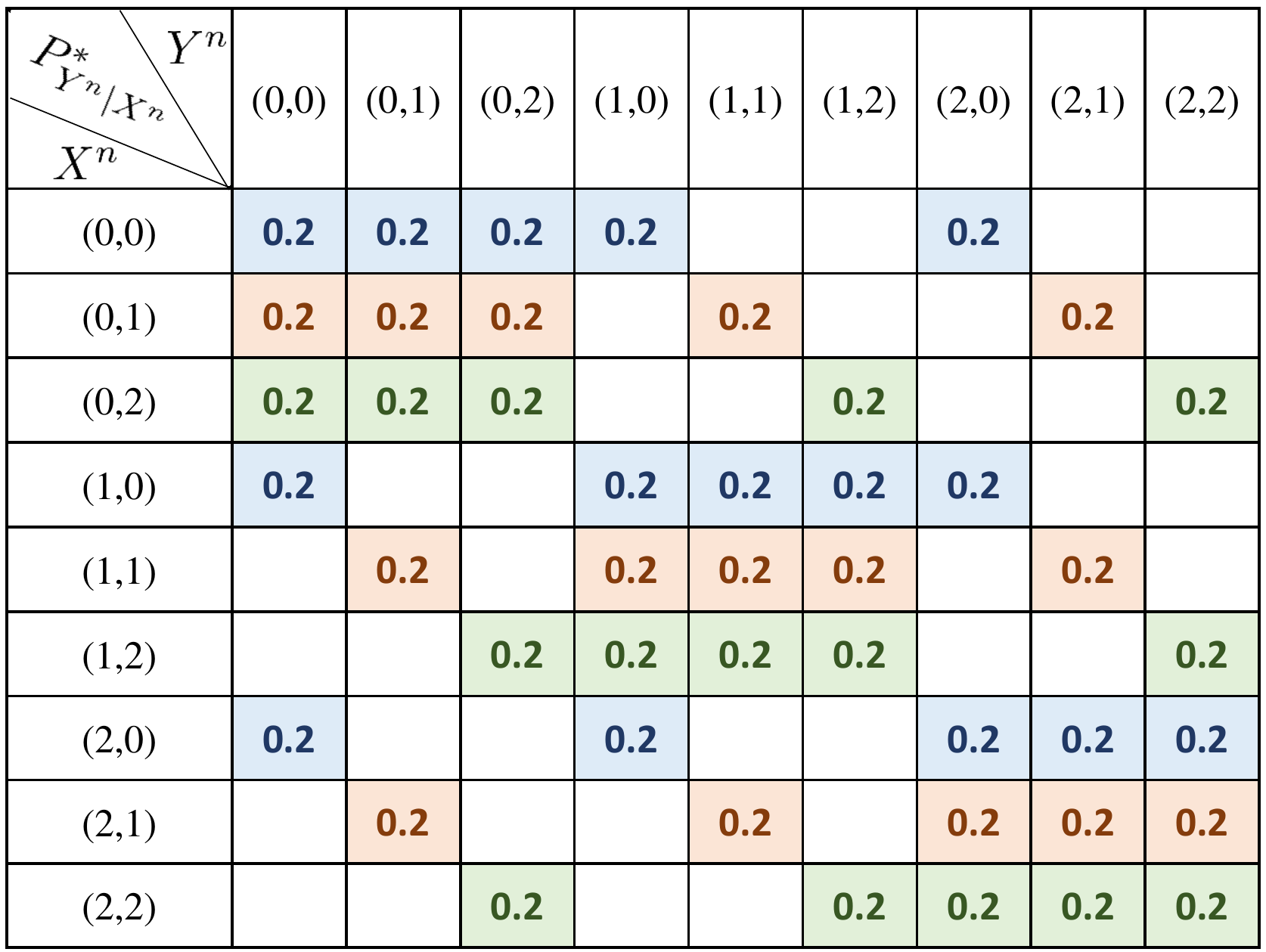}
	\caption{An optimal mechanism of \eqref{eq:PUT-MaxAlphaLeak-UvsHardHamingDist_OptProbExample} for $\alpha>1$ with $(n,m)=(2,1)$ and $\mcl X=\{0,1,2\}$ where rows and columns are $x^n$ and $y^n$, respectively. Note that we color the conditional probabilities of feasible outputs (respect to the hard Hamming distortion) and their values are the same as $0.2$ in the optimal mechanism.}
	\label{fig:PUT-MaxAlphaLeak-UvsHardHamingDist_Mechanism}
\end{figure}
Note that permuting items of a dataset does not change the type but will lead to a non-zero Hamming distortion. 
The distortion on types in \eqref{eq:HDonTypes} can be viewed as a relaxation of the Hamming distortion, in the sense that
 the set of feasible privacy mechanisms in \eqref{eq:PUT-MaxAlphaLeak-UvsHardHamingDist_OptProbExample} belongs to that in \eqref{eq:PUT-MaxAlphaLeak-UvsHardType_OptProbExample}, i.e., 
\begin{align}
	\left\{P_{Y^n|X^n}: d_{\text{H}}(x^n,y^n)\leq \frac{m}{n} \right\}
	 \subset \left\{P_{Y^n|X^n}: d_{\text{T}}(x^n,y^n)\leq \frac{m}{n}\right\}.\nonumber
\end{align}
Therefore, for non-binary alphabets, the result in Thm. \ref{thm:OptPUT-MaxAlphaLeak-UvsHardHamingDist} limits the minimal leakage in \eqref{eq:PUT-MaxAlphaLeak-UvsHardType_OptProbExample}. %for the hard distortion on types (in \eqref{eq:HDonTypes}). 

\subsection{Example 3: Average Hamming Distortion on Binary Alphabet}\label{Section:Discussion}

 \begin{figure*}
	\centering
	\subfloat[]{\includegraphics[width=3.5 in]{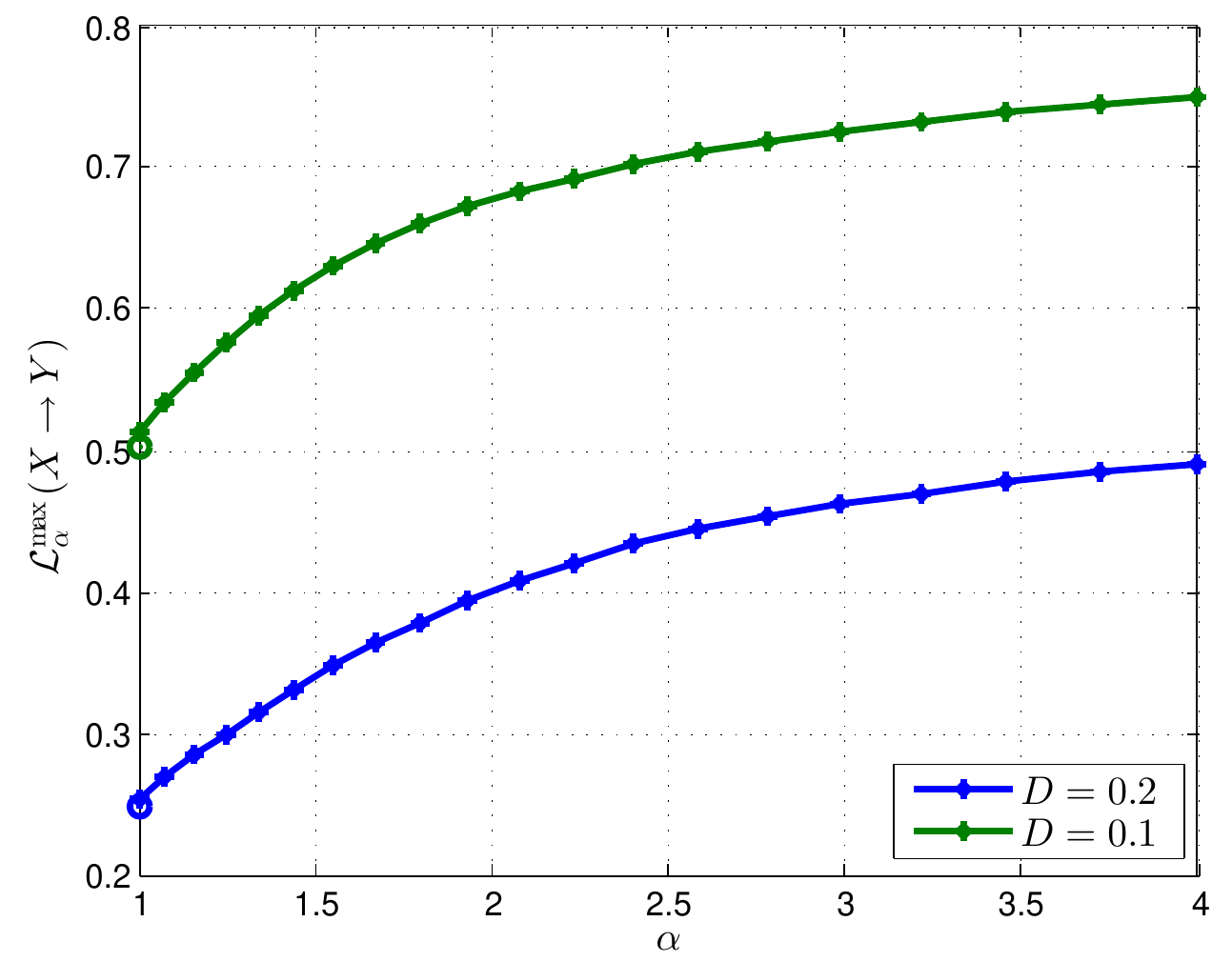}
		\label{fig:PUT-MaxAlphaLeak-BiHam}}
	\subfloat[]{\includegraphics[width=3.5 in]{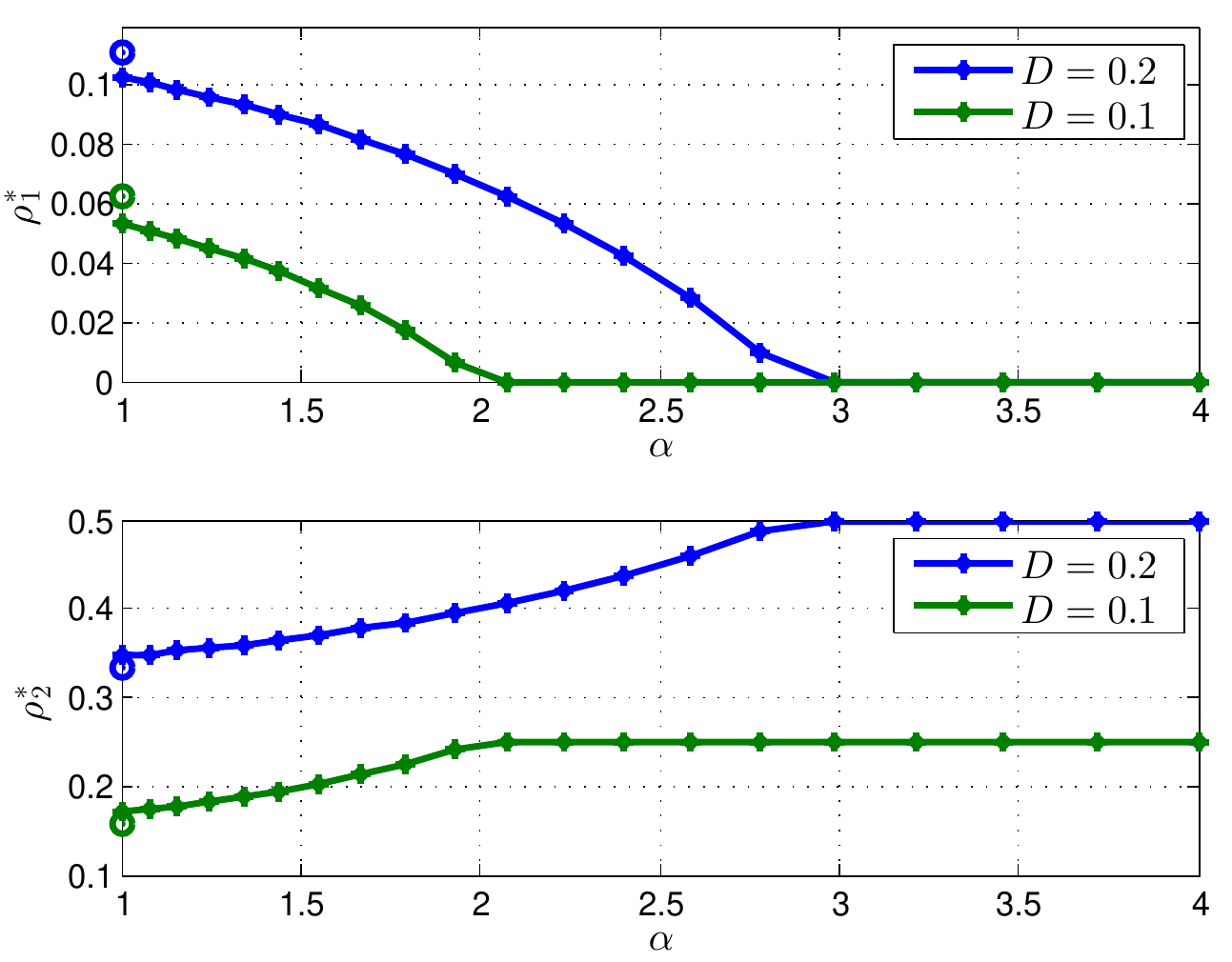}
		\label{fig:PUT-MaxAlphaLeak-BiHam-Mech}}
	\caption{Numerical results for the privacy-utility tradeoff in \eqref{eq:PUT-MaxAlpLK_AvgHamD-Binary} with $p=0.4$ and $D\in\{0.2,0.1\}$. Figure \ref{fig:PUT-MaxAlphaLeak-BiHam} plots the minimal values of maximal $\alpha$-leakage as a function of $\alpha$ (circles indicate $\alpha=1$ and stars are for $1.001\leq \alpha\leq 4$). Figure \ref{fig:PUT-MaxAlphaLeak-BiHam-Mech} illustrates the behavior of the crossover probabilities $\rho_1^*$ and $\rho_2^*$ for the optimal privacy mechanisms as a function of $\alpha$.}	 
	\label{fig:PUT-MaxAlphaLeak-BiHam-All}
\end{figure*}

We consider a PUT setting with maximal $\alpha$-leakage as the privacy measure and average Hamming distortion as the distortion constraint. Such an average utility constraint can
be relevant to data publishing settings where preserving statistics of the dataset is desired.
This example also illustrates that, in contrast to the hard distortion constraint, the optimal mechanism may depend on $\alpha$.
	
Consider the following PUT problem that minimizes maximal $\alpha$-leakage subject to the average Hamming distortion constraint:
	\begin{subequations}\label{eq:PUT-MaxAlpLK_AvgHamD}
		\begin{align}
		\min_{P_{Y|X}} \quad & \mcl L_{\alpha}^{\text{max}}(X\to Y)\\
		\text{s.t.,} \quad & \sum_{x,y\in \mcl X}P_{X,Y}(x,y)\mathbf{1}\left(y\neq x\right)\leq D
		\end{align}
	\end{subequations}
where $0<D<1-\max_x P_{X}(x)$ is the maximum permitted average Hamming distortion. We focus on the binary case: let $X,Y\in \{0,1\}$ where $X$ follow the Bernoulli distribution $\text{Bern}(p)$ ($0<p<1$), i.e., $P_X(1)=p$. We represent the privacy mechanism $P_{Y|X}$ via the two crossover probabilities $P_{Y|X}(1|0)=\rho_1$ and $P_{Y|X}(0|1)=\rho_2$. By solving the supremum in the expression of maximal $\alpha$-leakage, for $\alpha>1$, the optimization in \eqref{eq:PUT-MaxAlpLK_AvgHamD} can be written as
\begin{subequations}\label{eq:PUT-MaxAlpLK_AvgHamD-Binary}
		\begin{IEEEeqnarray}{l l}
		\min_{\rho_1,\rho_2} \,\, & \frac{1}{\alpha-1}\log\hspace{-2pt}\Big((1-\rho_1)^{\alpha}(1-\rho_2)^{\alpha}-\left(\rho_1\rho_2\right)^{\alpha}\Big)\hspace{-2pt}+\log\hspace{-2pt}\Big(\quad\quad\nonumber\\
		&\,\big((1-\rho_1)^{\alpha}-\rho_2^{\alpha}\big)^{\frac{1}{1-\alpha}}+\big( (1-\rho_2)^{\alpha}-\rho_1^{\alpha} \big)^{\frac{1}{1-\alpha}}\Big)\\
		\text{s.t.} \, & (1-p)\rho_1+p\rho_2\leq D.
		\end{IEEEeqnarray}
\end{subequations}

Fig.~\ref{fig:PUT-MaxAlphaLeak-BiHam-All} shows the optimal values and mechanisms in \eqref{eq:PUT-MaxAlpLK_AvgHamD-Binary} for $p=0.4$ and $D=0.2$ or $D=0.1$. From the plots, we can see that for $\alpha=1.001$, the optimal mechanism $P_{Y|X}^*$ (represented by $\rho_1^*$ and $\rho_2^*$) is slightly different from that of mutual information \cite[Figure 10.3]{IT_Cover} due to the fact that as $\alpha$ tends to $1$, the limit of maximal $\alpha$-leakage is Shannon channel capacity instead of mutual information, i.e., $\lim_{\alpha\to 1} \mcl L^{\text{max}}_\alpha(X\to Y)=\lim_{\alpha\to 1} \sup_{P_{\tilde{X}}}I^{\text{A}}_\alpha(\tilde{X};Y)=\sup_{P_{\tilde{X}}}I(\tilde{X};Y)$. We also observe that as $\alpha$ grows, the optimal crossover probabilities $\rho_1^*$ and $\rho_2^*$ gradually approach to $0$ and $\frac{D}{p}$, respectively. Therefore, for the PUT in \eqref{eq:PUT-MaxAlpLK_AvgHamD-Binary}, maximal $\alpha$-leakage with different values of $1<\alpha<\infty$ leads to various optimal privacy mechanisms, which can differ from that for either $\alpha=1$ or $\alpha=\infty$.

It is not difficult to check that the optimal privacy mechanisms (in Fig. \ref{fig:PUT-MaxAlphaLeak-BiHam-Mech}) for different values of $\alpha$ give the same probability of correctly guessing, defined as $\sum_y P_Y(y)\max_x P_{Y|X}(y|x)$ \cite{privacyGuessing_asoodeh2017}, which equals to $1-D$ in this example. Probability of correctly guessing is, in fact, the average accuracy of estimating the value of original data $X$ from $Y$ when the maximal posterior (MAP) estimator is used. Therefore, if the original data $X$ is released against an adversary who is only interested in the most likely value of $X$, all values of $\alpha$ will lead to the same privacy guarantee in the sense that the optimal mechanisms give the same average accuracy of estimation.

\section{Conclusion}
Via $\alpha$-loss ($1\leq \alpha\leq \infty$), we have defined two tunable measures of information leakage: $\alpha$-leakage for a specific function of original data, and maximal $\alpha$-leakage for any arbitrary function of original data, and proven that: (i) $\alpha$-leakage equals to Arimoto mutual information for $1\leq \alpha\leq \infty$; (ii) for $\alpha>1$, maximal $\alpha$-leakage equals to Arimoto channel capacity; and for $\alpha=1$ and $\alpha=\infty$ it simplifies to mutual information and maximal leakage, respectively. From properties of Arimoto mutual information, $\alpha$-leakage is known to be quasi-convex in the conditional distribution and satisfy the post-processing inequality. For maximal $\alpha$-leakage, we have proven that it is quasi-convex in the conditional distribution, and satisfies data processing inequalities as well as a composition property. 
	
In the context of privacy-guaranteed data publishing, we have explored PUT problems for the proposed tunable leakage measures and hard distortion utility constraints. This utility constraint has the advantage that it allows the data curator/provider to make specific, deterministic guarantees on the quality of the released dataset. For maximal $\alpha$-leakage, we have shown that: (i) for all $\alpha>1$, we obtain the same optimal privacy mechanism and optimal PUT, both of which are independent of the distribution of the original data; (ii) for $\alpha =1$, the optimal mechanism differs and depends on the distribution of the original data. In other words, for this hard distortion measure, maximal $\alpha$-leakage behaves as either mutual information or maximal leakage. We have also demonstrated that this extremal behavior may not hold when the hard distortion constraint is replaced by an average distortion constraint (e.g., average Hamming distortion) and the source alphabet is binary. Future directions include studying PUT problems with average distortion constraints for non-binary alphabets to further explore the impact of $\alpha$ on the design of privacy mechanisms.

\allowdisplaybreaks
\appendices
\renewcommand{\thesectiondis}[2]{\Alph{section}:}
\section{Proof of Lemma \ref{lem:Minimal-expectedAlphaLoss}}\label{Proof:lem:Minimal-expectedAlphaLoss}
   \begin{proof}[\nopunct]
   	For $1<\alpha<\infty$, the minimal expected value of the $\alpha$-loss in Definition \ref{Def:alpha-loss} can be expressed as
   	\begin{IEEEeqnarray}{l  l}
   		&\min_{P_{\hat{X}|Y}}\mathbb{E}\left[\ell_{\alpha}(X,Y,P_{\hat{X}|Y})\right]\nonumber\\
   		=&\min_{P_{\hat{X}|Y}}\frac{\alpha}{\alpha-1}\left(1-\sum\limits_{x,y}P_{X,Y}(x,y)P_{\hat{X}|Y}(x|y)^{\frac{\alpha-1}{\alpha}}\right)\label{eq:Max_AlphaReward_1}\\
   		=&\frac{\alpha}{\alpha-1}\left(1-\max_{P_{\hat{X}|Y}}\sum\limits_{x,y}P_{X,Y}(x,y)P_{\hat{X}|Y}(x|y)^{\frac{\alpha-1}{\alpha}}\right)\label{eq:Max_AlphaReward_2}\\
   		=&\frac{\alpha}{\alpha-1}\hspace{-3pt}\left(\hspace{-2pt}1\hspace{-2pt}-\hspace{-2pt}\sum\limits_{y}\hspace{-2pt}P(y)\max_{P_{\hat{X}|Y=y}}\hspace{-2pt}\sum_x \hspace{-2pt}P(x|y)P_{\hat{X}|Y}(x|y)^{\frac{\alpha-1}{\alpha}}\hspace{-2pt}\right)\hspace{-2pt}.\quad\label{eq:Max_AlphaReward_3}
   	\end{IEEEeqnarray}
   	For each $y$ with $P_{Y}(y)>0$, the maximization in \eqref{eq:Max_AlphaReward_3} can be explicitly written as
   	\begin{subequations}\label{eq:Max_AlphaReward}
   		\begin{align}
   		\label{eq:Max_AlphaReward_obj}
   		\max_{\substack{P_{\hat{X}|Y=y}}}\quad &\sum_{x\in\mcl X}P_{X|Y}(x|y)P_{\hat{X}|Y}(x|y)^{\frac{\alpha-1}{\alpha}}\\
   		\label{eq:Max_AlphaReward_const1}
   		\text{s.t.}\quad & \sum_{x\in\mcl X}P_{\hat{X}|Y}(x|y)=1\\
   		\label{eq:Max_AlphaReward_const>0}
   		& P_{\hat{X}|Y}(x|y)\geq 0 \quad \text{  for all }x\in \mcl{X} .
   		\end{align}  	
   	 \end{subequations}
   		For $1\leq \alpha\leq \infty$, the exponent $\frac{\alpha-1}{\alpha}\geq 0$ such that the problem in \eqref{eq:Max_AlphaReward} is a convex program. Therefore, by using Karush-Kuhn-Tucker (KKT) conditions \cite[Chapter 5.5.3]{boydconvex}, we obtain the optimal value of \eqref{eq:Max_AlphaReward} as 
   		\begin{IEEEeqnarray}{l}
   		\hspace{-20pt} \max_{P_{\hat{X}|Y=y}} \sum_{x} P_{X|Y} (x|y)P_{\hat{X}|Y}(x|y)^{\frac{\alpha-1}{\alpha}}=\|P_{X|Y}(\cdot|y)\|_{\alpha}
   		\end{IEEEeqnarray}
   		with the optimal solution $P^{\star}_{\hat{X}|Y}$ as
   		\begin{align}
   		\label{eq:MinExpAlphaLoss-OptStategy-inPf}
   		P^{\star}_{\hat{X}|Y}(x|y)=\frac{P_{X|Y}(x|y)^{\alpha}}{\sum_{x\in\mcl X}P_{X|Y}(x|y)^{\alpha}}\quad \text{for all } x\in\mcl X.
   		\end{align}	
   	For $\alpha=1$, the optimal solution is $P^{\star}_{\hat{X}|Y}=P_{X|Y}$. For $\alpha=\infty$, we have
   	\begin{IEEEeqnarray}{r l}
   	\lim_{\alpha\to \infty}P^{\star}_{\hat{X}|Y}(x|y)
   	=&\lim_{\alpha\to \infty} \frac{\left(\frac{P_{X|Y}(x|y)}{\max_x P_{X|Y}(x|y)}\right)^{\alpha}}{\sum\limits_{x\in\mcl X}\left(\frac{P_{X|Y}(x|y)}{\max_x P_{X|Y}(x|y)}\right)^{\alpha}}\\
   	=&\begin{cases}
   	\frac{1}{k(y)},&x=\arg\max_x P_{X|Y}(x|y)\\
   	0,&\text{otherwise},
   	\end{cases}\quad
   	\end{IEEEeqnarray}
   	where the integer $k(y)$ indicates the cardinality of the set $\{x:x=\arg\max_x P_{X|Y}(x|y)\}$. \\
   	Applying the optimal solution $P^{\star}_{\hat{X}|Y}$ to \eqref{eq:Max_AlphaReward_3}, we have
   	\begin{align}
   	&\min_{P_{\hat{X}|Y}}\mathbb{E}\left[\ell_{\alpha}(X,Y,P_{\hat{X}|Y})\right] \nonumber\\
   	=&\begin{cases}
   	\frac{\alpha}{\alpha-1}\left(1-\sum\limits_{y}\|P_{X,Y}(Xy)\|_{\alpha}\right),& \alpha>1\\
   	\sum\limits_{x,y}P_{X,Y}(x,y)\log\frac{1}{P_{X|Y}(x|y)}, & \alpha=1
   	\end{cases},\\
   	=&\begin{cases}
   	\frac{\alpha}{\alpha-1}\left(1-\exp\left(\frac{1-\alpha}{\alpha}H_{\alpha}^{\text{A}}(X|Y)\right)\right),& \alpha>1\\
   	H(X|Y), & \alpha=1
   	\end{cases}.
   	\end{align}
   \end{proof}
   	
	\section{Proof of Theorem \ref{Thm:DefEquialentExpression_alphaleakage}}\label{Proof:DefEquialentExpression_alphaleakage}
	\begin{proof}[\nopunct]
		The expression \eqref{eq:alphaLeak_definition} can be explicitly written as
		\begin{IEEEeqnarray}{l l}	
			&\mcl L_{\alpha}(X\to  Y)\nonumber\\
			=& \lim_{\alpha'\to \alpha}\frac{\alpha'}{\alpha'-1}\log\hspace{-2pt}\left(\hspace{-2pt}\frac{\max\limits_{P_{\hat{X}|Y}}\sum\limits_{x,y}P_{X,Y}(x,y)P_{\hat{X}|Y}(x|y)^{\frac{\alpha'-1}{\alpha'}}}{\max\limits_{P_{\hat{X}}}\sum\limits_{x}P_X(x)P_{\hat{X}}(x)^{\frac{\alpha'-1}{\alpha'}}}\hspace{-2pt}\right).\label{eq:GealLeak_definition1}
		\end{IEEEeqnarray}
		To simplify the expression in \eqref{eq:GealLeak_definition1}, we need to solve the two maximizations in the logarithm. From \eqref{eq:MaxAlphaLK-AlphaLoss}, we know that to solve the maximization in the numerator equals to find the minimal expected $\alpha$-loss. Making use of the result in Lemma \ref{lem:Minimal-expectedAlphaLoss}, we have that for $\alpha'\in(1,\infty)$,
		\begin{IEEEeqnarray}{l}	
		\hspace{-15pt}	\max_{P_{\hat{X}|Y}}\hspace{-1pt}\sum_{\hspace{-1pt}x,y\hspace{-1pt}}\hspace{-1pt}P_{X,Y}\hspace{-1pt}(\hspace{-0.5pt}x,\hspace{-1pt}y\hspace{-0.5pt})P_{\hat{X}|Y}\hspace{-1pt}(\hspace{-0.5pt}x|y\hspace{-0.5pt})^{\hspace{-2pt}\frac{\alpha'\hspace{-0.5pt}-\hspace{-0.5pt}1}{\alpha'}}
			\hspace{-3pt}=\hspace{-2pt}\exp\hspace{-2pt}\bigg(\hspace{-1.5pt}\frac{1\hspace{-2pt}-\hspace{-2pt}\alpha'}{\alpha'}\hspace{-1pt}H_{\hspace{-0.5pt}\alpha'}^{\hspace{-0.5pt}\text{A}}\hspace{-1pt}(\hspace{-0.5pt}X|Y\hspace{-0.5pt})\hspace{-3.5pt}\bigg)\hspace{-2pt}.
		\end{IEEEeqnarray}
		Similarly, by applying KKT conditions to the maximization in the denominator, we have that for $\alpha'\in(1,\infty)$
		\begin{align}
					\max_{P_{\hat{X}}}\sum_{x\in\mathcal{X}}P_X(x)P_{\hat{X}}(x)^{\frac{\alpha'-1}{\alpha'}}
					=\exp\left(\frac{1-\alpha'}{\alpha'}H_{\alpha'}(X)\right).
		\end{align}
		Therefore, we have for $\alpha'\in(1,\infty)$
	  \begin{align}	
			\label{eq:GealLeak_EquivalentInproof0}
	    &\mcl L_{\alpha}(X\to Y)\nonumber\\
	   =& \frac{\alpha'}{\alpha'-1}\log \exp\Bigg(\frac{1-\alpha'}{\alpha'}\bigg(H_{\alpha'}^{\text{A}}(X|Y)-H_{\alpha'}(X)\bigg)\Bigg)\\ 
	   =& I^{\text{A}}_{\alpha'}(X;Y).
	\end{align}
		From the continuous extensions of Arimoto MI for $\alpha=1$ and $\infty$, respectively, we have that for $1\leq \alpha\leq \infty$, $\alpha$-leakage equals to Arimoto MI.\\
	\end{proof}

	\section{Proof of Theorem \ref{Thm:DefEquialentExpression}}\label{Proof:DefEquialentExpression}
	\begin{proof}[\nopunct]
		From Thm. \ref{Thm:DefEquialentExpression_alphaleakage}, we have for $1\leq \alpha\leq \infty$,
		\begin{align}
			\label{eq:Inproof_Thm:DefEquialentExpression}
			\mcl L_{\alpha}^{\text{max}}(X\to Y)=\sup_{U- X- Y }I_{\alpha}^{\text{A}}(U;Y).
		\end{align}
		If $\alpha=1$, we have 
		\begin{align}
		\label{eq:Inproof_Thm:DefEquialentExpression_alpha1_UB}
			\mcl L_{1}^{\text{max}}(X\to Y)=\sup_{U- X- Y }I(U;Y)\leq I(X;Y)
		\end{align}
		where the inequality is from data processing inequalities of MI \cite[Thm 2.8.1]{IT_Cover}. We then prove that the upper bound $I(X;Y)$ in \eqref{eq:Inproof_Thm:DefEquialentExpression_alpha1_UB} can be achieved. Let $U$ be a function of $X$ satisfying $H(X|U)=0$. From the condition $H(X|U)=0$ and the Markov chain $U-X-Y$, we have
		\begin{IEEEeqnarray}{l l}
			H(X,Y|U)=&H(X|U)+H(Y|X,U)=0+H(Y|X),\quad\nonumber\\
			H(X,Y|U)=&H(Y|U)+H(X|Y,U)=H(Y|U)+0,\nonumber
		\end{IEEEeqnarray}
		i.e., $H(Y|X)=H(Y|U)$. Therefore, for a function $U$ satisfying $H(X|U)=0$, there is $\mathcal{L}_1(U\to Y)=I(U;Y)=I(X;Y)$. 
		
		If $\alpha=\infty$,  we have
		\begin{IEEEeqnarray}{l}
		\hspace{-15pt}	\mcl L_{\infty}^{\text{max}}(X\to Y)=\sup_{U- X- Y }\log\frac{\sum\limits_{y} P_Y(y)\max\limits_{u} P_{U|Y}(u|y)}{\max\limits_{u} P_U(u)},
		\end{IEEEeqnarray}
		which is exactly the expression of MaxL, and therefore, we have that for $\alpha=\infty$, the maximal $\alpha$-leakage equals to the Sibson MI of order $\infty$ \cite[Thm. 1]{OperationalLeak_issa2018}, i.e., 
		\begin{align}
			\mcl L_{\infty}^{\text{max}}(X\to Y)=\log\sum\limits_{y} \max\limits_{x} P_{Y|X}(y|x).
		\end{align}
		For $\alpha\in(1,\infty)$, we provide an upper bound for $\mcl L_{\alpha}^{\text{max}}(X\to Y)$, and then, give an achievable scheme as follows.	\\
		\textbf{Upper Bound}: %\textbf{Converse}: 
		We have an upper bound of $\mcl L_{\alpha}^{\text{max}}(X\to Y)$ as
		\begin{IEEEeqnarray}{r  l}
				\label{eq:GealLeak_EquivalentInproofConverse0}
				\mcl L_{\alpha}^{\text{max}}(X\to Y)
				=&\sup_{U- X- Y}I_{\alpha}^{\text{A}}(U;Y)\\
				=&\sup_{\substack{P_{Y,\tilde{X}|\tilde{U}}:P_{\tilde{X}}=P_{X}\\
						P_{Y|\tilde{X},\tilde{U}}=P_{Y|X} }} \sup_{P_{\tilde{U}}}I_{\alpha}^{\text{A}}(\tilde{U};Y)
				\label{eq:GealLeak_EquivalentInproofConverse1}\\
				\label{eq:GealLeak_EquivalentInproofConverse2}
				\leq &\sup_{\substack{P_{Y,\tilde{X}|\tilde{U}}:P_{\tilde{X}|\tilde{U}}\left(\cdot|u\right)\ll P_{X}\\
				P_{Y|\tilde{X},\tilde{U}}=P_{Y|X}}} \sup_{P_{\tilde{U}}} I_{\alpha}^{\text{A}}(\tilde{U};Y)\\
				\label{eq:GealLeak_EquivalentInproofConverse3}
				= &\sup_{\substack{P_{Y,\tilde{X}|\tilde{U}}:P_{\tilde{X}|\tilde{U}}\left(\cdot|u\right)\ll P_{X}\\
						P_{Y|\tilde{X},\tilde{U}}=P_{Y|X}}} \sup_{P_{\tilde{U}}} I_{\alpha}^{\text{S}}(\tilde{U};Y)\\
				\label{eq:GealLeak_EquivalentInproofConverse4}
				= &\sup_{P_{\tilde{X}}\ll P_{X}} I_{\alpha}^{\text{S}}(\tilde{X};Y)\\
				\label{eq:GealLeak_EquivalentInproofConverse5}
				= &\sup_{P_{\tilde{X}}\ll P_{X}} I_{\alpha}^{\text{A}}(\tilde{X};Y)
			\end{IEEEeqnarray}
		where $P_{\tilde{X}}\ll P_{X}$ indicate that the support of $P_{\tilde{X}}$ is a subset of the support of $P_{X}$\footnote{Note that any set is also the subset of itself, such that the support of $\tilde{X}$ can be the same as that of $X$}. 
		In \eqref{eq:GealLeak_EquivalentInproofConverse1}, $\tilde{U}-\tilde{X}-Y$ forms a Markov chain and the probability distribution of $\tilde{X}$ is constrained to be $P_X$. The upper bound in \eqref{eq:GealLeak_EquivalentInproofConverse2} results from allowing $\tilde{X}$ to be distributed arbitrarily over the support of $X$.
		The equations in \eqref{eq:GealLeak_EquivalentInproofConverse3} and \eqref{eq:GealLeak_EquivalentInproofConverse5} result from that Arimoto MI and Sibson MI of order $\alpha>0$ have the same supremum \cite[Thm. 5]{alphaMI_verdu}, which can be proved from the expressions of Arimoto and Sibson MIs as follows:
		\begin{IEEEeqnarray}{r  l}
			&\sup_{P_{\tilde{U}}} I^{\text{A}}_{\alpha}(\tilde{U};Y)\nonumber\\
			= &\sup_{P_{\tilde{U}}} \frac{\alpha}{\alpha-1}\log \frac{\sum\limits_y\Big(\sum\limits_u P_{\tilde{U},Y}(u,y)^{\alpha}\Big)^{\frac{1}{\alpha}}}{\Big(\sum\limits_u P_{\tilde{U}}(u)^\alpha\Big)^{\frac{1}{\alpha}}} \\
			=&\sup_{P_{\tilde{U}}} \frac{\alpha}{\alpha-1}\log \sum\limits_y\Bigg(\sum\limits_u\frac{P_{\tilde{U}}(u)^{\alpha}}{\sum\limits_u P_{\tilde{U}}(u)^\alpha} P_{Y|\tilde{U}}(y|u)^{\alpha}\Bigg)^{\frac{1}{\alpha}} \\
			=&\sup_{P_{\tilde{U}'}} \frac{\alpha}{\alpha-1}\log \sum\limits_y\Bigg(\sum\limits_u P_{\tilde{U}'}(u) P_{Y|\tilde{U}}(y|u)^{\alpha}\Bigg)^{\frac{1}{\alpha}} \\
			=& \sup_{P_{\tilde{U}'}} I^{\text{A}}_{\alpha}(\tilde{U}';Y) 
		\end{IEEEeqnarray}
		where $P_{\tilde{U}}$ and $P_{\tilde{U}'}$ are probability distributions over the same support and for each $u$, $P_{\tilde{U}'}(u)=\frac{P_{\tilde{U}}(u)^{\alpha}}{\sum\limits_u P_{\tilde{U}}(u)^\alpha}$.
		From the data processing inequalities of Sibson MI for the Markov chain $\tilde{U}-\tilde{X}-Y$, we have that $I_{\alpha}^{\text{S}}(\tilde{U};Y)\leq I_{\alpha}^{\text{S}}(\tilde{X};Y)$ with equality if and only if $\tilde{U}=\tilde{X}$ \cite[Thm. 3]{alphaMI_verdu}. Therefore, in \eqref{eq:GealLeak_EquivalentInproofConverse3} $\sup_{P_{\tilde{U}}} I_{\alpha}^{\text{S}}(\tilde{U};Y)= I_{\alpha}^{\text{S}}(\tilde{X};Y)$, and then, by replacing $\tilde{U}$ with $\tilde{X}$ we have \eqref{eq:GealLeak_EquivalentInproofConverse4}.  \\
		\textbf{Lower bound}: %\textbf{Achievability}: 
		We bound \eqref{eq:Inproof_Thm:DefEquialentExpression} from below by considering a random variable $U$ such that $U-X-Y$ is a Markov chain and $H(X|U)=0$. Specifically, let the alphabet $\mcl U$ consist of $\mcl U_x$, a collection of $U$ mapped to a $x\in \mcl X$, i.e.,
		$\mcl U=\bigcup_{x\in\mcl X} \mcl U_x $ with $U=u\in \mcl U_x$ if and only if $X=x$.
		Therefore, for the specific variable $U$, we have
		\begin{align}
			\label{eq:GealLeak_EquivalentInproofAchieval0}
			P_{Y|U}(y|u)&=\begin{cases}
				P_{Y|X}(y|x) \quad &\text{ for all } u\in  \mcl U_x\\%y\in \{y:P_{Y|X}(y|x)>0\}
				0                  &\text{ otherwise}.
			\end{cases}
		\end{align}
		Construct a probability distribution $P_{\tilde{X}}$ over $\mathcal{X}$ from $P_U$ as
		\begin{align}
			\label{eq:GealLeak_EquivalentInproofAchievalConstructPX}
			P_{\tilde{X}}(x)=\frac{\sum_{u\in\mathcal{U}_x}P_U(u)^{\alpha}}{\sum_{x\in\mcl X}\sum_{u\in\mathcal{U}_x}P_U(u)^{\alpha}}  \quad \text{ for all }  x\in \mcl X.
		\end{align} Thus,
		\begin{IEEEeqnarray}{l  l}
			&I_{\alpha}^{\text{A}}(U;Y)\nonumber\\
			=&\frac{\alpha}{\alpha-1}\log\frac{\sum\limits_{y\in\mcl Y}\left(\sum\limits_{x\in\mcl X}\sum\limits_{u\in\mathcal{U}_x}P_{Y|U}(y|u)^{\alpha}P_{U}(u)^{\alpha}\right)^{\frac{1}{\alpha}}}{\left(\sum\limits_{x\in\mcl X}\sum\limits_{u\in\mathcal{U}_x}P_U(u)^{\alpha}\right)^{\frac{1}{\alpha}}}\\
			%\label{eq:GealLeak_EquivalentInproofAchieval1}
			=& \frac{\alpha}{\alpha-1}\log\sum\limits_{y\in\mcl Y}\left(\sum\limits_{x\in\mcl X}P_{Y|X}(y|x)^{\alpha}\frac{\sum\limits_{u\in\mathcal{U}_x}P_{U}(u)^{\alpha}}{ \sum\limits_{x\in\mcl X}\sum\limits_{u\in\mathcal{U}_x}P_U(u)^{\alpha} }\right)^{\frac{1}{\alpha}}\quad	\\
			%\label{eq:GealLeak_EquivalentInproofAchieval2}
			=& \frac{\alpha}{\alpha-1}\log\sum_{y\in\mcl Y}\left(\sum_{x\in\mcl X}P_{Y|X}(y|x)^{\alpha}P_{\tilde{X}}(x)\right)^{\frac{1}{\alpha}}\\
			=&I_{\alpha}^{\text{S}}(\tilde{X};Y).
		\end{IEEEeqnarray}
		Therefore, 
			\begin{align}
				\mcl L_{\alpha}^{\text{max}}(X\to Y) =&\sup_{U-X-Y} I_{\alpha}^{\text{A}}(U;Y)\\
				\geq & \sup_{U:U-X-Y,H(X|U)=0} I_{\alpha}^{\text{A}}(U;Y)\\
				=&\sup_{P_{\tilde{X}}\ll P_X}I_{\alpha}^{\text{S}}(\tilde{X};Y), \label{eq:GealLeak_EquivalentInproofAchievable}
			\end{align}
			where the last equality follows because, for any $P_{\tilde{X}}\ll P_X$, there exists a distribution $P_U(u)$ for $u\in\mathcal{U}$ such that  \eqref{eq:GealLeak_EquivalentInproofAchievalConstructPX} holds;
			therefore, the supremum over these $U$ in \eqref{eq:GealLeak_EquivalentInproofAchieval0} is equivalent to the supremum of $P_{\tilde{X}}$.		
		Therefore, combining \eqref{eq:GealLeak_EquivalentInproofConverse4} and \eqref{eq:GealLeak_EquivalentInproofAchievable}, we obtain \eqref{eq:GealLeak_EquivDef_1infty}. 
	\end{proof}
	
	\section{Proof for Lemma \ref{Lem:alphaLeakage_fDivergenceLeakage}}\label{Proof:alphaLeakage_fDivergenceLeakage}
	\begin{proof}[\nopunct]
		Define the convex function
		\begin{equation}
		f_\alpha(t)=\frac{1}{\alpha-1} (t^\alpha-1),
		\end{equation}
		then for the two distributions $P$ and $Q$ over the support $\mcl Y$, we have a $f$-divergence $\mcl H_{\alpha}(P\|Q)$, which is the Hellinger divergence of order $\alpha$ \cite{Liese2006}, given by
		\begin{equation}
		\mcl H_{\alpha}(P\|Q)=\frac{1}{\alpha-1} \left(\sum_{\mcl Y} P(y)^\alpha Q(y)^{1-\alpha}-1\right).
		\end{equation}
		Therefore, the R\'enyi divergence can be written in terms of the Hellinger divergence as
		\begin{equation}
		D_\alpha(P\|Q)=\frac{1}{\alpha-1}\log (1+(\alpha-1) \mcl H_{\alpha}(P\|Q)).
		\end{equation}	
		Thus, since $z\mapsto \frac{1}{\alpha-1}\log(1+(\alpha-1) z)$ is monotonically increasing in $z$ for $\alpha>1$, we can write maximal $\alpha$-leakage as
		\begin{IEEEeqnarray}{l l}
			&\mathcal{L}^{\text{max}}_\alpha(X\to Y)\nonumber\\
			=&\sup_{P_X}\,\inf_{Q_Y}\,D_\alpha(P_{X,Y}\|P_X\times Q_Y)\\
			=&\frac{1}{\alpha-1} \log \big(1+(\alpha-1) \sup_{P_X}\,\inf_{Q_Y}\,\mcl H_{\alpha}(X\to Y)\big)\\
			=&\frac{1}{\alpha-1} \log \big(1+(\alpha-1) \mathcal{L}_{\mcl H_{\alpha}}(X\to Y)\big).
		\end{IEEEeqnarray}	
		That is, for $\alpha>1$ maximal $\alpha$-leakage is a monotonic function of the Hellinger divergence-based measure.
	\end{proof}

	\section{Proof of Theorem \ref{Thm:Geneleak_qusiconvex_nondecreasing_dataprocessing}}\label{Proof:Geneleak_qusiconvex_nondecreasing_dataprocessing}
	\begin{proof}[\nopunct]
		\textbf{The proof of part 1}: We know that for $\alpha\geq 1$, $I^{\text{S}}_{\alpha}(X;Y)$ is quasi-convex $P_{Y|X}$ for given $P_X$ \cite[Thm. 2.7.4]{IT_Cover}, \cite[Thm. 10]{ConvexityAlphaMI_Ho}. In addition, the supremum of a set of quasi-convex functions is also quasi-convex, i.e., if the function $f(a,b)$ is quasi-convex in $b$ for any given $a$, the supremum $\sup_a f(a,b)$ is also quasi-convex in $b$ \cite{boydconvex}. Therefore, maximal $\alpha$-leakage in \eqref{eq:GealLeak_EquivDef} is quasi-convex $P_{Y|X}$.\\
		\textbf{The proof of part 2}: 
		Let $\beta>\alpha\geq1$, and $P_{X\alpha}^{\star}=\arg \sup_{P_X} I^{\text{S}}_{\alpha}(P_X,P_{Y|X})$ for given $P_{Y|X}$, such that
		\begin{IEEEeqnarray}{l l}
				\mcl L_{\alpha}^{\text{max}}(X\to Y)&= I^{\text{S}}_{\alpha}(P_{X\alpha}^{\star},P_{Y|X})\\
				\label{eq:GeneLeak_Property1inProof4}
				& \leq I^{\text{S}}_{\beta}(P_{X\alpha}^{\star},P_{Y|X})\\
				\label{eq:GeneLeak_Property1inProof5}
				& \leq \sup_{P_X} I^{\text{S}}_{\beta}(P_X,P_{Y|X})\\
				&=\mcl L^{\text{max}}_{\beta}(X\to Y)
		\end{IEEEeqnarray}
		where \eqref{eq:GeneLeak_Property1inProof4} results from that $I^{\text{S}}_{\alpha}$ is non-decreasing in $\alpha$ for $\alpha>0$ \cite[Thm. 4]{ConvexityAlphaMI_Ho}, and the equality in \eqref{eq:GeneLeak_Property1inProof5} holds if and only if $P_{X\alpha}^{\star}=\arg \sup_{P_X} I_{\beta}(P_X,P_{Y|X})$.\\	
		\textbf{The proof of part 3}: 
		Let random variables $X$, $Y$ and $Z$ form the Markov chain $X-Y-Z$. Making use of that Sibson MI of order $\alpha>1$ satisfies data processing inequalities \cite[Thm. 3]{alphaMI_verdu}, i.e., 
		\begin{IEEEeqnarray}{l l}
				I^{\text{S}}_{\alpha}(X; Z)\leq I^{\text{S}}_{\alpha}(X; Y) \label{eq:DPInq_inproof01}\\
				I^{\text{S}}_{\alpha}(X; Z)\leq I^{\text{S}}_{\alpha}(Y; Z) \label{eq:DPInq_inproof02},
		\end{IEEEeqnarray}
		we prove that maximal $\alpha$-leakage satisfies data processing inequalities as follows.\\
		We first prove \eqref{eq:GeneLeak_DataProcessIneq_XY}. Let $P^{\star}_X=\arg\sup_{P_X} I^{\text{S}}_{\alpha}(P_X,P_{Z|X})$. For the Markov chain $X-Y-Z$, we have
		\begin{IEEEeqnarray}{l l}
				\mcl L_{\alpha}^{\text{max}}(X\to Z)&=I^{\text{S}}_{\alpha}(P^{\star}_X,P_{Z|X}) \label{eq:DPInq_inproof1}\\
				&\leq I^{\text{S}}_{\alpha}(P^{\star}_X,P_{Y|X}) \label{eq:DPInq_inproof2}\\
				&\leq \sup_{P_X} I^{\text{S}}_{\alpha}(P_X,P_{Y|X}) \label{eq:DPInq_inproof3}\\
				&=\mcl L_{\alpha}^{\text{max}}(X\to Y) \label{eq:DPInq_inproof4}
		\end{IEEEeqnarray}
		where the inequality in \eqref{eq:DPInq_inproof2} results from \eqref{eq:DPInq_inproof01}. Similarly, the inequality in \eqref{eq:GeneLeak_DataProcessIneq_YZ} can be proved directly from \eqref{eq:DPInq_inproof02}.\\	
		\textbf{The proof of part 4}: 
		For $\alpha=1$, we have
		\begin{align}
		\mcl L^{\text{max}}_{1}(X\to Y)=I(X;Y)\geq 0,
		\end{align}
		with equality if and only if $X$ is independent of $Y$ \cite{IT_Cover}. 
		For $1<\alpha\leq \infty$, 
		referring to \eqref{eq:Sibson_MI} and \eqref{eq:GealLeak_EquivDef_1infty} we have 
		\begin{IEEEeqnarray}{l l}
				&\mcl L_{\alpha}^{\text{max}}(X\to Y)\nonumber\\
				=&\sup_{P_X}	\frac{\alpha}{\alpha-1}\log \sum_{y}\left(\sum_{x}P_X(x)P_{Y|X}(y|x)^{\alpha}\right)^{\frac{1}{\alpha}}\\
				\label{eq:alphaLeak_SpecialMechanism_Inproof1}
				\geq & \sup_{P_X}	\frac{\alpha}{\alpha-1}\log \sum_{y}\bigg(\sum_{x}P_X(x)P_{Y|X}(y|x)\bigg)^{\frac{\alpha}{\alpha}}\\
				= & \sup_{P_X}	\frac{\alpha}{\alpha-1}\log 1=0,
		\end{IEEEeqnarray} 	
		where \eqref{eq:alphaLeak_SpecialMechanism_Inproof1} results from applying Jensen’s inequality to the convex function $f: t\to t^{\alpha}$ ($t\geq 0$), such that the equality holds if and only if given any $y\in\mcl Y$, $P_{Y|X}(y|x)$ are the same for all $x\in \mcl X$, such that
		\begin{align}
			P_{Y|X}(y|x)=P_Y(y)\quad x\in \mcl X, y\in\mcl Y
		\end{align} which means $X$ and $Y$ are independent, i.e., $P_{Y|X}$ is a rank-1 row stochastic matrix. \\			
		For $\alpha=1$, from \eqref{eq:GealLeak_EquivDef_1} we know $\mcl L^{\text{max}}_{1}(X\to Y)=I(X;Y)$. Therefore, 
		\begin{IEEEeqnarray}{l l}
			& \mcl L^{\text{max}}_{1}(X\to Y)-H(X)\nonumber\\
			=& \sum_{x\in\mcl X,y\in\mcl Y}P(x,y)\log \frac{P(y|x)}{P(y)} - \sum_{x\in\mcl X }P(x)\log \frac{1}{P(x)}\quad \\
			=&\sum_{x,y}\hspace{-5pt}P(x,y)\log \frac{P(y|x)}{P(y)} - \hspace{-5pt} \sum_{x,y }\hspace{-5pt}P(x,y)\log \frac{1}{P(x)}\\
			=&\sum_{x,y}P(x,y)\log P(x|y) \leq 0,
		\end{IEEEeqnarray}
		with equality if and only if for all $x,y\in \mcl X\times \mcl Y$, the conditional probability $P_{X|Y}(x|y)$ is either $1$ or $0$. That is, $\mcl L^{\text{max}}_{1}(X\to Y)\leq H(X)$ with equality if and only if $X$ is a deterministic function of $Y$. 
		% \cite[Lem. 1]{MaximalLeakageHT_Liao2017}. 
		For $1<\alpha\leq \infty$, from the monotonicity of maximal $\alpha$-leakage in $\alpha$ and \eqref{eq:GealLeak_EquivDef_1infty}, we have
		\begin{align}
		\mcl L^{\text{max}}_\alpha(X\to Y)&\leq \mcl L^{\text{max}}_{\infty}(X\to Y)\\
		&=\log \sum_{y\in\mcl Y}\max_x P_{Y|X}(y|x)\\
		\label{eq:MaxAlphaLeak-Property-Range_InPf}
		&\leq\log \sum_{y\in\mcl Y}\sum_{x\in\mcl X} P_{Y|X}(y|x)=\log|\mcl X|.
		\end{align}
		where the equality in \eqref{eq:MaxAlphaLeak-Property-Range_InPf} holds if and only if for every $y\in \mcl Y$, $\sum_{\mcl X} P(y|x)=\max_x P(y|x)$, i.e., $X$ is a deterministic function of $Y$. To prove that for $\alpha\in(1,\infty)$, the upper bound in \eqref{eq:MaxAlphaLeak-Property-Range_InPf} is achievable, we construct a mapping $P_{X\Leftarrow Y}$ such that $X$ is a deterministic function of $Y$. That is, for every $y\in \mcl Y$, there exists a unique $x_y\in \mcl X$ such that $P(x_y|y)=1$. Therefore, we have $x_y=\arg_x P_{X\Leftarrow Y}(y|x)>0$.
		For $\alpha\in(1,\infty)$, from \eqref{eq:Sibson_MI} and \eqref{eq:GealLeak_EquivDef_1} we have
		\begin{IEEEeqnarray}{l l}
				 &\mcl L_{\alpha}^{\text{max}}(P_{X\Leftarrow Y}) \nonumber\\
				=&\sup_{P_X} \frac{\alpha}{\alpha-1}\log \sum_{y\in\mathcal{Y}}\left(P_X^{\frac{1}{\alpha}}(x_y)P_{X\Leftarrow Y}(y|x_y)\right)\\
				\label{eq:GeneLeak_LemmaSpecialMech_proof1}
				= &\sup_{P_X} \frac{\alpha}{\alpha-1}\log \sum_{x\in\mathcal{X}}P_X^{\frac{1}{\alpha}}(x);
		\end{IEEEeqnarray} 
		in addition, since the function maximized in \eqref{eq:GeneLeak_LemmaSpecialMech_proof1} is symmetric and concave in $P_X$, it is Schur-concave in $P_X$, and therefore, the optimal distribution of $X$ achieving the supreme in \eqref{eq:GeneLeak_LemmaSpecialMech_proof1} is uniform. Thus, 
		\begin{align}
			\mcl L_{\alpha}^{\text{max}}(P_{X\Leftarrow Y})=\log |\mathcal{X}|,\quad 1<\alpha\leq \infty .
		\end{align}
		Therefore, maximal $\alpha$-leakage achieves its maximal value $\log|\mathcal{X}|$ and $H(P_X)$ for $\alpha>1$ and $\alpha=1$, respectively, if and only if $X$ is a deterministic function of $Y$.
	\end{proof}

\section{Proof for Theorem \ref{Thm:MaxAlphaLeak-Bounds}}\label{Proof:Thm:MaxAlphaLeak-Bounds}
To prove Thm. \ref{Thm:MaxAlphaLeak-Bounds}, we define a divergence function $k_{\alpha}$ for $\alpha>1$ and provide a lower bound for its sum in the following definition and lemma, respectively.
\begin{definition}\label{Def:AlphaDiv_insidelog}
	Given two discrete distributions $P_Y$ and $Q_Y$ over the support $\mcl Y$, a divergence function $k_{\alpha}$ for $\alpha>1$ is defined as
	\begin{align}
	\label{eq:AlphaDiv_insidelog}
	k_{\alpha}(P_Y\|Q_Y)
	\triangleq\sum_{y}Q_Y(y)\left(\frac{P_Y(y)}{Q_Y(y)}\right)^{\alpha}.
	\end{align}
\end{definition}

\begin{proposition}\label{Pro:AlphaDiv_insidelog}
	The function $k_{\alpha}(P_Y\|Q_Y)$ in \eqref{eq:AlphaDiv_insidelog} is jointly convex in $(P_Y,Q_Y)$, and $k_{\alpha}(P_Y\|Q_Y)\geq 1$ with equality if and only if $P_Y=Q_Y$.
\end{proposition}
\begin{proof}
	For $\alpha\geq 1$, the function $f(t)=t^\alpha$ is convex in $t\geq 0$, such that the perspective of $f(t)$, defined as $g(t,a)=af(t/a)$ ($a>0$), is convex in $(t,a)$ \cite[Chapter. 3.2.6]{boydconvex}. Let $t=P_Y(y)$ and $a=Q_Y(y)>0$ such that the perspective function can be written as
	\begin{align}
		g(P_Y(y),Q_Y(y))=Q_Y(y)\left(\frac{P_Y(y)}{Q_Y(y)}\right)^{\alpha},
	\end{align}
	which is therefore convex in $(P_Y(y),Q_Y(y))$. For $Q_Y(y)=0$, the function $g(P_Y(y),Q_Y(y))$ is zero, which is also convex in $(P_Y(y),Q_Y(y))$.
	 Thus, the function $k_{\alpha}(P_Y\|Q_Y)$ in \eqref{eq:AlphaDiv_insidelog} is a sum of convex functions, and therefore, it is convex in $(P_Y(y),Q_Y(y))$.\\
	 Let $t=\frac{P_Y(y)}{Q_Y(y)}$. From the convexity of $f(t)=t^\alpha$ in $t\geq 0$ and Jensen’s inequality \cite[Chapter. 3.1.8]{boydconvex}, we have that 
	 \begin{align}
	 	k_{\alpha}(P_Y\|Q_Y)
	 	\triangleq&\sum_{y}Q_Y(y)\left(\frac{P_Y(y)}{Q_Y(y)}\right)^{\alpha}\\
	 	\geq& \sum_{y}\left(\sum_{y}Q_Y(y)\frac{P_Y(y)}{Q_Y(y)}\right)^{\alpha}=1.
	 \end{align}
\end{proof}

\begin{lemma}\label{lemma:Ineqaulity-k_alpha}
	Let $K$ be a positive integer with $K<\infty$. Given a group of distributions $\{P_k:k\in[1,K]\}$ and an arbitrary distribution $P$ on a discrete set $\mcl Y$, there is
		\begin{align}\label{eq:Ineqaulity-k_alpha}
		\sum_{k=1}^{K} k_{\alpha}(P_k\|P)\geq &\sum_{k=1}^{K} k_{\alpha}(P_k\|P_c)\\
		=&\left(\sum_y \left(\sum\limits_{k=1}^{K}P_k(y)^{\alpha}\right)^{\frac{1}{\alpha}}\right)^{\alpha},
		\end{align}
	with equality if and only if $P=P_c$, where $P_c$ is given by
	\begin{align}\label{eq:constructed-Pc}
	P_c(y)=\frac{1}{Z}\left(\sum\limits_{k=1}^{K}P_k(y)^{\alpha}\right)^{\frac{1}{\alpha}},\, \alpha\in [1,\infty]
	\end{align} 
	where $Z$ is the constant as
	\begin{align}\label{eq:ConstZ_constructed-Pc}
	Z=\sum_y \left(\sum\limits_{k=1}^{K}P_k(y)^{\alpha}\right)^{\frac{1}{\alpha}},
	\end{align}
	which guarantees that $P_c$ is a distribution.
\end{lemma}
\begin{proof}
	From the definition $k_\alpha$ in \eqref{eq:AlphaDiv_insidelog}, we have
	\begin{IEEEeqnarray}{l l}
		&\sum_{k=1}^{K} k_{\alpha}(P_k\|P)-\sum_{k=1}^{K} k_{\alpha}(P_k\|P_c)\nonumber\\
		=& \sum_{k=1}^{K} \sum_{y}P_k(y)^\alpha\left(P(y)^{1-\alpha}-P_c(y)^{1-\alpha}\right)\\
		=& \sum_{y}  \left(\sum_{k=1}^{K}P_k(y)^\alpha\right) \left(P(y)^{1-\alpha}-P_c(y)^{1-\alpha}\right)\\
		=&\sum_{y} Z^{\alpha}P_c(y)^{\alpha}\left(P(y)^{1-\alpha}-P_c(y)^{1-\alpha}\right)\\
		=&Z^{\alpha} \sum_{y}\left(P_c(y)^{\alpha}P(y)^{1-\alpha}-P_c(y)\right)\\
		=& Z^{\alpha}(k_{\alpha}(P_c\|P)-1)\geq 0
	\end{IEEEeqnarray}
	with equality if and only if $P=P_c$. In addition, making use of the expression of $P_c$ and $Z$ in \eqref{eq:constructed-Pc} and \eqref{eq:ConstZ_constructed-Pc}, respectively, we have
	\begin{IEEEeqnarray}{l l}
		&\sum_{k=1}^{K} k_{\alpha}(P_k\|P_c)\nonumber\\
		=&\sum_{k=1}^{K}\sum_{y}P_c(y)\left(\frac{P_k(y)}{P_c(y)}\right)^{\alpha}\\
		=&\sum_{k=1}^{K}\sum_{y}Z^{\alpha-1}\left(\sum\limits_{k'=1}^{K}P_{k'}(y)^{\alpha}\right)^{\frac{1}{\alpha}}\frac{P_k(y)^{\alpha}}{\sum\limits_{k'=1}^{K}P_{k'}(y)^{\alpha}} \label{eq:revise1_e1}\\
		=& Z^{\alpha-1}\sum_{y}\left(\sum\limits_{k'=1}^{K}P_{k'}(y)^{\alpha}\right)^{\frac{1}{\alpha}}\frac{\sum_{k=1}^{K} P_k(y)^{\alpha}}{\sum_{k'=1}^{K}P_{k'}(y)^{\alpha}}\\
		=& \left(\sum_y \left(\sum\limits_{k=1}^{K}P_k(y)^{\alpha}\right)^{\frac{1}{\alpha}}\right)^{\alpha}.
	\end{IEEEeqnarray}
\end{proof}
Making use of the results in Lemma \ref{lemma:Ineqaulity-k_alpha}, we prove Thm. \ref{Thm:MaxAlphaLeak-Bounds} as follows.
\begin{proof}%[\nopunct]
	From Thm. \ref{Thm:DefEquialentExpression}, we have that for $\alpha>1$
	\begin{IEEEeqnarray}{l l}
	&\mcl L_{\alpha}^{\text{max}}(X\to Y)\nonumber\\
	=& \sup_{P_{\tilde{X}}}I^{\text{S}}_{\alpha}(\tilde{X},Y)
	=\sup_{P_{\tilde{X}}}\inf_{Q_Y}D_{\alpha}(P_{\tilde{X}}P_{Y|X}\|P_{\tilde{X}}Q_Y)\\
	\label{eq:MaxAlphaLeak-Bounds-inPf}
	=&\sup_{P_{\tilde{X}}}\inf_{Q_Y}\frac{1}{\alpha-1}\log \sum_x P_{\tilde{X}}(x)k_{\alpha}(P_{Y|X=x}\|Q_Y).
	\end{IEEEeqnarray}
	For $\alpha>1$, the function $f:t\to \frac{1}{\alpha-1}\log t$ is increasing in $t\geq 0$. Therefore, we simplify the optimization in \eqref{eq:MaxAlphaLeak-Bounds-inPf} as
		\begin{align}\label{eq:MaxAlphaLeak-Bounds-inPf_1}
		\sup_{P_{\tilde{X}}}\, \inf_{Q_Y}\, \sum_x P_{\tilde{X}}(x)k_{\alpha}(P_{Y|X=x}\|Q_Y)
		\end{align}
	and provide a lower bound of \eqref{eq:MaxAlphaLeak-Bounds-inPf_1} as follows.
	Since the divergence function $k_{\alpha}$ is joint convex in the pair of distributions, the objective function in \eqref{eq:MaxAlphaLeak-Bounds-inPf_1} is joint convex in $(P_{Y|X},Q_Y)$ for fixed $P_{\tilde{X}}$, and linear in $P_{\tilde{X}}$ for fixed $(P_{Y|X},Q_Y)$. Therefore, the max-min equals to the min-max as followed:
	\begin{IEEEeqnarray}{l l}
		&\sup_{P_{\tilde{X}}} \inf_{Q_Y} \sum_x P_{\tilde{X}}(x)k_{\alpha}(P_{Y|X=x}\|Q_Y)\nonumber\\
		=&\inf_{Q_Y} \sup_{P_{\tilde{X}}} \sum_x P_{\tilde{X}}(x)k_\alpha(P_{Y|X=x}\|Q_Y)\label{eq:PUT-k_alpha-InPrf2}\\
		=& \inf_{Q_Y} \max_{x} \, k_\alpha(P_{Y|X=x}\|Q_Y)\label{eq:PUT-k_alpha-InPrf3}\\
		\geq&  \inf_{Q_Y} \frac{\sum_x k_\alpha(P_{Y|X=x}\|Q_Y)}{|\mcl X|}\label{eq:PUT-k_alpha-InPrf4}\\
		\geq&   \frac{\sum_x k_\alpha(P_{Y|X=x}\|P_c)}{|\mcl X|}\label{eq:PUT-k_alpha-InPrf5}\\
		=& \frac{1}{|\mcl X|}\left(\sum_{y}\|P_{Y|X}(y|\cdot)\|_{\alpha}\right)^{\alpha},\label{eq:PUT-k_alpha-InPrf6}
	\end{IEEEeqnarray}
	where the inequality in \eqref{eq:PUT-k_alpha-InPrf5} is directly from \eqref{eq:Ineqaulity-k_alpha} in Lemma~ \ref{lemma:Ineqaulity-k_alpha} with equality if and only if 
	\begin{align}\label{eq:k_alpha-P_c-InPrf}
	Q_Y(y)=P_c(y)=\frac{1}{Z}\|P_{Y|X}(y|\cdot)\|_{\alpha},
	\end{align}
	with the constant $Z=\sum_y \|P_{Y|X}(y|\cdot)\|_{\alpha}$.
	Therefore, for any $P_{Y|X}$, we have 
	\begin{align}
     \mcl L_{\alpha}^{\text{max}}(X\to Y) \geq \frac{\alpha}{\alpha-1}\log\frac{\sum_{y}\|P_{Y|X}(y|\cdot)\|_{\alpha}}{|\mcl X|^{\frac{1}{\alpha}}},
	\end{align}
	with equality if and only if the $P_{Y|X}$ guarantees that the divergence function $k_\alpha(P_{Y|X=x}\|P_c)$ are the same for all $x\in \mcl X$, i.e., the $P_{Y|X}$ satisfies \eqref{eq:MaxAlphaLeak-LowBD-ComplementCnst}. 
\end{proof}

	\section{Proof of Theorem \ref{Thm:GeneLeak_CompositionTheory}}\label{proof:Thm:GeneLeak_CompositionTheory}
	\begin{proof}[\nopunct]
		Let $\mcl Y_1$ and $\mcl Y_2$ be the alphabets of $Y_1$ and $Y_2$, respectively. For any $(y_1,y_2)\in \mcl Y_1\times \mcl Y_2$, due to the Markov chain $Y_1-X-Y_2$, the corresponding entry of the conditional probability matrix of $(Y_1,Y_2)$ given $X$ is
		\begin{IEEEeqnarray}{l}
			P(y_1y_2|x)=P(y_1|x)P(y_2|xy_1)=P(y_1|x)P(y_2|x).\nonumber
		\end{IEEEeqnarray}
		Therefore, for $\alpha\in(1,\infty)$
		\begin{IEEEeqnarray}{l l}
				&\mcl L_{\alpha}^{\text{max}}(X\to Y_1,Y_2)\nonumber\\
				=&\sup_{P_X} \frac{\alpha}{\alpha-1}\log \hspace{-2pt}\sum_{y_1,y_2}\hspace{-3pt}\left(\hspace{-1pt}\sum_{x}P_X(x)P_{Y_1,Y_2|X}(y_1,y_2|x)^{\alpha}\hspace{-2pt}\right)^{\frac{1}{\alpha}}\\
				=&\sup_{P_X} \frac{\alpha}{\alpha\hspace{-2pt}-\hspace{-2pt}1}\hspace{-2pt}\log\hspace{-3pt}\sum_{y_1,y_2}\hspace{-4pt}\bigg(\hspace{-3pt}\sum_{x}\hspace{-2pt}P(x)
				\label{eq:alphaLeakage_CompostionTheoremProof_0}
				P(y_1|x)^{\alpha}\hspace{-2pt}P(y_2|x)^{\alpha}\hspace{-3pt}\bigg)^{\hspace{-2pt}\frac{1}{\alpha}}.
	  \end{IEEEeqnarray}	
		Let $K(y_1)=\sum_{x\in\mathcal{X}}P_X(x)P_{Y_1|X}(y_1|x)^{\alpha}$, for all $y_1\in\mcl Y_1$, such that we can construct a set of distributions over $\mcl X$ as 
		\begin{align}
			P_{\tilde{X}}(x|y_1)=\frac{P_X(x)P_{Y_1|X}(y_1|x)^{\alpha}}{K(y_1)}.
		\end{align} Therefore, from \eqref{eq:alphaLeakage_CompostionTheoremProof_0}, $\mcl L_{\alpha}^{\text{max}}(X\to Y_1,Y_2)$ can be rewritten as
		\begin{IEEEeqnarray}{l l}
			&\mcl L_{\alpha}^{\text{max}}(X\to Y_1,Y_2)\nonumber\\
			&=\sup_{P_X} \frac{\alpha}{\alpha-1}\log \sum_{y_1,y_2\in \mcl Y_1\times \mcl Y_2}\nonumber\\
			&\qquad\bigg(\sum_{x\in\mathcal{X}}\hspace{-3pt}K(y_1)P_{\tilde{X}}(x|y_1)P_{Y_2|X}(y_2|x)^{\alpha}\bigg)^{\frac{1}{\alpha}}\\
			&=\sup_{P_X} \frac{\alpha}{\alpha-1}\log \sum_{y_1,y_2}\Bigg(\bigg(\sum_{x}P_X(x)P_{Y_1|X}(y_1|x)^{\alpha}\bigg)^{\frac{1}{\alpha}}\nonumber\\
			&\quad \cdot\bigg(\sum_{x} P_{\tilde{X}}(x|y_1)P_{Y_2|X}(y_2|x)^{\alpha}\bigg)^{\frac{1}{\alpha}}\Bigg)\\
			&= \sup_{P_X} \frac{\alpha}{\alpha-1}\log \sum_{y_1}\Bigg(\bigg(\sum_{x}P_X(x)P_{Y_1|X}(y_1|x)^{\alpha}\bigg)^{\frac{1}{\alpha}}\nonumber\\
			&\quad \cdot\sum_{y_2}\bigg(\sum_{x} P_{\tilde{X}}(x|y_1)P_{Y_2|X}(y_2|x)^{\alpha}\bigg)^{\frac{1}{\alpha}}\Bigg)\\
			 &\leq \sup\limits_{P_X} \frac{\alpha}{\alpha-1}\log \Bigg(\sum\limits_{y_1}\left(\sum\limits_{x}P_X(x)P_{Y_1|X}(y_1|x)^{\alpha}\right)^{\frac{1}{\alpha}}\nonumber\\
			\label{eq:alphaLeakage_CompostionTheoremProof_1}
			&\quad \cdot\max\limits_{y_1}\sum\limits_{y_2}\left(\sum\limits_{x} P_{\tilde{X}}(x|y_1)P_{Y_2|X}(y_2|x)^{\alpha}\right)^{\frac{1}{\alpha}}\Bigg)\\			
			&=\sup\limits_{P_X} \frac{\alpha}{\alpha-1}\log \Bigg(\sum\limits_{y_1}\left(\sum\limits_{x}P_X(x)P_{Y_1|X}(y_1|x)^{\alpha}\right)^{\frac{1}{\alpha}}\\
			\label{eq:alphaLeakage_CompostionTheoremProof_2}
			&\quad \cdot\sum\limits_{y_2}\left(\sum\limits_{x} P_{\tilde{X}}(x|y_1^{\star})P_{Y_2|X}(y_2|x)^{\alpha}\right)^{\frac{1}{\alpha}}\Bigg)\\			
			&\leq \sup\limits_{P_X} \frac{\alpha}{\alpha-1}\log \sum\limits_{y_1}\left(\sum\limits_{x}P_X(x)P_{Y_1|X}(y_1|x)^{\alpha}\right)^{\frac{1}{\alpha}}\nonumber\\
			\label{eq:alphaLeakage_CompostionTheoremProof_3}
			&\quad +\sup\limits_{P_{\tilde{X}}}\frac{\alpha}{\alpha-1}\log\sum\limits_{y_2}\left(\sum\limits_{x} P_{\tilde{X}}(x)P_{Y_2|X}(y_2|x)^{\alpha}\right)^{\frac{1}{\alpha}}\\
			&=\mcl L_{\alpha}^{\text{max}}(X\to Y_1)+\mcl L_{\alpha}^{\text{max}}(X\to Y_2),
		\end{IEEEeqnarray}	
		where $y_1^{\star}$ in \eqref{eq:alphaLeakage_CompostionTheoremProof_2} is the optimal $y_1$ achieving the maximum in \eqref{eq:alphaLeakage_CompostionTheoremProof_1}. Therefore, the equality in \eqref{eq:alphaLeakage_CompostionTheoremProof_1} holds if and only if for all $y_1\in \mcl Y_1$
		\begin{IEEEeqnarray}{r l}
			&\sum_{y_2}\left(\sum_{x} P_{\tilde{X}}(x|y_1)P_{Y_2|X}(y_2|x)^{\alpha}\right)^{\frac{1}{\alpha}}\nonumber\\
			=&\sum_{y_2}\left(\sum_{x} P_{\tilde{X}}(x|y_1^{\star})P_{Y_2|X}(y_2|x)^{\alpha}\right)^{\frac{1}{\alpha}};
		\end{IEEEeqnarray}
		and the equality in \eqref{eq:alphaLeakage_CompostionTheoremProof_3} holds if and only if the optimal solutions $P_X^{\star}$ and $P_{\tilde{X}}^{\star}$ of the two maximizations in \eqref{eq:alphaLeakage_CompostionTheoremProof_3} satisfy, for all $x\in\mcl X$,
		%of the two corresponding maximizations in \eqref{eq:alphaLeakage_CompostionTheoremProof_3} 
		\begin{align}
			P_{\tilde{X}}^{\star}(x)=\frac{P_X^{\star}(x)P_{Y_1|X}^{\alpha}(y_1^{\star}|x)}{\sum_{x\in\mathcal{X}}P_X(x)P_{Y_1|X}^{\alpha}(y_1^{\star}|x)}.
		\end{align}
		Now we consider $\alpha=1$. For $Y_1-X-Y_2$, we have
		\begin{align}
			I(Y_2;X|Y_1)\leq I(Y_2;X).
		\end{align}
		From Thm. \ref{Thm:DefEquialentExpression}, there is
		\begin{IEEEeqnarray}{l l}
				&\mcl L^{\text{max}}_{1}(X\to Y_1,Y_2)\nonumber\\
				= & I(X;Y_1)+I(X;Y_2|Y_1)\\
				\leq & I(X;Y_1)+I(X;Y_2)\\
				= & \mcl L^{\text{max}}_{1}(X\to Y_1)+\mcl L^{\text{max}}_{1}(X\to Y_2).
		\end{IEEEeqnarray}
		For $\alpha=\infty$, we also have
		\begin{IEEEeqnarray}{l l}
				&\mcl L^{\text{max}}_{\infty}(X\to Y_1,Y_2)\nonumber\\
				= &\log \sum_{y_1,y_2\in \mcl Y_1\times \mcl Y_2}\max_{x\in\mcl X} P(y_1|x)P(y_2|x)\\
				\leq & \log \sum_{y_1,y_2\in \mcl Y_1\times \mcl Y_2}\left(\max_{x\in\mcl X} P(y_1|x)\right)\left(\max_{x\in\mcl X} P(y_2|x)\right)\\
				=& \log  \sum_{y_1\in \mcl Y_1}\max_{x\in\mcl X} P(y_1|x)+\log \sum_{y_2\in \mcl Y_2}\max_{x\in\mcl X} P(y_2|x)\\
				=&	\mcl L^{\text{max}}_{\infty}(X\to Y_1)+	\mcl L^{\text{max}}_{\infty}(X\to Y_2).
		\end{IEEEeqnarray}
		
	\end{proof}

\section{Proof of Theorem \ref{thm:MaxAlphaLK-Addtivity}}\label{proof:thm:MaxAlphaLK-Addtivity}
\begin{proof}[\nopunct]
	For $\alpha>1$, a function $f(t)=\frac{\alpha}{\alpha-1}\log t$ is monotonically increasing in $t>0$. Therefore, 
	to solve maximal $\alpha$-leakage from $X^n$ to $Y^n$, i.e., 
	\begin{align}
		&\mcl L_{\alpha}^{\text{max}}(X^n\to Y^n)\nonumber\\
		=&\sup_{P_{\tilde{X}^n}} \frac{\alpha}{\alpha-1}\log \sum_{y^n}\left(\sum_{x^n}P(x^n)P(y^n|x^n)^\alpha\right)^{\frac{1}{\alpha}},
	\end{align}
	it is sufficient to prove that 
	\begin{align}\label{eq:MaxAlphaLK-Add-inpf0}
		& \sup_{P_{\tilde{X}^n}} \sum_{y^n}\left(\sum_{x^n}P(x^n)P(y^n|x^n)^\alpha\right)^{\frac{1}{\alpha}}\nonumber\\
		= & \sup\limits_{\substack{P_{\tilde{X}_i}\\i\in[1,n]}}
		\prod\limits_{i=1}^n\left(\sum\limits_{y_{i}}\left(\sum\limits_{x_{i}}P\left(x_{i}\right)P\left(y_{i}|x_{i}\right)^{\alpha}\right)^{\frac{1}{\alpha}}\right).
	\end{align} 
	For a memoryless $P_{Y^n|X^n}$ with no feedback, we simplify \eqref{eq:MaxAlphaLK-Add-inpf0} as
	\begin{IEEEeqnarray}{l l}
			& \sup_{P_{\tilde{X}^n}} \,
			\sum_{y^n}\left(\sum_{x^n}\frac{P(x^n,y^n)}{P(y^n|x^n)^{1-\alpha}}\right)^{\frac{1}{\alpha}}\nonumber\\
			=&\sup\limits_{\prod_{i=1}^{n}P_{\tilde{X}_i|\tilde{X}_{i-1},\cdots,\tilde{X}_{1}}}\sum\limits_{y_{1},\cdots,y_{n} }\Bigg(\sum\limits_{x_{1},\cdots,x_{n}}\prod\limits_{i=1}^n\nonumber\\
			&\quad \frac{P(y_{i},x_{i}|x_{i-1},y_{i-1},\cdots,x_{1}y_{1})}{P\left(y_{i}|x^ny_{i-1},\cdots,y_{1}\right)^{1-\alpha}}\Bigg)^{\frac{1}{\alpha}}\label{eq:MaxAlphaLK-Add-inpf3}\\
			=&\sup\limits_{\prod_{i=1}^{n}P_{\tilde{X}_i|\tilde{X}_{i-1},\cdots,\tilde{X}_{1}}}\sum\limits_{y_{1},\cdots,y_{n} }\Bigg(\hspace{-3pt}\sum\limits_{x_{1},\cdots,x_{n}}\prod\limits_{i=1}^n\nonumber\\
			&\bigg(\frac{P\left(y_{i}|x_{i},\cdots,x_{1}\right)P\left(x_{i}|x_{i-1},\cdots,x_{1}\right)}{P\left(y_{i}|x^n\right)^{1-\alpha}}\bigg)\Bigg)^{\frac{1}{\alpha}}\label{eq:MaxAlphaLK-Add-inpf4}\\
			=&\sup\limits_{\prod_{i=1}^{n}P_{\tilde{X}_i|\tilde{X}_{i-1},\cdots,\tilde{X}_{1}}}\sum\limits_{y_{1},\cdots,y_{n} }\nonumber\\
			&\left(\sum\limits_{x_{1},\cdots,x_{n}}\prod\limits_{i=1}^nP\left(y_{i}|x_{i}\right)^{\alpha}P\left(x_{i}|x_{i-1},\cdots,x_{1}\right)\right)^{\frac{1}{\alpha}}\label{eq:MaxAlphaLK-Add-inpf5}\\
			= &\sup\limits_{\prod_{i=1}^{n}P_{\tilde{X}_i}}\sum\limits_{y_{1},\cdots,y_{n} }\left(\sum\limits_{x_{1},\cdots,x_{n}}\prod\limits_{i=1}^nP\left(y_{i}|x_{i}\right)^{\alpha}P\left(x_{i}\right)\right)^{\frac{1}{\alpha}}\label{eq:MaxAlphaLK-Add-inpf6}\\
			=&\sup\limits_{P_{\tilde{X}_{1}}, \cdots, P_{\tilde{X}_{n}}}\sum\limits_{y_{1},\cdots,y_{n} }\hspace{-5pt}\Bigg(\prod\limits_{i=1}^n\sum\limits_{x_{i}}P\left(x_{i}\right)P\left(y_{i}|x_{i}\right)^{\alpha}\Bigg)^{\frac{1}{\alpha}} \label{eq:MaxAlphaLK-Add-inpf7}\\
			=&\sup\limits_{\substack{P_{\tilde{X}_i}\\i\in[1,n]}}
			\prod\limits_{i=1}^n\left(\sum\limits_{y_{i}}\left(\sum\limits_{x_{i}}P\left(x_{i}\right)P\left(y_{i}|x_{i}\right)^{\alpha}\right)^{\frac{1}{\alpha}}\right) \label{eq:MaxAlphaLK-Add-inpf9}\\
			=&\sup\limits_{\substack{P_{\tilde{X}_i},i\in[1,n]}}
			\prod\limits_{i=1}^n \exp\Big\{\frac{\alpha-1}{\alpha}I^{\text{S}}_{\alpha}(\tilde{X}_i;Y_{i})\Big\}\label{eq:MaxAlphaLK-Add-inpf9-1}
		\end{IEEEeqnarray}
	where 
	\begin{itemize}
		\item \eqref{eq:MaxAlphaLK-Add-inpf3} {is} from the chain rule of probability;
		\item \eqref{eq:MaxAlphaLK-Add-inpf4} and \eqref{eq:MaxAlphaLK-Add-inpf5} are directly from the mechanism has no feedback and is memoryless, respectively;
		\item the equality in \eqref{eq:MaxAlphaLK-Add-inpf6} holds for memoryless sources, i.e., $P_{\tilde{X}_i|\tilde{X}_{i-1},\cdots,\tilde{X}_{1}}=P_{\tilde{X}_i}$ for all $i\in[1,n]$;
		\item both \eqref{eq:MaxAlphaLK-Add-inpf7} and \eqref{eq:MaxAlphaLK-Add-inpf9} are from the distributive property of multiplication;
		\item \eqref{eq:MaxAlphaLK-Add-inpf9-1} is from the definition of Sibson MI in \eqref{eq:Sibson_MI} and that the base of the logarithm is $2$.
	\end{itemize}
	Therefore, we have for $\alpha>1$,
	\begin{IEEEeqnarray}{l l}
		\sup\limits_{P_{\tilde{X}^n}} I^{\text{S}}_{\alpha}(\tilde{X}^n;Y^n)
		=\sum\limits_{i=1}^n \sup\limits_{\substack{P_{\tilde{X}_i}}}I^{\text{S}}_{\alpha}(\tilde{X}_i;Y_{i}).
	\end{IEEEeqnarray}
	That is,
	\begin{align}
		\mcl L_{\alpha}^{\text{max}}\left(X^n\to Y^n\right)=\sum\limits_{i=1}^n\mcl L_{\alpha}^{\text{max}}\left(X_{i}\to Y_{i}\right).
	\end{align}
	For $\alpha=1$, we have
   \begin{IEEEeqnarray}{l l}
			&I\left(X^n;Y^n\right)\nonumber\\
			=&\sum\limits_{i,j=1}^{n}I\mathsmaller{\left(X_{i};Y_{j}\big|X_{i-1},\cdots,X_{1},Y_{j-1},\cdots,Y_{1}\right)}\quad \label{eq:MaxAlphaLK-Add-inpf10}\\
			=&\sum_{i,j=1}^{n}I\left(X_{i};Y_{j}\big|X_{i-1},\cdots,X_{1}\right)\label{eq:MaxAlphaLK-Add-inpf11}\\
			=&\sum_{i=1}^{n}I\left(X_{i};Y_{i}\big|X_{i-1},\cdots,X_{1}\right)\label{eq:MaxAlphaLK-Add-inpf12}\\
			\leq & \sum_{i=1}^{n}I\left(X_{i};Y_{i}\right)\label{eq:MaxAlphaLK-Add-inpf13}
    \end{IEEEeqnarray}
	where 
	\begin{itemize}
		\item \eqref{eq:MaxAlphaLK-Add-inpf10} is from the chain rule of MI;
		\item \eqref{eq:MaxAlphaLK-Add-inpf11} and \eqref{eq:MaxAlphaLK-Add-inpf12} are from the facts that the mechanism has no feedback and is memoryless, respectively;
		\item from \cite[(2.122)]{IT_Cover}, we know that for a Markov chain $X-Y-Z$, conditioning reduces mutual information, i.e., $I(X;Y|Z)\leq I(X;Y)$ with equality if and only if $I(X;Z)=0$. Therefore, since for any $i\in[1,n]$ $\left(X_{i-1},\cdots,X_{1}\right)-X_{i}-Y_{i}$, the equality in \eqref{eq:MaxAlphaLK-Add-inpf6} holds if and only if the source is memoryless, i.e., $P_{\tilde{X}_i|\tilde{X}_{i-1},\cdots,\tilde{X}_{1}}=P_{\tilde{X}_i}$ for all $i\in[1,n]$.
	\end{itemize}
	
\end{proof}

\section{Proof of Theorem \ref{thm:PUT_fLeakKforDistvsHardDist}}\label{Proof:thm:PUT_fLeakKforDistvsHardDist}
\begin{proof}[\nopunct]
	Given $P_X$, the collection of stochastic matrices is denoted as $\mcl P_{Y|X}$. The feasible ball $B_D(x)$ around $x$ is defined in \eqref{eq:PUT_HardDist_CollectofFeasibleY}. For the distribution dependent PUT in \eqref{eq:PUT_fLeakKforDistvsHardDist}, we have
	\begin{IEEEeqnarray}{l l}
		&\text{PUT}_{\text{HD},\mcl L_f} (D)\nonumber\\
		\label{eq:PUT_fLeakKforDistvsHardDist1}
		=&\inf_{\substack{P_{Y|X}\in \mcl P_{Y|X}\\:d(X,Y)\le D}}\,\, \inf_{Q_Y}D_f(P_{Y|X}P_X\|P_X\times Q_Y)\\
		\label{eq:PUT_fLeakKforDistvsHardDist2}
		=&\inf_{Q_Y} \,\,\inf_{\substack{P_{Y|X}\in \mcl P_{Y|X}\\:d(X,Y)\le D}}\, \sum_{x\in\mcl X} P_{X}(x)D_f(P_{Y|X=x}\|Q_Y)\\
		\label{eq:PUT_fLeakKforDistvsHardDist3}
		=&\inf_{Q_Y} \sum_{x\in\mcl X} P_X(x)  \inf_{\substack{P_{Y|X=x}\\Y\in B_D(x)}}\sum_{y\in \mcl Y} Q_{Y}(y)f \left(\frac{ P_{Y|X}(y|x)}{Q_Y(y)}\right)\\
		=&\inf_{Q_Y} \sum_{x\in\mcl X} P_X(x) \inf_{\substack{P_{Y|X=x}\\Y\in B_D(x)}} \left(\sum\limits_{y\in B_D(x)^c} Q_{Y}(y)f \left(\frac{P_{Y|X}(y|x)}{Q_Y(y)}\right)\right.\nonumber\\
		&\left.+\frac{Q_{Y}(B_D(x))}{Q_{Y}(B_D(x))}\sum\limits_{y\in B_D(x)} Q_{Y}(y)f \left(\frac{ P_{Y|X}(y|x)}{ Q_Y(y)}\right)\right) \\
		=&\inf_{Q_Y} \sum_{x\in\mcl X} P_X(x) \inf_{\substack{P_{Y|X=x}\\Y\in B_D(x)}} \Bigg(\sum\limits_{y\in B_D(x)^c} Q_{Y}(y)f \left(0\right)\nonumber\\
		&+Q_{Y}(B_D(x))\hspace{-6pt}\sum\limits_{y\in B_D(x)} \frac{Q_{Y}(y)}{Q_{Y}(B_D(x))}f \left(\frac{ P_{Y|X}(y|x)}{ Q_Y(y)}\right)\Bigg) \label{eq:PUT_fLeakKforDistvsHardDist3-1}\\		
		\geq&\inf_{Q_Y} \sum_{x\in\mcl X} P_X(x) \inf_{\substack{P_{Y|X=x}\\Y\in B_D(x)}}\Bigg(Q_Y\left(B_D(x)^c\right) f(0)\nonumber\\
		&+Q_{Y}(B_D(x))f\bigg(\frac{1}{Q_Y(B_D(x))}\bigg)\Bigg)\label{eq:PUT_fLeakKforDistvsHardDist4}\\
		=& 
		f(0)\hspace{-2pt}+\hspace{-2pt}\inf\limits_{Q_Y} \hspace{-2pt}\sum_{x\in \mcl X}\hspace{-4pt}P(x) Q_{Y}(B_D(x))
		\label{eq:PUT_fLeakKforDistvsHardDist5}
		\bigg(\hspace{-3pt}f\Big(\hspace{-2pt}\frac{1}{Q_Y(B_D(x))}\hspace{-2pt}\Big)\hspace{-3pt}-\hspace{-2pt}f(0)\hspace{-4pt}\bigg)\hspace{-2pt},\,\qquad
	\end{IEEEeqnarray}
	where \begin{itemize}
		\item \eqref{eq:PUT_fLeakKforDistvsHardDist2} follows from the fact that $D_f(P_{Y|X}P_X\|P_X\times Q_Y)$ is convex in $(P_{Y|X},Q_Y)$ for fixed $P_X$,
		\item \eqref{eq:PUT_fLeakKforDistvsHardDist3-1} is directly from the hard distortion constraint $d(X;Y)\leq 0$ such that for any $y\notin B_D(x)$  $P_{Y|X}(y|x)=0$, and therefore, $\sum_{y\in B_D(x)}P_{Y|X}(y|x)=1$,
		\item \eqref{eq:PUT_fLeakKforDistvsHardDist4} is from the Jensen's inequality such that
		\begin{IEEEeqnarray}{ l l}
			&\sum\limits_{y\in B_D(x)} \frac{Q_{Y}(y)}{Q_{Y}(B_D(x))}f \left(\frac{ P_{Y|X}(y|x)}{ Q_Y(y)}\right)\nonumber\\
		    \geq &f \left(\sum\limits_{y\in B_D(x)} \frac{Q_{Y}(y)}{Q_{Y}(B_D(x))}\frac{ P_{Y|X}(y|x)}{ Q_Y(y)}\right)\\
			=  &f \left( \frac{\sum\limits_{y\in B_D(x)}P_{Y|X}(y|x) }{Q_{Y}(B_D(x))}\right)=f \left( \frac{1 }{Q_{Y}(B_D(x))} \hspace{-2pt}\right),\qquad
		\end{IEEEeqnarray}	
		with equality if and only if there is a mechanism $P_{Y|X}$ satisfying
		\begin{align}\label{eq:OptMechvsQ_Y-InProof}
		\frac{P_{Y|X}(y|x)}{Q_Y(y)}=\frac{\mathbf{1}(y\in B_D(x))}{Q_Y(B_D(x))}.
		\end{align}
	\end{itemize}	
    Note that $f:\mathbb{R}_+\to \mathbb{R}$ is a convex function, such that the function $tf(\frac{1}{t})$ is convex in $t\in \mathbb{R}_+$. Therefore, the objective function in \eqref{eq:PUT_fLeakKforDistvsHardDist5} is convex in $Q_Y$. Furthermore, in \eqref{eq:PUT_fLeakKforDistvsHardDist5} the feasible region of $Q_Y$ is the probability distribution simplex over the set $\{B_D(x), x\in \mcl X\}$. For finite supports $\mcl X$ and $\mcl Y$ of $X$ and $Y$, respectively, the set $\{B_D(x), x\in \mcl X\}$ is a compact, and therefore, the infimum in \eqref{eq:PUT_fLeakKforDistvsHardDist5} is achievable.
\end{proof}

\section{Proof of Theorem \ref{Thm:PUT_fLeak_HardDist}}\label{Proof:Thm:PUT_fLeak_HardDist}
\begin{proof}[\nopunct]
Given $P_X$, the collection of stochastic matrices is denoted as $\mcl P_{Y|X}$. The feasible ball $B_D(x)$ around $x$ is defined in \eqref{eq:PUT_HardDist_CollectofFeasibleY}. For the distribution independent PUT in \eqref{eq:PUT_fDivergenceLeak_HardDistortion}, we have
	\begin{IEEEeqnarray}{l l}
		&\text{PUT}_{\text{HD},\mcl L_f^{\text{max}}} (D)\nonumber\\
		=&\inf_{\substack{P_{Y|X}\in \mcl P_{Y|X}\\:d(X,Y)\le D}}\,\sup_{P_{\tilde{X}}}\,\inf_{Q_Y}\, D_f(P_{\tilde{X}}P_{Y|X}\|P_{\tilde{X}}\times Q_Y)\label{eq:put1}\\
		=& \inf_{Q_Y}\,\sup_{P_{\tilde{X}}}\,\inf_{\substack{P_{Y|X}\in \mcl P_{Y|X}\\:d(X,Y)\le D}}\,D_f(P_{\tilde{X}}P_{Y|X}\|P_{\tilde{X}}\times Q_Y)\label{eq:put2}\\
		=&\inf_{Q_Y}\,\sup_{P_{\tilde{X}}}\,\inf_{\substack{P_{Y|X}\in \mcl P_{Y|X}\\:d(X,Y)\le D}}\,\sum_{x\in \mcl X}  P_{\tilde{X}}(x)D_f(P_{Y|X=x}\|Q_Y)\label{eq:put3}\\
		= &\inf_{Q_Y}\, \sup_{P_{\tilde{X}}}\hspace{-2pt}\sum_{x\in \mcl X} \hspace{-2pt} P_{\tilde{X}}(x)\hspace{-5pt}\inf\limits_{\substack{P_{Y|X=x}\\Y\in B_D(x)}}\hspace{-2pt} \sum_{y\in \mcl Y}\hspace{-2pt}Q_Y(y) f\hspace{-3pt}\left(\hspace{-2pt}\frac{P_{Y|X}(y|x)}{Q_Y(y)}\hspace{-2pt}\right)\label{eq:put5}\\
		=&\inf_{Q_Y}\,\sup_{P_{\tilde{X}}}\, \sum_{x\in \mcl X}  P_{\tilde{X}}(x)\, \inf\limits_{\substack{P_{Y|X=x}\\Y\in B_D(x)}}\Bigg(\sum_{y\in B_D(x)} Q_Y(y)\nonumber\\
		&\quad \cdot f\left(\frac{P_{Y|X}(y|x)}{Q_Y(y)}\right)+\sum_{y\in B_D(x)^c} Q_Y(y) f(0)\Bigg)\label{eq:put6}\\
		=&\inf_{Q_Y}\,\sup_{P_{\tilde{X}}}\, \sum_{x\in \mcl X}  P_{\tilde{X}}(x)\, \inf\limits_{\substack{P_{Y|X=x}\\Y\in B_D(x)}}\,
			\Bigg( Q_Y(B_D(x))\sum_{y\in B_D(x)}\nonumber\\ & \frac{Q_Y(y)}{Q_{Y}(B_D(x))}f\hspace{-2pt}\left(\hspace{-2pt}\frac{P_{Y|X}(y|x)}{Q_Y(y)}\hspace{-2pt}\right)+Q_Y(B_D(x)^c)f(0)\Bigg)\label{eq:put7}\qquad \\
		\geq&\inf_{Q_Y}\,\sup_{P_{\tilde{X}}}\, \sum_{x\in \mcl X}  P_{\tilde{X}}(x)\, \inf\limits_{\substack{P_{Y|X=x}\\Y\in B_D(x)}}
			\bigg( Q_Y(B_D(x)) \nonumber\\
			&\quad \cdot  f\left(\frac{1}{Q_Y(B_D(x))}\right)+Q_Y(B_D(x)^c)f(0)\bigg)\label{eq:put8}\\
		=&\inf_{Q_Y}\,\sup_{P_{\tilde{X}}}\, \sum_{x\in \mcl X}  P_{\tilde{X}}(x)\,
		\bigg( Q_Y(B_D(x))f\left(\frac{1}{Q_Y(B_D(x))}\right)\nonumber\\
		&\quad +\big(1-Q_Y(B_D(x))\big)f(0)\bigg)\\
		=&\inf_{Q_Y}\,\sup_{P_{\tilde{X}}}\,  \sum_{x\in \mcl X}  P_{\tilde{X}}(x)\, g\big(Q_Y(B_D(x))\big)\label{eq:put9}\\
		=&\inf_{Q_Y}\,\sup_x\, g\big(Q_Y(B_D(x))\big)\label{eq:put9+}
	\end{IEEEeqnarray}	
	where \begin{itemize}
		\item \eqref{eq:put2} and \eqref{eq:put5} follow from the fact that $D_f(P_{\tilde{X}}P_{Y|X}\|P_{\tilde{X}}\times Q_Y)$ is linear in $P_{\tilde{X}}$ for fixed $(P_{Y|X},Q_Y)$ and convex in $(P_{Y|X},Q_Y)$ for fixed $P_{\tilde{X}}$,
		\item \eqref{eq:put8} follows from the convexity of $f$ and Jensen's inequality. The equality holds if and only if there exists a mechanism $P_{Y|X}$ satisfying \eqref{eq:OptMechvsQ_Y-InProof}.
		\item \eqref{eq:put9} results from $q\triangleq Q_Y(B_D(x))$ and 
		\begin{equation}\label{eq:put10}
		g(q)\triangleq qf(q^{-1})+(1-q)f(0).
		\end{equation}
	\end{itemize}
	Due to the convexity of $f$, we have $f(q^{-1})-f(0)\le f'(q^{-1})\left(q^{-1} -0\right)$,
	from which, the derivative $g'(q)=f(q^{-1})-q^{-1} f'(q^{-1})-f(0)\leq 0$. Therefore, the function $g$ in \eqref{eq:put10} is non-increasing, such that \eqref{eq:put9+} is simplified as $g(q^\star)$, where $q^{\star}$ is given by 
	\begin{equation}\label{eq:q_star_def-inPf}
	q^\star\triangleq \sup_{Q_Y}\, \inf_x\, Q_Y(B_D(x)).
	\end{equation}
	 Note that in \eqref{eq:q_star_def-inPf}, the feasible region of $Q_Y$ is the probability distribution simplex over the set $\{B_D(x), x\in \mcl X\}$. For finite supports $\mcl X$ and $\mcl Y$ of $X$ and $Y$, respectively, the set $\{B_D(x), x\in \mcl X\}$ is a compact, and therefore, the supremum in \eqref{eq:q_star_def-inPf} is achievable.

\end{proof}

\section{Proof of Theorem \ref{Thm:PUT_AlphaLeak_HardDist}}\label{proof:Thm:PUT_AlphaLeak_HardDist}
\begin{proof}[\nopunct]
	From Thm.~\ref{Thm:DefEquialentExpression_alphaleakage}, we know that for $\alpha\geq 1$, $\alpha$-leakage $\mcl L_{\alpha}(S;Y)$ equals to Arimoto MI $I^{\text{A}}_\alpha(S;Y)$. Since $I^{\text{A}}_\alpha(S;Y)=H_{\alpha}(S)-H^{\text{A}}_\alpha(S|Y)$ and $H_{\alpha}(S)$ is independent of $P_{Y|S,X}$, to minimize $I^{\text{A}}_\alpha(S;Y)$ with respect to $P_{Y|S,X}$ can be simplified to maximize $H^{\text{A}}_\alpha(S|Y)$. In addition, for $\alpha> 1$, the function $g: t\to \frac{\alpha}{1-\alpha}\log t$ is a monotonically non-increase function in $t>0$. Therefore, the problem in \eqref{eq:PUT_AlphaLeak_HD} can be simplified to 
	\begin{align}
		\inf_{\substack{P_{Y|SX}\\:d(X,Y)\le D}}\, \sum_{y\in \mcl Y}\Big(\sum_{s\in \mcl S} P(s,y)^{\alpha}\Big)^{\frac{1}{\alpha}}\label{eq:HDvsAlphaLK_InPf0}.
	\end{align}
	The hard distortion on $X$ and $Y$ in \eqref{eq:PUT_AlphaLeak_HD} determines a collection of feasible $x$ and therefore $s$ for each $y$. We define the two collections for each $y\in \mcl Y$ as
	\begin{align}
		\mcl X_D(y)&\triangleq \{x\in \mcl X:d(x,y)\leq D\},\\
		\mcl S_D(y)&\triangleq\{s\in \mcl S:\exists\, x \in \mcl X_D(y), P_{SX}(sx)>0\}.
	\end{align}
	Note that both sets defined above are independent of the privacy mechanism $P_{Y|S,X}$. 
    
    For $\alpha\in (1,\infty)$, we have
\begin{IEEEeqnarray}{l l}
	& \inf_{\substack{P_{Y|SX}\\:d(X,Y)\le D}}\, \sum_{y}\Big(\sum_s P(s,y)^{\alpha}\Big)^{\frac{1}{\alpha}}\nonumber\\
	=&\inf_{\substack{P_{Y|SX}\\:d(X,Y)\le D}} \sum\limits_{y\in \mcl Y}\hspace{-3pt}\left( \sum\limits_{s\in \mcl S}\hspace{-3pt}\Big(\sum\limits_{x\in \mcl X} P(y|s,x)P(s,x)\hspace{-2pt}\Big)^{\alpha}\hspace{-2pt} \right)^{\frac{1}{\alpha}} \label{eq:HDvsAlphaLK_InPf1}\\
	=&\inf_{\substack{P_{Y|S,X}}} \sum\limits_{y}\Bigg( \sum\limits_{\mcl S_D(y)}\bigg(\sum\limits_{\mcl X_D(y)} P(s,x,y)\bigg)^{\alpha}\hspace{-5pt}+\hspace{-5pt}\sum_{\substack{x\notin \mcl X_D(y)\\s\notin \mcl S_D(y)}} 0  \Bigg)^{\frac{1}{\alpha}}\label{eq:HDvsAlphaLK_InPf2}\\
	=&\inf\limits_{\substack{P_{Y|S,X}}} \sum\limits_{y}\Big(\hspace{-3pt}\sum\limits_{s'\in\mcl S_D(y) }P(s')^{\alpha}\Big)^{\frac{1}{\alpha}}\Bigg( \sum\limits_{s\in\mcl S_D(y)}\nonumber\\
	&\quad \frac{P(s)^{\alpha}}{\sum\limits_{s'\in \mcl S_D(y) }P(s')^{\alpha}}\Big(\sum\limits_{\mcl X_D(y)} P(x,y|s) \Big)^{\alpha}\Bigg)^{\frac{1}{\alpha}}\label{eq:HDvsAlphaLK_InPf3}\\
	\geq & \inf\limits_{\substack{P_{Y|S,X}}}  \sum\limits_{y}\Big(\hspace{-3pt}\sum\limits_{s'\in\mcl S_D(y) }P(s')^{\alpha}\Big)^{\frac{1}{\alpha}}\Bigg(\sum\limits_{s\in\mcl S_D(y)}\nonumber\\
	&\quad \frac{P(s)^{\alpha}}{\sum\limits_{s'\in\mcl S_D(y) }P(s')^{\alpha}}\Big(\sum\limits_{\mcl X_D(y)} P(x,y|s) \Big)\Bigg)  \label{eq:HDvsAlphaLK_InPf4}\\
	=&\inf\limits_{\substack{P_{Y|S,X}}}  \sum\limits_{\substack{y,\mcl S_D(y)\\\mcl X_D(y)}}\hspace{-2pt}\Big(\hspace{-3pt}\sum\limits_{s'\in\mcl S_D(y) }P(s')^{\alpha}\Big)^{\frac{1}{\alpha}-1}\hspace{-2pt}P(s)^{\alpha}P(x,y|s)         \label{eq:HDvsAlphaLK_InPf5}\\
	= &\inf\limits_{\substack{P_{Y|S,X}}}  \sum\limits_{\substack{s,x\\B_D(x)}}\Big(\hspace{-3pt}\sum\limits_{s'\in\mcl S_D(y) }P(s')^{\alpha}\Big)^{\frac{1}{\alpha}-1}P(s)^{\alpha}P(x,y|s)\label{eq:HDvsAlphaLK_InPf6}\\
	\geq &\inf\limits_{\substack{P_{Y|S,X}}}\sum\limits_{\substack{s,x}}P(s)^{\alpha}P(x|s) \min\limits_{y\in B_D(x)}\Big(\hspace{-5pt}\sum\limits_{s'\in\mcl S_D(y) }\hspace{-5pt}P(s')^{\alpha}\Big)^{\hspace{-3pt}\frac{1-\alpha}{\alpha}}\label{eq:HDvsAlphaLK_InPf7}\\
	= & \sum\limits_{\substack{s,x}}P(s)^{\alpha}P(x|s) \Big(\max\limits_{y\in B_D(x)}\sum\limits_{s'\in\mcl S_D(y) }P(s')^{\alpha}\Big)^{\frac{1}{\alpha}-1}\label{eq:HDvsAlphaLK_InPf8},
\end{IEEEeqnarray}
where
\vspace{-10pt }\begin{itemize}
		\item \eqref{eq:HDvsAlphaLK_InPf4} is directly from the concavity of the function $g_1: t\to t^{\frac{1}{\alpha}}$ ($\alpha>1$) and Jensen's inequality. The equality holds if and only if the optimal $P_{Y|S,X}^{\star}$ achieving the infimum satisfies that for all $s\in \mcl S_D(y)$,		
		\begin{IEEEeqnarray}{l l}\label{eq:HDvsAlphaLK_InPf-Mech1}
		\hspace{-20pt}	P^{\star}(y|s)=\hspace{-7pt}\sum\limits_{x\in\mcl X_D(y)}\hspace{-5pt}P^{\star}(y|sx)P(x|s)=\hspace{-2pt}\frac{P_Y^{\star}(y)}{\sum\limits_{s'\in\mcl S_D(y)}\hspace{-3pt}P_S(s')}.
		\end{IEEEeqnarray}
		\vspace{-2pt }where $P_Y^{\star}$ is the probability distribution of $Y$ derived from $P^{\star}_{Y|S,X}$.
		\item in \eqref{eq:HDvsAlphaLK_InPf6}, $B_D(x)$ is the feasible ball defined in \eqref{eq:PUT_HardDist_CollectofFeasibleY}.
		\item the equality in \eqref{eq:HDvsAlphaLK_InPf7} holds if and only if for any $(s,x)$, all $y$ with $P^{\star}(y|s,x)>0$ lead to the same $\sum_{s'\in\mcl S_D(y)}P(s')$.
		\item the equality in \eqref{eq:HDvsAlphaLK_InPf8} is from the fact that the function $g:t\to t^{\frac{1}{\alpha}-1}$ is monotonically non-increasing in $t>0$ for $\alpha\geq 1$.
	\end{itemize}
\vspace{-8pt } Similarly, for $\alpha=\infty$, we have
\begin{IEEEeqnarray}{l l}
	 &\inf_{\substack{P_{Y|SX}\\:d(X,Y)\le D}}\, \sum_{y}	P_Y(y)\max_s P_{S|Y}(s|y)\nonumber\\
	 =&\inf_{\substack{P_{Y|S,X}}} \sum\limits_{y}P_Y(y) \max\limits_{\mcl S_D(y)}\bigg(\sum\limits_{\mcl X_D(y)} P_{S,X|Y}(s,x|y)\bigg)\label{eq:HDvsAlphaInfLK_InPf2}\\
	\geq& \inf\limits_{\substack{P_{Y|S,X}}}\sum\limits_{y}\hspace{-3pt}P(y)\hspace{-3pt}\Bigg(\hspace{-3pt}\sum\limits_{\mcl S_D(y)}\hspace{-3pt}\frac{P(s) }{\sum\limits_{s'\in\mcl S_D(y)}P(s')}\hspace{-4pt}\sum\limits_{\mcl X_D(y)}\hspace{-4pt}P(s,x|y)\hspace{-3pt}\Bigg)\label{eq:HDvsAlphaInfLK_InPf3}\qquad \\
	= &\inf\limits_{\substack{P_{Y|S,X}}}\sum\limits_{s,x}\sum\limits_{B_D(x)} \frac{P(s)}{\sum\limits_{s'\in\mcl S_D(y)}P(s')}P(s,x,y)\label{eq:HDvsAlphaInfLK_InPf6}\\
	\geq &\inf\limits_{\substack{P_{Y|S,X}}}\sum\limits_{s,x}\sum\limits_{B_D(x)}\hspace{-6pt} P(s,x,y)\hspace{-3pt} \min\limits_{y\in B_D(x)}\hspace{-3pt} \frac{P(s)}{\sum\limits_{s'\in\mcl S_D(y)}P(s')}\label{eq:HDvsAlphaInfLK_InPf7}\\
	=& \sum_{s,x}P(s,x)P(s)\left(\max_{\substack{y\in B_D(x)}}\sum\limits_{s'\in\mcl S_D(y)}P(s')\right)^{-1}\label{eq:HDvsAlphaInfLK_InPf8}.
\end{IEEEeqnarray}
Note that the sufficient and necessary conditions for the equalities in \eqref{eq:HDvsAlphaInfLK_InPf3} and \eqref{eq:HDvsAlphaInfLK_InPf7} hold are the same as that for \eqref{eq:HDvsAlphaLK_InPf4} and \eqref{eq:HDvsAlphaLK_InPf7}, respectively.

	For $\alpha=1$, $\mcl L_{\alpha}(S\to Y)=I^{\text{A}}(S;Y)=I(S;Y)$, such that
	\begin{IEEEeqnarray}{r l}
		\text{PUT}_{\text{HD},\mcl L_\alpha} (D)
		=&	\inf_{\substack{P_{Y|SX}\\:d(X,Y)\le D}}\, \sum_{s,y}P(s,y)\log \frac{P(s,y)}{P(s)P(y)}\label{eq:HDvsAlpha=1LK_InPf1}\\
		=&\inf_{\substack{P_{Y|S,X}}} \sum\limits_{y} \sum\limits_{\mcl S_D(y)}\bigg(\Big(\sum\limits_{\mcl X_D(y)} P(s,x,y)\Big)\nonumber\\
		&\quad \cdot \log \frac{\sum_{\mcl X_D(y)} P(s,x,y)}{P(s)P(y)}\bigg)\label{eq:HDvsAlpha=1LK_InPf2}\\
		\geq  &\inf_{\substack{P_{Y|S,X}}} \sum\limits_{y} \bigg(\Big(\sum\limits_{\mcl S_D(y)}\sum\limits_{\mcl X_D(y)} P(s,x,y)\Big)\nonumber\\
		&\quad \cdot\log \frac{\sum_{\mcl S_D(y)} \sum_{\mcl X_D(y)} P(s,x,y)}{\sum_{\mcl S_D(y)}P(s)P(y)}\bigg)\label{eq:HDvsAlpha=1LK_InPf3}\\
		= &\inf_{\substack{P_{Y|S,X}}}\hspace{-3pt}\sum\limits_{\substack{y,\mcl S_D(y)\\ \mcl X_D(y)}}\hspace{-5pt} P(s,x,y)\log \frac{1}{\sum\limits_{s'\in\mcl S_D(y)}\hspace{-4pt}P(s')}\label{eq:HDvsAlpha=1LK_InPf4}\hspace{20pt} \\
		\geq &\sum\limits_{s,x}P(s,x)\min\limits_{\substack{y\in B_D(x)}}\log \frac{1}{\sum\limits_{s'\in\mcl S_D(y)}P(s')} \label{eq:HDvsAlpha=1LK_InPf5}.
	\end{IEEEeqnarray}
	Note that the inequality in \eqref{eq:HDvsAlpha=1LK_InPf4} is from log-sum inequality in \cite[Thm. 2.7.1]{IT_Cover}, and the sufficient and necessary conditions for the equalities in \eqref{eq:HDvsAlpha=1LK_InPf4} and \eqref{eq:HDvsAlpha=1LK_InPf5} hold are the same as that for \eqref{eq:HDvsAlphaLK_InPf4} and \eqref{eq:HDvsAlphaLK_InPf7}, respectively.  
\end{proof}

\section{Proof of Theorem \ref{Thm:PUTMaxAlphaLk&HardDist-SeqType-Alpha>1}}\label{Proof:Thm:PUTMaxAlphaLk&HardDist-SeqType-Alpha>1}
\begin{proof}[\nopunct]
Define the distortion ball for the type-distance distortion in \eqref{eq:HDonTypes} as
	\begin{align}
	B_{m}(x^n)\triangleq \left\{y^n: |P_{x^n}(0)-P_{y^n}(0)|\leq \frac{m}{n}\right\}.
	\end{align}
	From Corollary \ref{Col:PUT_MaximalAlphaLeak_HardDist}, to find an optimal mechanism $P_{Y^n|X^n}^{\star}$, we need to find an output distribution $Q_{Y^n}^{\star}$ which optimizes \eqref{eq:q_star_def} with $x^n$ and $y^n$ in place of $x,y$.
	
	Note that for the hard distortion $|P_{x^n}(0)-P_{y^n}(0)|\leq \frac{m}{n}$, all datasets in a type class share the same group of feasible output datasets, and this feasible group can be represented by output type classes. Therefore, for any $x^n\in T(i)$ ($i\in[0,n]$), we rewrite $B_{m}(x^n)$ as
	\begin{align}
	B_{m}(x^n)=B_m(T(i))\triangleq\bigcup_{\substack{|i-j|\leq m\\j\in [0,n]}}T(j).
	\end{align} 
    We define a distribution $Q_{T}$ of type classes for outputs as
	\begin{align}\label{eq:PUT-TYPE-RelatDisonTypeandSeq-inProof}
	Q_{T}(T(j))\triangleq\sum\limits_{y^n\in T(j)}Q_{Y^n}(y^n), \text{ for } j\in[0,n],
	\end{align}
	such that
	\begin{equation}\label{eq:q_star_def_TypeExamp}
	q^{\star}=\sup_{Q_{T}}\, \min\limits_{i\in[0,n]}\, Q_{T}(B_m(T(i))).
	\end{equation}
	Note that we replace the infimum by minimum in \eqref{eq:q_star_def_TypeExamp} due to the fact that the infimum over finite integers equals to the minimum. 
	The optimal distribution $Q_{T}$ is determined by bounding $q*$ from above and below in \eqref{eq:q_star_def_TypeExamp}. 
	The upper bound is determined by restricting the optimization in \eqref{eq:q_star_def_TypeExamp} to a judicious choice of a small set of input types. The lower bound is a constructive scheme.
		
	We define an index set $\mcl I_T\subset [0,n]$ for types as 
		\begin{IEEEeqnarray}{l l}
			\label{eq:index_types}
			I_T\triangleq \left\{l+(2m+1)k: k\in\left[0,\ceil*{\frac{n+1}{2m+1}}-1\right]\right\} 
		\end{IEEEeqnarray}
	where $l=m$ if $\ceil{\frac{n+1}{2m+1}}\leq \frac{m+n+1}{2m+1}$, and otherwise, $l=n-\left(\ceil{\frac{n+1}{2m+1}}-1\right)(2m+1)$.
	From the expression of $\mcl I_T$ in \eqref{eq:index_types}, we observe that: (i) the difference between adjacent elements is $2m+1$; (ii) for the first and last elements,
	\begin{itemize}
	\item if $\ceil{\frac{n+1}{2m+1}}\leq \frac{m+n+1}{2m+1}$ holds, the first element is $m$ and the last element is
	\begin{IEEEeqnarray}{ l  l}
		 &m+(2m+1)\left(\ceil*{\frac{n+1}{2m+1}}-1\right)\nonumber\\
		 =&(2m+1)\ceil*{\frac{n+1}{2m+1}}-m-1\in[n-m,n],\quad 
	\end{IEEEeqnarray}	
	due to the inequalities $\frac{n+1}{2m+1}\leq \ceil{\frac{n+1}{2m+1}}\leq \frac{m+n+1}{2m+1} $;
	\item if $\ceil{\frac{n+1}{2m+1}}> \frac{m+n+1}{2m+1}$ holds, the last element is $n$ and the first element is
	\begin{IEEEeqnarray}{l l}
		&n-\left(\ceil*{\frac{n+1}{2m+1}}-1\right)(2m+1)\nonumber\\
		=& n+2m+1-\ceil*{\frac{n+1}{2m+1}}(2m+1)\in [0,m),\quad 
	\end{IEEEeqnarray}
    due to the inequalities $\frac{n+1}{2m+1}+1-\frac{1}{2m+1} \geq \ceil{\frac{n+1}{2m+1}}> \frac{m+n+1}{2m+1} $ for $n\in \mathbb{Z}_{++}$.
	\end{itemize} 
Therefore, it is not difficult to see that feasible balls of input type classes indexed by $I_T$ are a partition of the set of all type classes, i.e., 
	\begin{subequations}\label{eq:PartitionSets_inPf}
		\begin{align}
		&B_m(T(i_1))\cap B_m(T(i_2))=\emptyset\quad i_1,i_2\in \mcl I_T,\\
		&\bigcup_{j\in[0,n]} T(j)= \bigcup_{i\in\mcl I_T} B_m(T(i)).
		\end{align}
	\end{subequations}
	Therefore, the problem in \eqref{eq:q_star_def_TypeExamp} {is bounded from above} by
	\begin{IEEEeqnarray}{ l l }
		q^{\star}  &\leq  \sup_{Q_{T}}\, \min\limits_{i\in\mcl I_T }\, Q_{T}(B_m(T(i)))\\
			&\leq \sup_{Q_{T}}\, \frac{1}{|\mcl I_T|} \sum_{i\in\mcl I_T }\, Q_{T}(B_m(T(i)))\label{eq:Example1-InPf-1}\\
			&= \sup_{Q_{T}}\left(\ceil*{\frac{n+1}{2m+1}}\right)^{-1} \sum_{j\in[0,n]}Q_{T}(T(j))\label{eq:Example1-InPf-2}\\
			&=\left(\ceil*{\frac{n+1}{2m+1}}\right)^{-1},\label{eq:Example1-InPf-3}
	\end{IEEEeqnarray}
    where
    \begin{itemize}
    		\item the inequality in \eqref{eq:Example1-InPf-1} is from that the average probability of $B_m(T(i))$ over $i\in \mcl I_T$ is no less than the minimal probability of $B_m(T(i))$ for $i\in \mcl I_T$;
    		\item the equality in \eqref{eq:Example1-InPf-2} is from that the cardinality of $\mcl I$ defined in \eqref{eq:index_types} is $\ceil{\frac{n+1}{2m+1}}$; 
    		\item the equality in \eqref{eq:Example1-InPf-3} is from that for any distribution over types $T(j)$ with $j\in[0,n]$, the sum of $Q_{T}(T(j))$ over $j\in[0,n]$ is $1$.
    \end{itemize}     
	To bound $q^\star$ from below, we construct a distribution $Q'_{T}$ as% such that for every $j\in I_T$ and an unique $y^n\in T(j)$
	\begin{align}\label{eq:PUT-TYPE-OptDisOnType}
	Q_{T}'(T(j))=\begin{cases}
	\left(\ceil*{\frac{n+1}{2m+1}}\right)^{-1}& j\in I_T\\
	0& \textit{otherwise}.
	\end{cases}
	\end{align}
	By \eqref{eq:PartitionSets_inPf} for each $i\in[0,n]$, there is a \textit{unique} $j$ satisfying $|i-j|\leq m$. Therefore, we bound \eqref{eq:q_star_def_TypeExamp} by 
	\begin{IEEEeqnarray}{r l}
		q^{\star}\geq &\,\min_i \,Q_{T}'(B_m(T(i)))\\
		= &\,\min_i Q_{T}'\Big(\bigcup_{\substack{|i-j|\leq m\\j\in \mcl I_T}}T(j)\Big)\label{eq:Example1-InPf-4}\\
		= &\left(\ceil*{\frac{n+1}{2m+1}}\right)^{-1},\label{eq:Example1-InPf-5}%\,\inf_k Q_{T}'(T(I_T(k)))= 
	\end{IEEEeqnarray}  
    where the equality in \eqref{eq:Example1-InPf-5} holds because for any $i\in[0,n]$, there is only one $j\in \mcl I_T$ satisfying $|i-j|\leq m$ such that the union in \eqref{eq:Example1-InPf-4} has exactly one element in it.
    
	Therefore, $q^{\star}=\left(\ceil*{\frac{n+1}{2m+1}}\right)^{-1}$ and the $Q_{T}'$ defined in \eqref{eq:PUT-TYPE-OptDisOnType} achieve the optimum in \eqref{eq:q_star_def_TypeExamp}.
Thus, we can derive an optimal $Q^{\star}_{Y^n}$, which assigns the same non-zero probability to only one dataset of each type classes indexed by $I_T$, i.e., $Q^{\star}_{Y^n}(y^n)=q^{\star}$ for one $y^n\in T(j)$ for each $j\in I_T$. Therefore, from \eqref{eq:opt_mech} we have the corresponding optimal privacy mechanism, which maps all input datasets in one input type class to one feasible output dataset with probability $1$. 	
\end{proof}

\section{Proof of Theorem \ref{thm:OptPUT-MaxAlphaLeak-UvsHardHamingDist}}\label{proof:thm:OptPUT-MaxAlphaLeak-UvsHardHamingDist}

\begin{proof}[\nopunct]

	For the Hamming distortion function on datasets in \eqref{eq:HamingHD_data sets}, the feasible ball $B_{m}(x^n)$ of any dataset $x^n\in \mcl X^n$ is given by
	\begin{align}
	B_{m}(x^n)=\left\{y^n\in \mcl X^n: d_{\text{H}}(x^n,y^n)\leq \frac{m}{n} \right\}.
	\end{align}	
	For each $x^n\in \mcl X^n$, the number of datasets having different values at exactly $k>0$ different positions is ${n \choose k}\left(|\mcl X|-1\right)^k$.
	Therefore, the number of elements in its feasible ball $B_{m}(x^n)$ is 
	\begin{align}
	\label{eq:Exampe_PUT-HD-HamSeq_inPf}
	|B_{m}(x^n)|=\sum_{i=0}^{m}{n \choose i}\left(|\mcl X|-1\right)^i,
	\end{align} 
	Note that the cardinality $|B_{m}(x^n)|$ in \eqref{eq:Exampe_PUT-HD-HamSeq_inPf} of a feasible ball is independent of the input dataset. 
	We denote the cardinality as $N_{\text{ball}}$, i.e., $N_{\text{ball}}\triangleq |B_{m}(x^n)|$.
	Due to the symmetric property of the Hamming distortion on datasets in \eqref{eq:HamingHD_data sets}, i.e., for any two datasets $x_1^n,x_2^n\in \mcl X^n$, $x_1^n\in B_D(x_2)$ if and only if $x_2\in B_D(x_1)$, we know that each output dataset is in exactly $N_{\text{ball}}$ different feasible balls (the example in Fig. \ref{fig:PUT-MaxAlphaLeak-UvsHardHamingDist_Mechanism} may help to figure out the above relationships). Therefore,
	\begin{IEEEeqnarray}{l l}
		q^{\star}&=\sup_{Q_{Y^n}} \inf_{x^n\in \mcl X^n} Q_{Y^n}\left(B_{m}(x^n)\right)\\
		&\leq \sup_{Q_{Y^n}} \frac{1}{\left|\mcl X^n\right|}\sum_{x^n\in \mcl X^n} Q_{Y^n}\left(B_{m}(x^n)\right)
		\label{thm:PUT_MaxAlphaLkHardHaming_inproof1}\\
		&=\sup_{Q_{Y^n}} \frac{1}{\left|\mcl X^n\right|}\sum_{x^n\in \mcl X^n} \sum_{y^n\in B_{m}(x^n)}Q_{Y^n}\left(y^n\right)
		\label{thm:PUT_MaxAlphaLkHardHaming_inproof2}\\
		&=\sup_{Q_{Y^n}} \frac{1}{\left|\mcl X\right|^n}\sum_{\substack{x^n\in \mcl X^n\\y^n\in B_{m}(x^n)}}Q_{Y^n}\left(y^n\right)
		\label{thm:PUT_MaxAlphaLkHardHaming_inproof3}\\
		&=\sup_{Q_{Y^n}} \frac{1}{\left|\mcl X\right|^n}\sum_{y^n\in \mcl X^n} N_{\text{ball}}Q_{Y^n}\left(y^n\right)
		\label{thm:PUT_MaxAlphaLkHardHaming_inproof4}\\
		&=\frac{N_{\text{ball}}}{\left|\mcl X\right|^n}\label{thm:PUT_MaxAlphaLkHardHaming_inproof5}
	\end{IEEEeqnarray}
	where 
	\begin{itemize}
		\item the equality in \eqref{thm:PUT_MaxAlphaLkHardHaming_inproof1} holds if and only if for an arbitrary pair of datasets $x_1^n, x_2^n$, there is 
		\begin{align}
		Q_{Y^n}\left(B_D(x_1^n)\right)=Q_{Y^n}\left(B_D(x_2^n)\right),
		\end{align}
		which can be satisfied by a uniform distribution over $\mcl X^n$, i.e., $Q_{Y^n}^{\star}=\frac{1}{\left|\mcl X\right|^n}$.
		\item the equality in \eqref{thm:PUT_MaxAlphaLkHardHaming_inproof4} holds because, for each $y^n$, the number of sequences $x^n$ where $d_H(x^n,y^n)\le \frac{m}{n}$ is exactly $N_{ball}$.	
	\end{itemize}
\end{proof}

	\section*{Acknowledgment}

	The authors would like to thank Dr. Mario Alberto Diaz Torres and Prof. Vincent Y. F. Tan  for many useful discussions.
	%from National University of Singapore
	
	\ifCLASSOPTIONcaptionsoff
	\newpage
	\fi

	
	
\bibliographystyle{IEEEtran}
\bibliography{JL_References-J}
	

	
\begin{IEEEbiography}[{\includegraphics[width=1in,height=1.25in,clip,keepaspectratio]{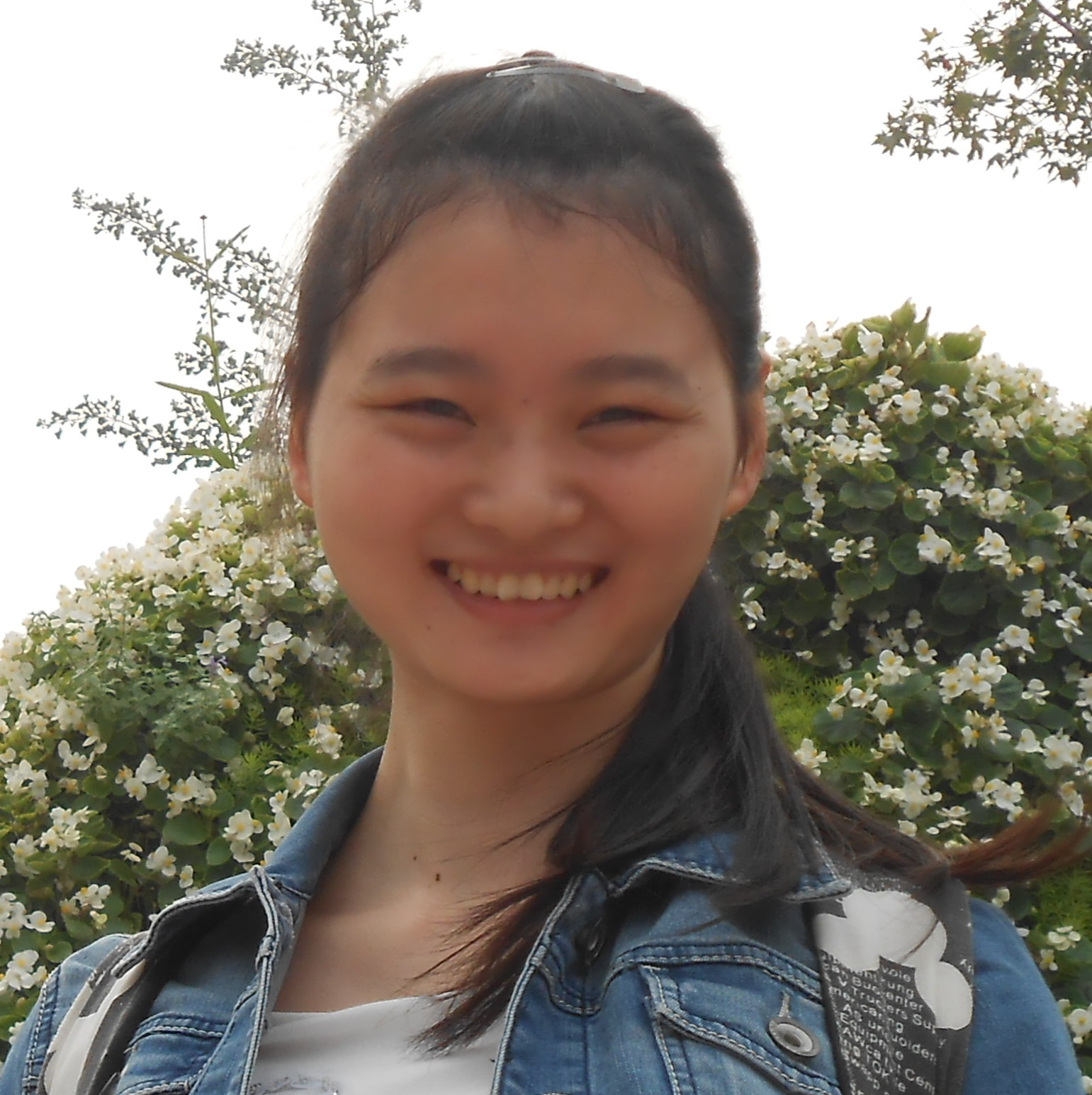}}]{Jiachun Liao} (S'16) received the B.Eng. in communication engineering and M.Eng. degrees in communication and information system from Beijing Jiaotong University, in 2012 and 2015, respectively. She is currently pursuing the Ph.D. degree in the School of Electrical, Computer, and Energy Engineering at Arizona State University. Her research interests includes wireless communications, information privacy and fairness in machine learning. 
\end{IEEEbiography}

\vskip 0pt plus -1fil

\begin{IEEEbiography}[{\includegraphics[width=1in,height=1.25in,clip,keepaspectratio]{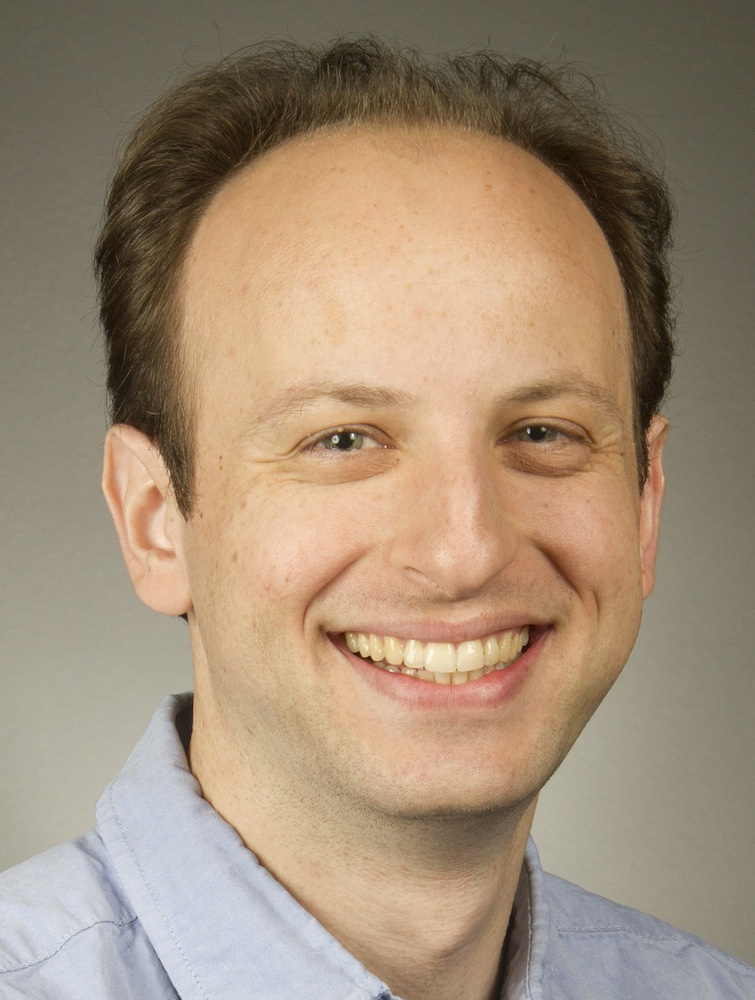}}]{Oliver Kosut} (S'06--M'10) received B.S. degrees in electrical engineering and mathematics from the Massachusetts Institute of Technology, Cambridge, MA, USA in 2004 and the Ph.D. degree in electrical and computer engineering from Cornell University, Ithaca, NY, USA in 2010.
	
Since 2012, he has been a faculty member in the School of Electrical, Computer and Energy Engineering at Arizona State University, Tempe, AZ, USA, where he is an Associate Professor. Previously, he was a Postdoctoral Research Associate in the Laboratory for Information and Decision Systems at MIT from 2010 to 2012. His research interests include information theory, cybersecurity, and power systems. Prof.~Kosut received the NSF CAREER award in 2015.
\end{IEEEbiography}

\vskip 0pt plus -1fil

\begin{IEEEbiography}[{\includegraphics[width=1in,height=1.25in,clip,keepaspectratio]{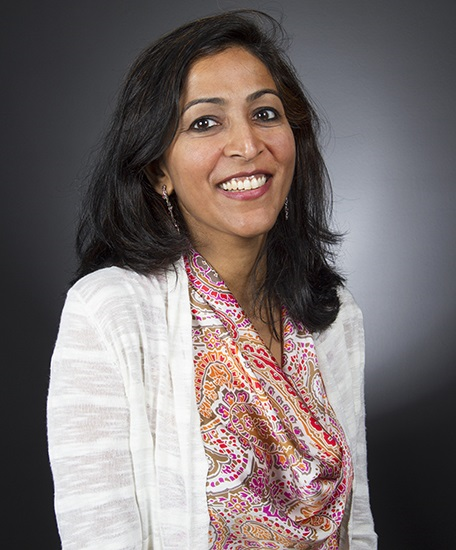}}]{Lalitha Sankar} (S'02-M'07-SM'15) received the B.Tech. degree from the Indian Institute of Technology, Bombay, the M.S. degree from the University of Maryland, and the Ph.D. degree from Rutgers University. She is currently an Associate Professor in the School of Electrical, Computer, and Energy Engineering at Arizona State University. Prior to this, she was an Associate Research Scholar at Princeton University and a recipient of a three year Science and Technology Teaching postdoctoral fellowship from the Council on Science and Technology at Princeton University. Sankar has also worked as a Senior Member of Technical Staff at AT\&T Shannon Labs and Polaroid Engineering R\&D Labs.
	
Her research interests include applying information sciences to study data privacy as well as cybersecurity and resilience in critical infrastructure networks. For her doctoral work, she received the 2007\--2008 Electrical Engineering Academic Achievement Award from Rutgers University. She received the IEEE Globecom 2011 Best Paper Award for her work on privacy of side-information in multi-user data systems, and the National Science Foundation CAREER Award in 2014.
\end{IEEEbiography}

\vskip 0pt plus -1fil

\begin{IEEEbiography}[{\includegraphics[width=1in,height=1.25in,clip,keepaspectratio]{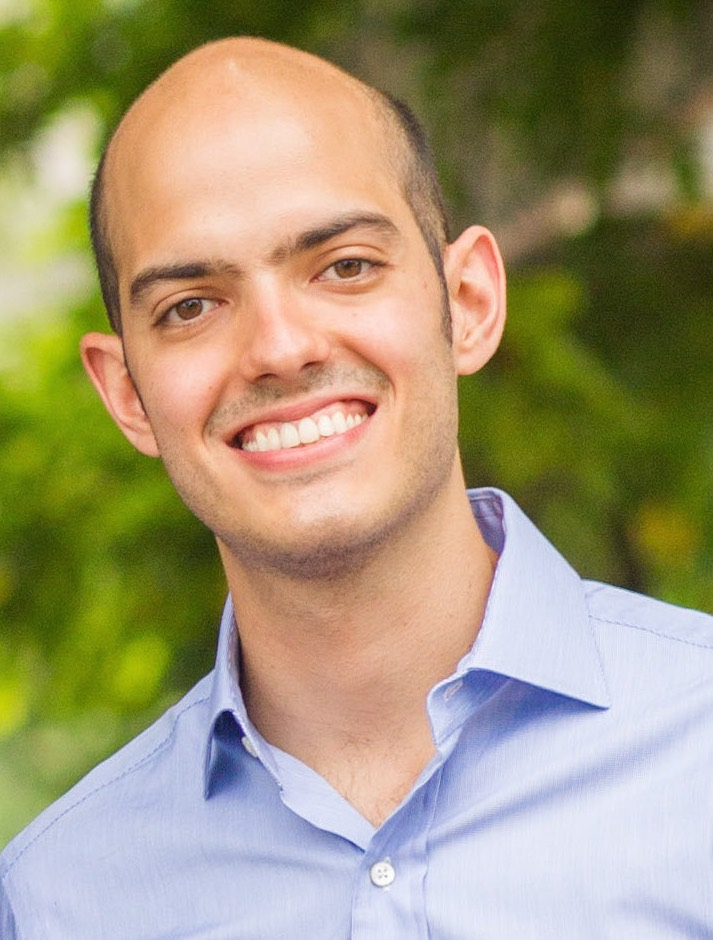}}]{Flavio du Pin Calmon} is an Assistant Professor of Electrical Engineering at Harvard's John A. Paulson School of Engineering and Applied Sciences. Before joining Harvard, he was the inaugural data science for social good post-doctoral fellow at  IBM Research in Yorktown Heights, New York. He received his Ph.D. in Electrical Engineering and Computer Science at MIT. His main research interests are information theory, inference, and statistics, with applications to fairness, privacy, machine learning, and communications engineering. Prof. Calmon has received the NSF CAREER Award in 2019, the Google Research Faculty Award, the IBM Open Collaborative Research Award, and Harvard's Lemann Brazil Research Fund Award.
\end{IEEEbiography}

\end{document}